\documentclass{article}
\newtheorem{definition}{Definition}
\newtheorem{example}[definition]{Example}

\newtheorem{theorem}[definition]{Theorem}

\newtheorem{lemma}[definition]{Lemma}

\newtheorem{proposition}[definition]{Proposition}

\newtheorem{remark}[definition]{Remark}

\newtheorem{conjecture}[definition]{Conjecture}

\newtheorem{corollary}[definition]{Corollary}
\makeatletter
\renewcommand{\@begintheorem}[2]{ 
\trivlist\item[\hskip\labelsep{\bf #1\ #2}]}
\renewcommand{\@opargbegintheorem}[3]{
\trivlist\item[\hskip \labelsep{\bf #1\ #2\ (#3)}]}
\makeatother
\newtheorem{proof}{Proof}
\usepackage[T1]{fontenc} 
\usepackage{ae,aecompl}
\usepackage[latin1]{inputenc} 
\usepackage{stmaryrd}
\usepackage[lutzsyntax]{virginialake}\vlsmallbrackets\aftrianglefalse
\usepackage{amsmath,amssymb,latexsym}
\usepackage{paralist}
\usepackage{calrsfs}
\usepackage{hyperref}
\usepackage{graphicx}
\usepackage{txfonts} 
\usepackage{color} 
\usepackage{subfigure}

\usepackage{textcomp}
\usepackage{xspace}
\usepackage{rotating}
\usepackage{setspace}
\usepackage{xypic}
\newcommand{\pincLTSJud}[3]{\xymatrix@1{#1\, \ar[r]^-{\ #3\ } & \,#2}} 
\newcommand{\pincLTSEXTJud}[5]
{\xymatrix@1{#1\, \ar[r#4]^-{\ #3\ } &#5 \,#2}} 
\newcommand{\pincRSJud}[3]{\xymatrix@1{#1\, \ar@{>>}[r]^-{\ #3\ } &\, #2}} %
\newcommand{\pincRSEXTJud}[5]
{\xymatrix@1{#1\, \ar@{>>}[r#4
]^{\ #3\ } &#5
\, #2}} %
\newcommand{\cutelimination}{cut-elimination\xspace}

\newcommand{\Set}[1]{ \{ #1 \}}
\newcommand{\Size}[1]{|#1|}
\newcommand{\smalltitle}[1]{\small #1}
\newcommand{\ST}{such that}

\newcommand{\IFF}{iff}

\newcommand{\wrt}{w.r.t.\xspace}
\newcommand{\dfn}[1]{\emph{#1}}
\newcommand{\grammareq}{::=\ }
\newcommand{\DI}{DI\xspace}

\newcommand{\NEL}{\mathsf{NEL}\xspace}
\newcommand{\BV}{\mathsf{BV}\xspace}
\newcommand{\SBV}{\mathsf{SBV}\xspace}

\newcommand{\SBVsub}{\mathsf{B}\xspace}
\newcommand{\BVT}{\BV\mathsf{Q}\xspace}

\newcommand{\CCSR}{\mathsf{CC_{sp}}\xspace} 
\newcommand{\SBVT}{\SBV\mathsf{Q}\xspace}

\newcommand{\CCS}{\mathsf{CCS}\xspace}

\newcommand{\MLL}{\mathsf{MLL}\xspace}

\definecolor{Blue}{rgb}{0,0,1}

\definecolor{Red}{rgb}{1,0.2,0}

\newtheorem{lthm}{Theorem}[section]{\bfseries}{\itshape}
\newtheorem{lprop}[lthm]{Proposition}{\bfseries}{\itshape}
\newtheorem{llem}[lthm]{Lemma}{\bfseries}{\itshape}
\newtheorem{lcor}[lthm]{Corollary}{\bfseries}{\itshape}
\newtheorem{lcon}[lthm]{Conjecture}{\bfseries}{\itshape}
\newtheorem{lfac}[lthm]{Fact}{\bfseries}{\itshape}
\newtheorem{ldef}[lthm]{Definition}{\bfseries}{}
\newtheorem{lrem}[lthm]{Remark}{\bfseries}{}
\newtheorem{lexa}[lthm]{Example}{\bfseries}{}

\catcode`\@=11
\def\@thmcountersep{.}
\def\@thmcounterend{\;}
\catcode`@=12

\newcommand{\mapLcToDi}[2]
           {\llparenthesis\,\! #1\, \rrparenthesis_{#2}} 

\newcommand{\llcxlamterm}{$\lambda$-term}

\newcommand{\llcxlamterms}{\llcxlamterm s}

\newcommand{\llcxlamcalc}{$\lambda$-calculus}

\newcommand{\llcxlamvar}{$\lambda$-variable}
\newcommand{\llcxlamvars}{\llcxlamvar s}
\newcommand{\llcxlinlamterm}{linear \llcxlamterm}

\newcommand{\llcxlinlamterms}{\llcxlinlamterm s}
\newcommand{\llcxlinlamcalc}{linear $\lambda$-calculus}
\newcommand{\llcxLinlamcalc}{Linear $\lambda$-calculus}
\newcommand{\llcxwithexs}{with explicit substitutions}
\newcommand{\llcxVar}{\mathcal{V}} %
\newcommand{\llcxSet}[1]{\Lambda_{#1}} %
\newcommand{\llcxFV}[1]{\operatorname{fv}(#1)} 
\newcommand{\llcxF}[2]{\lambda #1.#2} 
\newcommand{\llcxA}[2]{(#1)\, #2} 
\newcommand{\llcxE}[3]{#1\, \subst{#3}{#2}} 

\newcommand{\llcxM}{M} 
\newcommand{\llcxN}{N}
\newcommand{\llcxP}{P}
\newcommand{\llcxQ}{Q}
\newcommand{\llcxNot}{\mathtt{Not}}
\newcommand{\llcxTT}{\mathtt{True}}
\newcommand{\llcxFF}{\mathtt{False}}
\newcommand{\llcxX}{x} 
\newcommand{\llcxY}{y}
\newcommand{\llcxW}{w}
\newcommand{\llcxZ}{z}
\newcommand{\llcxo}{o} 
\newcommand{\llcxno}{\overline{o}} 
\newcommand{\llcxp}{p} %
\newcommand{\llcxq}{q} %
\newcommand{\llcxnr}{\overline{r}} %
\newcommand{\llcxSOSJud}[2]{ #1 \Rightarrow #2} %

\newcommand{\llcxSOSBetarule}{\beta} 
\newcommand{\llcxSOStrarule}{\mathsf{tra}} %
\newcommand{\llcxSOSrflrule}{\mathsf{rfl}} %
\newcommand{\llcxSOSfrule  }{\mathsf{f}} %
\newcommand{\llcxSOSalrule}{\mathsf{@l}} %
\newcommand{\llcxSOSarrule}{\mathsf{@r}} %
\newcommand{\llcxSOStrarulein}{\mathsf{tra}} %
\vllineartrue
\renewcommand{\vlscn}[1]{
\!\ifvirginialakesmallbrackets\{\else\left\{\fi #1
  \ifvirginialakesmallbrackets\}\else\right\}\fi   } 
\newcommand{\vlone}{\circ} 

\newcommand{\vlfo}[2]{\vlqu{#1}{#2}}
\newcommand{\vlex}[2]{\vlqu{#1}{#2}}

\newcommand{\vlpl}{\vlor}

\newcommand{\OpNameSeq}{\textnormal{\textsf{Seq}}\xspace}
\newcommand{\OpNamePar}{\textnormal{\textsf{Par}}\xspace}
\newcommand{\OpNameCop}{\textnormal{\textsf{CoPar}}\xspace}
\newcommand{\OpNameRen}{\textnormal{\textsf{Sdq}}\xspace}

\newcommand{\OpNameNot}{\textnormal{\textsf{Not}}\xspace}
\renewcommand{\vlne}[1]{\overline{#1}}

\newcommand{\strFN}[1]{\operatorname{fn}(#1)} 
\newcommand{\strBN}[1]{\operatorname{bn}(#1)} 
\newcommand{\strK}{K} 
\newcommand{\strP}{P}
\newcommand{\strR}{R}
\newcommand{\strS}{S}

\newcommand{\strT}{T}
\newcommand{\strU}{U}
\newcommand{\strV}{V}

\newcommand{\nstrR}{\vlne\strR}

\newcommand{\atma}{a} 

\newcommand{\atmb}{b}

\newcommand{\atmc}{c}

\newcommand{\natma}{\vlne\atma} 

\newcommand{\natmb}{\vlne\atmb}

\newcommand{\natmc}{\vlne\atmc}

\newcommand{\atmLabL}{\mathfrak{l}} 
\newcommand{\atmLabM}{\mathfrak{m}}
\newcommand{\atmLabN}{\mathfrak{n}}

\newcommand{\subst}[2]{\{^{#1}\!/\!_{#2}\}} 

\newcommand{\bvtrdrule}{\mathsf{u}\downarrow}
\newcommand{\bvtrurule}{\mathsf{u}\uparrow}

\newcommand{\bvtrhorule}{\rho}
\newcommand{\bvtmixprule}{\mathsf{mixp}} 
\newcommand{\bvtpmixrule}{\mathsf{pmix}}

\newcommand{\bvttradrule}{\mathsf{t}\downarrow}

\newcommand{\bvttradrulep}{\mathsf{mt}\downarrow}

\newcommand{\bvtintdrule}{\mathsf{i}\downarrow}
\newcommand{\bvtinturule}{\mathsf{i}\uparrow}

\newcommand{\bvtseqdrule}{\mathsf{q}\downarrow}
\newcommand{\bvtsequrule}{\mathsf{q}\uparrow}
\newcommand{\bvtatidrule}{\mathsf{ai}\downarrow}
\newcommand{\bvtatiurule}{\mathsf{ai}\uparrow}
\newcommand{\bvtswirule }{\mathsf{s}}
\newcommand{\bvtpludrule}{\mathsf{p}\downarrow}
\newcommand{\bvtpluurule}{\mathsf{p}\uparrow}

\newcommand{\bvtorerule}{\mathsf{subst}}
\newcommand{\bvtsubsrule}{\bvtorerule}

\newcommand{\bvtsinrule}{\mathsf{beta}}

\newcommand{\bvtrdrulein}{\mbox{$\mathsf{u}\!\!\downarrow$}}
\newcommand{\bvtrurulein}{\mbox{$\mathsf{u}\!\!\uparrow$}}

\newcommand{\bvtrhorulein}{\rho}
\newcommand{\bvtmixprulein}{\mathsf{mixp}} 
\newcommand{\bvtpmixrulein}{\mathsf{pmix}}

\newcommand{\bvttradrulein}{\mbox{$\mathsf{t}\!\!\downarrow$}}

\newcommand{\bvttradruleinp}{\mbox{$\mathsf{mt}\!\!\downarrow$}}

\newcommand{\bvtintdrulein}{\mbox{$\mathsf{i}\!\!\downarrow$}}
\newcommand{\bvtinturulein}{\mbox{$\mathsf{i}\!\!\uparrow$}}
\newcommand{\bvtseqdrulein}{\mbox{$\mathsf{q}\!\!\downarrow$}}

\newcommand{\bvtsequrulein}{\mbox{$\mathsf{q}\!\!\uparrow$}}
\newcommand{\bvtatidrulein}{\mbox{$\mathsf{ai}\!\!\downarrow$}}
\newcommand{\bvtatiurulein}{\mbox{$\mathsf{ai}\!\!\uparrow$}}
\newcommand{\bvtswirulein }{\mathsf{s}}
\newcommand{\bvtpludrulein}{\mbox{$\mathsf{p}\!\!\downarrow$}}
\newcommand{\bvtpluurulein}{\mbox{$\mathsf{p}\!\!\uparrow$}}

\newcommand{\bvtDder}{\mathcal{D}} 
\newcommand{\bvtEder}{\mathcal{E}}
\newcommand{\bvtPder}{\mathcal{P}} 
\newcommand{\bvtQder}{\mathcal{Q}}
\newcommand{\bvtInfer}[2]{#1: #2}
\newcommand{\bvtJudGen}[2]{\vdash_{#1 
}^{#2 
}}

\newcommand{\bvtJ}{\vdash}

\newcommand{\atmo}{o}
\newcommand{\atmp}{p}
\newcommand{\atmq}{q}
\newcommand{\atmr}{r}
\newcommand{\atms}{s}
\newcommand{\atmt}{t}
\newcommand{\natmo}{\vlne\atmo}
\newcommand{\natmp}{\vlne\atmp}
\newcommand{\natmq}{\vlne\atmq}
\newcommand{\natmr}{\vlne\atmr}
\newcommand{\natms}{\vlne\atms}
\newcommand{\natmt}{\vlne\atmt}

\newcommand{\TITLE}{Extending a system in the calculus of structures\\ with a self-dual quantifier}

\title
{\TITLE\\
{\small Luca Roversi}\\
{\small Universit\`a di Torino --- Dipartimento di Informatica\footnote
{{\it E-mail}: luca.roversi@unito.it
}}}

\begin{document}
\date{}
\maketitle

\begin{abstract}
We recall that $\SBV$, a proof system developed under the methodology of
deep inference, extends multiplicative linear logic with the self-dual non-commutative logical operator \OpNameSeq. We introduce $\SBVT$ that extends $\SBV$ by adding the self-dual quantifier \OpNameRen. The system $ \SBVT $ is consistent because we prove that (the analogous of) cut elimination holds for it. Its new logical operator \OpNameRen operationally behaves as a binder, in a way that the interplay between \OpNameSeq, and \OpNameRen\ can model $\beta$-reduction of \llcxlinlamcalc\ inside the cut-free subsystem $\BVT$ of $\SBVT$. The long term aim is to keep developing a programme whose goal is to give pure logical accounts of computational primitives under the proof-search-as-computation analogy, by means of minimal, and incremental extensions of $ \SBV $.
\end{abstract}
\section{Introduction.}
\label{section:Introduction}
This is a work in structural proof-theory. We extend $\SBV$ \cite{Gugl:06:A-System:kl}, the paradigmatic system of the deep inference methodology to design proof systems.
\paragraph{Deep inference (\DI).} One of the main aspects of \DI is that logical systems can be designed as they were rewriting systems, namely, systems with rules that apply \emph{deeply} inside terms, or, equivalently, in any suitable context. We must read ``deep'' as opposed to 
``shallow''. Rules of sequent and natural deduction systems are shallow because they
build proofs whose form mimics the one of formulas.
Thanks to the deep application of its rules, $\BV$ substantially extends
multiplicative linear logic ($\MLL$) \cite{GirardTaylorLafont89} with the non commutative binary operator \OpNameSeq, whose logical properties are strictly connected to the expressiveness of $\BV$ itself.
Any limits we might put on the application depth of $\BV$ rules would yield a strictly less expressive system \cite{Tiu:06:A-System:ai} indeed. An extension of $\BV$, by means of linear logic exponentials \cite{GuglStra:01:Non-comm:rp,GuglStra:02:A-Non-co:lq,
DBLP:journals/mscs/GuglielmiS11,DBLP:journals/tocl/StrassburgerG11} is $\NEL$, whose provability is undecidable \cite{Stra:03:System-N:mb}.
\paragraph{Contributions, and motivations.}
We introduce $\SBVT$. It is $\SBV$ plus a quantifier that we identify as \OpNameRen, which abbreviates ``\textsf{S}elf-\textsf{d}ual \textsf{q}uantifier''. The relevant feature of \OpNameRen is to bind variable names of $\SBVT$ only. The consequence is twofold. First, we do not need to classify \OpNameRen as either an existential, or a universal quantifier. Indeed, binding variable names only, it never requires to distinguish if the quantification is over a variable which we can think of as an assumption or as a conclusion. Hence, a second consequence is that \OpNameRen naturally becomes self-dual. So, $ \SBVT $ can be viewed as a minimal extension of $ \SBVT $ by means of a logical operator whose instances identify regions of formulas where specific variable names can essentially change freely.
\par
The work may be viewed as divided in two parts. 
The first is about proving that $ \SBVT $ is consistent. 
Namely, $ \SBVT $ enjoys Splitting (Section~\ref{section:Splitting theorem of BVT}) which
identifies the subset $\BVT$ of $\SBVT$ which plays the role of cut-free fragment.
\par
The second part of the work gives to \OpNameRen an operational semantics. Exploiting that \OpNameRen is a binder, we show that its interplay with \OpNameSeq makes proof-search inside $\BVT$ complete \wrt the basic functional computation expressed by \llcxlinlamcalc. We recall that functions \llcxlinlamcalc\ represents use their arguments exactly once in the course of the evaluation. So, the set of functions it can express is quite limited, but large enough to let the decision about which is the normal form of two \llcxlinlamterms\ a \emph{polynomial time complete} problem \cite{Mairson:2003JFP}.
Completeness amounts to first defining an embedding
$\mapLcToDi{\cdot}{\cdot}$ from \llcxlinlamterms\ to formulas of
$\BVT$ (Section~\ref{section:Completeness of SBVT and BVT}.) 
Then, completeness states that, for every \llcxlinlamterm\ $\llcxM$, and every atom $\atmo$, which plays the role of an output-channel, if $\llcxM$ reduces to $\llcxN$, then there
is a derivation $ \bvtDder $ of $\BVT$, that derives the conclusion $\mapLcToDi{\llcxM}{\atmo}$
from the assumption $\mapLcToDi{\llcxN}{\atmo}$.
(Theorem~\ref{theorem:Completeness of SBVT}.) 
For example, let us recall a possible encoding of boolean values, and of boolean negation:
\begin{center}
\fbox{
{\small
\begin{minipage}{.95\textwidth}
\[
\llcxNot\equiv
\llcxF{\llcxZ}
      {\llcxF{\llcxX}
             {\llcxF{\llcxY}
                    {\llcxA{\llcxA{\llcxZ}
                                  {\llcxY}}
                           {\llcxX}}}}
\qquad\qquad
\llcxTT\equiv
\llcxF{\llcxW}
      {\llcxF{\llcxZ}
             {\llcxA{\llcxW}
                    {\llcxZ}}}
\qquad\qquad
\llcxFF\equiv
\llcxF{\llcxW}
      {\llcxF{\llcxZ}
             {\llcxA{\llcxZ}
                    {\llcxW}}}
\]
\vspace{-.3cm}
\end{minipage}
} 
}
\end{center}
Figure~\ref{figure:Computing Not applied TT BVT} shows (part of) a non trivial example of completeness. We have a derivation of $\BVT$ whose conclusion encodes $\llcxA{\llcxNot}{\llcxTT}$, while the premise encodes its
$\beta$-reduct
$\llcxF{\llcxX}{\llcxF{\llcxY}{\llcxA{\llcxA{\llcxTT}{\llcxY}}{\llcxX}}}$.
\begin{figure}
\scalebox{.79}{
  \begin{minipage}{\textwidth}
  \input{example-not-applied-to-tt}
  \end{minipage}
 }
\caption{Computing \llcxlamterm\ $\llcxA{\llcxNot}{\llcxTT}$ in $\BVT$.}
\label{figure:Computing Not applied TT BVT}
\end{figure}
\par
Finally, showing completeness means we keep developing a programme whose goal is to give pure logical accounts of computational primitives under the proof-search-as-computa\-tion analogy, by means of minimal extensions of $\SBV$.
This programme begins in \cite{Brus:02:A-Purely:wd}. It shows that \OpNameSeq\ soundly, and completely models $ \CCSR $, the restriction of Milner $\CCS$ \cite{Miln:89:Communic:qo} 
to a fragment that contains sequential, and parallel composition only.
\paragraph{Related works.}
This work directly relates to \cite{Roversi:2010-LLCexDI}, and \cite{Roversi:TLCA11} as follows. 
First here we choose a better terminology. Current ``self-dual quantifier'' \OpNameRen were dubbed as ``renaming'' in both \cite{Roversi:2010-LLCexDI}, and \cite{Roversi:TLCA11}, putting too much emphasis about its operational meaning. Moreover, this work
(i) cleans up the definition, and the properties of \OpNameRen,
(ii) generalizes the statements of some principal property, correcting non-crucial flows in their proofs,
(iii) states and proves deduction and standardization properties,
(iv) includes details of many proofs of the given statements,
(v) simplifies the map $\mapLcToDi{\cdot}{\cdot}$ from \llcxlinlamterms\ to formulas of $ \BVT $ dropping any reference to explicit substitutions inside \llcxlinlamterms, which was, instead, mandatory in \cite{Roversi:2010-LLCexDI,Roversi:TLCA11},
(vi) among the conclusions (Section~\ref{section:Conclusions and future work}), anticipates that $\BVT$ can be complete, and not only sound, \wrt\ a suitable extension of the above fragment $\CCS$ of Milner $\CCS$.
\par
Besides \cite{Brus:02:A-Purely:wd}, further related works are \cite{BrunMcKi:08:An-Algor:wq,Guenot:2013}, and \cite{DBLP:conf/concur/BakelV09}.
\par
The former restates natural deduction of the negative fragment of intuitionistic logic into a \DI\ system from which extracting an algebra of combinators where interpreting \llcxlamterms. So, the connection is the aim of giving a computational interpretation to a \DI system. Further investigation on the computational nature of deep inference system is in \cite{Guenot:2013}. It shows relations among lambda calculi with explicit substitutions, and intuitionistic systems redefined in accordance with the deep inference approach to proof theory. Specifically, \cite{Guenot:2013} shows the impact on the design of $\lambda$-calculi with explicit substitutions  of intuitionistic logic reworked in terms of nested sequents or of calculus of structures,  under both the proofs-as-programs, and fromulas-as-programs paradigms.
\par
Finally, \cite{DBLP:conf/concur/BakelV09} inspired the two-arguments map $\mapLcToDi{\cdot}{\cdot}$ from \llcxlinlamterms\ to formulas of  $\BVT$. Anticipating a bit the content of Section~\ref{section:From llcxlinlamcalc llcxwithexs to structures}, the definition of the basic clause of $\mapLcToDi{\cdot}{\cdot}$ is
$\vlstore{\vlsbr<\llcxX;\natmo>}
\mapLcToDi{\llcxX}{\atmo} 
= \vlread $.
Intuitively, the \llcxlinlamcalc\ variable $\llcxX$ in
$\mapLcToDi{\cdot}{\cdot}$ becomes the name of an input channel to the left of
the occurrence $ \vlse $ of \OpNameSeq. The input channel is forwarded to the output channel $\atmo$ in analogy with the \emph{forwarder} 
$
\vlstore{
\vlsbr\llcxX(\vlone)\,\vldot\,\natmo <\vlone>
}
\textlbrackdbl\llcxX\textrbrackdbl_\atmo 
= \vlread$
which comes from \cite{HondaYoshida:TCS1995}, and which is one of the defining clauses of the \emph{output-based embedding} of
\emph{standard} \llcxlamcalc\ \emph{\llcxwithexs} into $\pi$-calculus
\cite{DBLP:conf/concur/BakelV09}. So, $\SBV$ can model a forwarder, the basic input/output communication flow that \llcxlamvars\ realize. The introduction of \OpNameRen allows to model any \emph{on-the-fly renaming} of channels that serves to model the substitution of a term for a bound variable, namely, the linear $\beta$-reduction process of \llcxlinlamcalc.
\paragraph{Road map.}
Section~\ref{section:Systems SBVT and BVT} introduces the extension $\SBVT$ ($\BVT$) of $\SBV$ ($\BV$).
Section~\ref{section:Splitting theorem of BVT} 
proves that $\SBVT$ is consistent by extending the proof of (the analogous of) \cutelimination for $ \SBV $ to $ \SBVT $.
Section~\ref{section:From llcxlinlamcalc llcxwithexs to structures} recalls \llcxlinlamcalc, and defines the embedding of its terms to formulas of $ \BVT $.
Section~\ref{section:Completeness of SBVT and BVT} shows the completeness of $ \BVT $ \wrt \llcxlinlamcalc, namely it shows that every computation in the latter corresponds to a proof-search in the former.
Section~\ref{section:Conclusions and future work} comments about the lack of a reasonable soundness of $ \BVT $ \wrt to \llcxlamcalc, and points to future work.
\section{Systems $\SBVT$ and $\BVT$}
\label{section:Systems SBVT and BVT}
We recall and clean-up the definitions of \cite{Roversi:2010-LLCexDI,Roversi:TLCA11}.
\paragraph{Structures.}
Let $\atma, \atmb, \atmc, \ldots$ denote the elements of a countable set of
\dfn{positive propositional variables}. Let $\natma, \natmb, \natmc, \ldots$
denote the elements of a countable set of \dfn{negative propositional variables}.
The set of \dfn{names}, which we range over by $\atmLabL, \atmLabM$, and
$\atmLabN$, contains both positive, and negative propositional variables,
and nothing else. Let $\vlone$ be a constant, different from any name, which we call \dfn{unit}. The set of \dfn{atoms} contains both names and the unit, while the set of
\dfn{structures} identifies formulas of $\SBV$. Structures belong to
the language of the grammar in~\eqref{fig:BVT-structures}.
\par\vspace{\baselineskip}\noindent
{\small
  \fbox{
   \begin{minipage}{.974\linewidth}
      \begin{equation}
       \label{fig:BVT-structures}
       \strR  \grammareq \vlone
              \ \ \mid \ \ \atmLabL
              \ \ \mid \ \ \vlne{\strR}
              \ \ \mid \ \ \vlrobrl\strR\vlte\strR\vlrobrr
              \ \ \mid \ \ \vlsbr<\strR;\strR>
              \ \ \mid \ \ \vlsqbrl\strR\vlpa\strR\vlsqbrr
              \ \ \mid \ \ \vlfo{\atma}{\strR}
      \end{equation}
    \end{minipage}
  }
}
\par\vspace{\baselineskip}\noindent
We use $\strK, \strP, \strR, \strT, \strU, \strV$ to range over structures. As in
$\SBV$, $\vlne{\strR}$ is a \OpNameNot structure, $\vlrobrl\strR\vlte\strT\vlrobrr$
is a \OpNameCop structure, $\vlsbr<\strR;\strT>$ is a \OpNameSeq structure, and
$\vlsqbrl\strR\vlpa\strT\vlsqbrr$ is a \OpNamePar structure. The \OpNameRen
structure $\vlfo{\atma}{\strR}$ is new.
It comes with the proviso that $\atma$ must be a positive atom. Namely, 
$\vlfo{\natma}{\strR}$ is not in the syntax. \OpNameRen\
induces notions of \dfn{free}, and \dfn{bound names}, defined in
\eqref{equation:BV2-free-names}.
\par\vspace{\baselineskip}\noindent
  \fbox{
    \begin{minipage}{.974\linewidth}
     \vspace{-.4cm}
     {\small
      \begin{equation}
      \label{equation:BV2-free-names}
  	   \begin{gathered}
  \begin{minipage}{.45\textwidth}
        \vlstore{
    && \Set{\atma}=\strFN{\atma}\cup\strFN{\natma}
    \\
    && \atma\in\strFN{\vlne \strR} \textrm{ if }
      \atma \in
      \strFN{\strR}
    \\
    && \atma\in\strFN{\vlsbr(\strR;\strT)} \textrm{ if }
      \atma \in \strFN{\strR} \cup \strFN{\strT}
    \\
    && \atma\in\strFN{\vlsbr<\strR;\strT>} \textrm{ if }
      \atma \in \strFN{\strR} \cup \strFN{\strT}
    \\
    && \atma\in\strFN{\vlsbr[\strR;\strT]} \textrm{ if }
      \atma \in \strFN{\strR} \cup \strFN{\strT}
    \\
    && \atma\in\strFN{\vlfo{\atmb}{\strR}} \textrm{ if }
      \atma \neq \atmb \textrm{ and }
      \atma\in\strFN{\strR}
}
\begin{eqnarray*}
\vlread
\end{eqnarray*}
  \end{minipage}
  \begin{minipage}{.45\textwidth}
      \vlstore{
    && \emptyset=\strBN{\atma}\cup\strBN{\natma}
    \\
    && \atma\in\strBN{\vlne \strR} \textrm{ if }
    \atma \in
    \strBN{\strR}
    \\
    && \atma\in\strBN{\vlsbr<\strR;\strT>} \textrm{ if }
    \atma \in
    \strBN{\strR}\cup\strBN{\strT}
    \\
    && \atma\in\strBN{\vlsbr(\strR;\strT)} \textrm{ if }
    \atma \in
    \strBN{\strR}\cup\strBN{\strT}
    \\
    && \atma\in\strBN{\vlsbr[\strR;\strT]} \textrm{ if }
    \atma \in
    \strBN{\strR}\cup\strBN{\strT}
    \\
    && \atma\in\strBN{\vlfo{\atmb}{\strR}} \textrm{ if }
    \atma \equiv \atmb \textrm{ or }
    \atma\in\strBN{\strR}
}
\begin{eqnarray*}
  \vlread
\end{eqnarray*}
  \end{minipage}
\end{gathered}
      \end{equation}
     }
     \vspace{-.2cm}
    \end{minipage}
  }
\par\vspace{\baselineskip}\noindent
Finally, \eqref{equation:BV2-structure-substitution} defines
the substitution $\strR\subst{\atma}{\atmb}$ that replaces
(i) the atom $\atma$ for the free occurrences of $\atmb$, and
(ii) the atom $\natma$ for those ones of $\natmb$, in $\strR$.
\par\vspace{\baselineskip}\noindent
  \fbox{
    \begin{minipage}{.974\linewidth}
     \vspace{-.3cm}
     {\small
      \begin{equation}
        \label{equation:BV2-structure-substitution}
        \begin{gathered}
  \begin{minipage}{.2\textwidth}
      \vlstore{
        \vlone\subst{\atma}{\atmb} & \equiv &\vlone           \\
        \atmb\subst{\atma}{\atmb}  & \equiv &\atma            \\
        \natmb\subst{\atma}{\atmb} & \equiv &\natma
}
{\setlength{\arraycolsep}{2pt}
\begin{eqnarray*}
 \vlread
\end{eqnarray*}}

  \end{minipage}
  \begin{minipage}{.35\textwidth}
      \vlstore{
        \atmc\subst{\atma}{\atmb}  & \equiv &\atmc            \\
        \natmc\subst{\atma}{\atmb} & \equiv &\natmc           \\
        {\vlne\strR\subst{\atma}{\atmb}} & \equiv&
           {\vlne{\strR\subst{\atma}{\atmb}}}                  \\
        {\vlsbr(\strR;\strT)}\subst{\atma}{\atmb} &\equiv&
           {\vlsbr(\strR\subst{\atma}{\atmb};
                   \strT\subst{\atma}{\atmb})}
}
{\setlength{\arraycolsep}{2pt}
\begin{eqnarray*}
 \vlread
\end{eqnarray*}}

  \end{minipage}
  \begin{minipage}{.35\textwidth}
      \vlstore{
        {\vlsbr<\strR;\strT>}\subst{\atma}{\atmb} &
              \equiv &
           {\vlsbr<\strR\subst{\atma}{\atmb};
                   \strT\subst{\atma}{\atmb}>}                   \\
        {\vlsbr[\strR;\strT]}\subst{\atma}{\atmb} &
              \equiv &
           {\vlsbr[\strR\subst{\atma}{\atmb};
                   \strT\subst{\atma}{\atmb}]}                   \\
\vlfo{\atmb}{\strR}\subst{\atma}{\atmb} & \equiv & 
          \vlfo{\atmb}{\strR} \\
        \vlfo{\atmc}{\strR}\subst{\atma}{\atmb} & \equiv &
             \vlfo{\atmc}{\strR\subst{\atma}{\atmb}}
}
{\setlength{\arraycolsep}{2pt}
\begin{eqnarray*}
 \vlread
\end{eqnarray*}}
  \end{minipage}
\end{gathered}
      \end{equation}
     }
     \vspace{-.5cm}
    \end{minipage}
  }
\par\vspace{\baselineskip}\noindent
\paragraph{Size of the structures.}
The \dfn{size} $\Size{\strR}$ of $\strR$ is the number of occurrences
of atoms in $\strR$ plus the number of occurrences of \OpNameRen\ that effectively
bind an atom.
\begin{example}[\textbf{\textit{Size of the structures}}]
\label{example:Size of the structures}
We have $\vlstore{\vlsbr[\atma;\natma]}\Size{\vlread} =
\vlstore{\vlfo{\atmb}{\vlsbr[\atma;\natma]}}\Size{\vlread}=2$ for we do not count
the occurrence of $\vlfo{\cdot}{\cdot}$. Instead, we count it in
$\vlstore{\vlsbr[\atma;\natma]}
 \vlfo{\atma}{\vlread}$, getting
$\vlstore{\vlfo{\atma}{\vlsbr[\atma;\natma]}} \Size{\vlread}=3$.
\end{example}
\paragraph{(Structure) Contexts.}
We denote them by $\strS\vlhole$. A context is a structure with a single hole
$\vlhole$ in it. If $\strS\vlscn{\strR}$, then $\strR$ is a \dfn{substructure} of
$\strS$. We shall tend to shorten
$\vlstore{\vlsbr[\strR;\strU]}\strS\vlscn{\vlread}$ as $\strS{\vlsbr[\strR;\strU]}$
when ${\vlsbr[\strR;\strU]}$ fills the hole $\vlhole$ of $\strS\vlhole$ exactly.
\paragraph{Congruence $\approx$ on structures.}
Structures are partitioned by the smallest congruence $\approx$ we obtain as
reflexive, symmetric, transitive and contextual closure of the relation $\sim$ whose
defining clauses are \eqref{align:negation-atom}, through \eqref{align:alpha-symm} here
below.
\par\vspace{\baselineskip}\noindent
  \fbox{
    \begin{minipage}{.974\linewidth}
     {\small
       \begin{minipage}{.48\textwidth}
      \begin{center}
      \smalltitle{\textbf{Negation}}%
{
 \vlstore{
   \label{align:negation-atom}
   \vlne{\vlone} &
   \sim & \vlone
   \\
   \label{align:negation-negation}
   \vlne{\vlne \strR} &
   \sim & \strR
   \\
   \label{align:negation-pa}
   \vlne{\vlsbr[\strR;\strT]} &
   \sim & {\vlsbr(\vlne{\strR};\vlne{\strT})}
   \\
   \label{align:negation-co}
   \vlne{\vlsbr(\strR;\strT)} &
   \sim & {\vlsbr[\vlne{\strR};\vlne{\strT}]}
   \\
   \label{align:negation-seq}
   \vlne{\vlsbr<\strR;\strT>} &
   \sim & {\vlsbr<\vlne{\strR};\vlne{\strT}>}
   \\
   \label{align:negation-fo}
   \vlne{\vlfo{\atma}{\strR}} &
   \sim & \vlex{\atma}{\vlne{\strR}}
}
{\setlength{\arraycolsep}{2pt}
\begin{eqnarray}
 \vlread
\end{eqnarray}}
 }

      \smalltitle{\textbf{Symmetry}}%
{
  \vlstore{
         \label{align:symm-pa}
         \vlsbr[\strR;\strT] & \sim & \vlsbr[\strT;\strR]
         \\
         \label{align:symm-co}
         \vlsbr(\strR;\strT) & \sim & \vlsbr(\strT;\strR)
  }
{\setlength{\arraycolsep}{2pt}
\begin{eqnarray}
 \vlread
\end{eqnarray}}
}
%
      \end{center}
\end{minipage}
\begin{minipage}{.5\textwidth}
      \begin{center}
      \smalltitle{\textbf{Associativity}}%
{
 \vlstore{
   \label{align:assoc-co}
   \vlsbr(\strR;(\strT;\strV))
   & \sim &
   \vlsbr((\strR;\strT);\strV)
   \\
   \label{align:assoc-se}
   \vlsbr<\strR;<\strT;\strV>>
   & \sim &
   \vlsbr<<\strR;\strT>;\strV>
   \\
   \label{align:assoc-pa}
   \vlsbr[\strR;[\strT;\strV]]
   & \sim &
   \vlsbr[[\strR;\strT];\strV]
 }
{\setlength{\arraycolsep}{2pt}
\begin{eqnarray}
 \vlread
\end{eqnarray}}
}

      \smalltitle{\textbf{Unit}}%
{
  \vlstore{
   \label{align:unit-co}
          \vlsbr(\vlone;\strR)  & \sim  &\strR\\
   \label{align:unit-seq}
          \vlsbr<\vlone;\strR>  & \sim  &\vlsbr<\strR;\vlone>
                                           \sim  \strR\\
   \label{align:unit-pa}
          \vlsbr[\vlone;\strR]  & \sim  &\strR
  }
{\setlength{\arraycolsep}{2pt}
\begin{eqnarray}
 \vlread
\end{eqnarray}}
}

      \smalltitle{$\alpha$-\textbf{rule}}%
{
    \vlstore{
	\label{align:alpha-intro}
	  \vlfo{\atma}{\strR} &
	  \sim & \strR
	  \quad\quad\ \textrm{ if } \atma\not\in\strFN{\strR}\\
        \label{align:alpha-varsub}
	  \vlfo{\atma}{\strR\subst{\atma}{\atmb}} &
	  \sim & \vlfo{\atmb}{\strR}
	  \quad \textrm{ if } \atma\not\in\strFN{\strR}
	  \\
        \label{align:alpha-symm}
	  \vlfo{\atma}{\vlfo{\atmb}{\strR}} &
	  \sim & \vlfo{\atmb}{\vlfo{\atma}{\strR}}
         }
{\setlength{\arraycolsep}{2pt}
\begin{eqnarray}
 \vlread
\end{eqnarray}}
}
      \end{center}
\end{minipage}
     }
     \vspace{-.5cm}
    \end{minipage}
  }
\par\vspace{\baselineskip}\noindent
\emph{Contextual closure} means that ${\strS{\vlscn{\strR}}} \approx {\strS{\vlscn{\strT}}}$ whenever $\strR \approx \strT$.
We remark that \OpNameRen\ is self-dual like \OpNameSeq is. When introducing the logical rules we shall clarify why. Thanks to \eqref{align:alpha-symm}, we abbreviate 
$ \vlfo{\atma_n}{\vldots\vlfo{\atma_1}{\strR}\vldots} $ as  $ \vlfo{\vec{\atma}}{\strR}$, where we may also interpret $ \vec{\atma} $ as one of the permutations of $ \atma_1, \ldots, \atma_n $.
\paragraph{The system $\SBVT$.}
It contains the set of inference rules in \eqref{fig:SBVT} here below.
Every rule has form $\vldownsmash{\vlinf{\bvtrhorule}{}{\strR}{\strT}}$, \dfn{name}
$\bvtrhorule$, \dfn{premise} $\strT$, and \dfn{conclusion} $\strR$.
\par\vspace{\baselineskip}\noindent
  \fbox{
    \begin{minipage}{.974\linewidth}
     {\small
      \begin{equation}
        \label{fig:SBVT}
        {\small
\begin{tabular}{ccc}
{$\vlinf{\bvtatidrule}{}{\vlsbr[\atma;\natma]}
                        {\vlone}$}
&
&
\qquad
{$\vlinf{\bvtatiurule}{}{\vlone}
                       {\vlsbr(\atma;\natma)}$}
\\
\\
{$\vlinf{\bvtseqdrule}{}
        {\vlsbr[<\strR;\strT>;<\strU;\strV>]}
        {\vlsbr<[\strR;\strU];[\strT;\strV]>}$}
&\qquad
{$\vlinf{\bvtswirule}{}
        {\vlsbr[(\strR;\strU);\strT]}
        {\vlsbr([\strR;\strT];\strU)}$}
&\qquad
{$\vlinf{\bvtsequrule}{}
        {\vlsbr<(\strR;\strU);(\strT;\strV)>}
        {\vlsbr(<\strR;\strT>;<\strU;\strV>)}$}
\\ \\
{$\vlinf{\bvtrdrule}{}
       {\vlsbr[{\vlfo{\atma}{\strR}};{\vlex{\atma}{\strU}}]}
       {\vlfo{\atma}{\vlsbr[\strR;\strU]}}$}
&\qquad&
\qquad
{$\vlinf{\bvtrurule}{}
       {\vlex{\atma}{\vlsbr(\strR;\strU)}}
       {\vlsbr({\vlex{\atma}{\strR}};{\vlfo{\atma}{\strU}})}$}
\end{tabular}
}
      \end{equation}
     }
    \end{minipage}
  }
\par\vspace{\baselineskip}\noindent
Every (instance of) inference rule can be used in any context, namely as
$
 \vlinf{\bvtrhorule}{}
       {\strS\vlscn{\strR}}
       {\strS\vlscn{\strT}}
$ for any $\strS\vlhole$.
This means that, if a structure $\strU$ matches $\strR$ in $\strS\vlhole$,
it can be rewritten to $\strS\vlscn{\strT}$. This justifies calling $\strR$ the
\dfn{redex} of $\bvtrhorule$, and $\strT$ its \dfn{reduct}.
\paragraph{\dfn{Up} and \dfn{down} fragments of $\SBVT$.}
The set $\Set{\bvtatidrulein, \bvtswirulein, \bvtseqdrulein,
\bvtrdrulein}$ is the \dfn{down fragment} $\BVT$ of $\SBVT$.
The \dfn{up fragment} is $\Set{\bvtatiurulein,\bvtswirulein, \bvtsequrulein,
\bvtrurulein}$. So $\bvtswirulein$ belongs to both.
The rule $\bvtatiurulein$ \emph{plays the role of the cut rule of sequent calculus}.
The down rule for \OpNameRen\ restricts the following one \cite{Strasburger:2008qf}:
{\small
\[
\renewcommand{\vlfo}[2]{\forall{#1}.{#2}}
\renewcommand{\vlex}[2]{\exists{#1}.{#2}}
\vlinf{}{}
      {\vlsbr[{\vlfo{\atma}{\strR}};{\vlex{\atma}{\strU}}]}
      {\vlfo{\atma}{\vlsbr[\strR;\strU]}}
\qquad
\]
}
to binding variable names only.
Limiting \OpNameRen to abstract variables implies that the difference between existentially, and universally quantified names disappears. The reason is that the \cutelimination will have no need to differentiate between the substitution of an existentially quantified variable for a universally quantified one, or vice versa. So, \OpNameRen becomes self-dual.
\paragraph{Derivations vs. proofs.}
A \dfn{derivation} in $\SBVT$ is either a structure or an instance of the above
rules or a sequence of two derivations. Both $\bvtDder$, and $\bvtEder$ will range
over derivations. The topmost structure in a derivation is
its \dfn{premise}. The bottommost is its \dfn{conclusion}. The \dfn{length}
$\Size{\bvtDder}$ of a derivation $\bvtDder$ is the number of rule instances in
$\bvtDder$. A derivation $\bvtDder$ of a structure $\strR$ in $\SBVT$ from a
structure $\strT$ in $\SBVT$, only using a subset $\SBVsub\subseteq\SBVT$ is
$\vlderivation                  {
\vlde{\bvtDder}{\SBVsub}{\strR}{
\vlhy                   {\strT}}}
$.
The equivalent \emph{space-saving} form we shall tend to use is
$\bvtInfer{\bvtDder}{\strT \bvtJudGen{\SBVsub}{}\strR}$.
The derivation
$\vlupsmash{
\vlderivation                  {
\vlde{\bvtDder}{\SBVsub}{\strR}{
\vlhy                   {\strT}}}}$
is a \dfn{proof} whenever $\strT\approx \vlone$. We denote it as
$\vlupsmash{
\vlderivation                  {
\vlde{\bvtPder}{\SBVsub}{\strR}{
\vlhy                   {\vlone}}}}$, or
$\vlupsmash{\vlproof{\bvtPder}{\SBVsub}{\strR}}$,
or $\bvtInfer{\bvtPder}{\ \bvtJudGen{\SBVsub}{}\strR}$. Both $\bvtPder$, and
$\bvtQder$ will range over proofs.
In general, we shall drop $\SBVsub$ when clear from the context.
In a derivation, we write
$
\vliqf{\bvtrhorule_1,\ldots,\bvtrhorule_m,n_1,\ldots,n_p}{}{\strR}{\strT}
$,
whenever we use the rules $\bvtrhorule_1,\ldots,\bvtrhorule_m$ to derive $\strR$
from $\strT$ with the help of $n_1,\ldots,n_p$ instances of
\eqref{align:negation-atom}, \ldots, \eqref{align:symm-co}.
To avoid cluttering derivations, whenever possible, we shall tend to omit the use of
negation axioms \eqref{align:negation-atom}, \ldots, \eqref{align:negation-fo},
associativity axioms \eqref{align:assoc-co}, \eqref{align:assoc-se},
\eqref{align:assoc-pa}, and symmetry aximos \eqref{align:symm-pa},
\eqref{align:symm-co}. This means we avoid writing all
brackets, as in $\vlsbr[\strR;[\strT;\strU]]$, in favor of
$\vlsbr[\strR;\strT;\strU]$, for example.
Finally if, for example, $q>1$ instances of some axiom $(n)$ of
\eqref{align:negation-atom}, \ldots, \eqref{align:alpha-symm} occurs among $n_1,\ldots,n_p$, then we write $(n)^q$.
\paragraph{Admissible and derivable rules.}
A rule $\bvtrhorule$ is \dfn{admissible} for the system $\SBVT$ if
$\bvtrhorule\notin\SBVT$ and, for every derivation $\bvtDder$ \ST\
$\bvtInfer{\bvtDder}{\strT \bvtJudGen{\Set{\bvtrhorule}\cup\SBVT}{}\strR}$, there is
a derivation $\bvtDder'$ \ST\ $\bvtInfer{\bvtDder'}{\strT
\bvtJudGen{\SBVT}{}\strR}$. A rule $\bvtrhorule$ is \dfn{derivable} in
$\SBVsub\subseteq\SBVT$ if $\bvtrhorule\not\in\SBVsub$ and, for every instance
$ 
 {\vlinf{\bvtrhorule}{}{\strR}{\strT}}$, there exists a derivation
$\bvtDder$ in
$\SBVsub$ \ST\ $\bvtInfer{\bvtDder}{\strT \bvtJudGen{\SBVsub}{}\strR}$.
\par
The rules in \eqref{equation:SBVT-core-set-derivable-rules} recall a core set of
rules derivable in $\SBV$, hence in $\SBVT$.
\par\vspace{\baselineskip}\noindent
  \fbox{
    \begin{minipage}{.974\linewidth}
     {\small
      \begin{equation}
         \label{equation:SBVT-core-set-derivable-rules}
    	 \begin{tabular}{ccccccccc}
$\vlinf{\bvtintdrule}{}{\vlsbr[\strR;\nstrR]}
                       {\vlone}$
&\qquad\qquad\qquad\qquad&
$\vlinf{\bvtmixprule}{}{\vlsbr<\strR;\strT>}
                       {\vlsbr(\strR;\strT)}$
\\\\
$\vlinf{\bvtinturule}{}{\vlone}
                       {\vlsbr(\strR;\nstrR)}$
&\qquad\qquad\qquad\qquad&
$\vlinf{\bvtpmixrule}{}{\vlsbr[\strR;\strT]}
                       {\vlsbr<\strR;\strT>}$
\end{tabular}
      \end{equation}
     }
    \end{minipage}
  }
\par\vspace{\baselineskip}\noindent
\paragraph{General interaction down and up.}
In \eqref{equation:SBVT-core-set-derivable-rules}, \dfn{general interaction up} is $\bvtinturulein$, derivable in the set
$\Set{\bvtatiurulein, \bvtswirulein, \bvtsequrulein, \bvtrurulein}$, reasoning by
induction on $\Size{\strR}$, and proceeding by cases on the form of $\strR$. We
show the few steps of the proof, relative the case \OpNameRen:
{\small
\[
\vlderivation                                                        {
\vliq{\eqref{align:alpha-intro}}{}{\vlone}        {
\vlin{\bvtinturule}{\textrm{ind. hypothesis}}
                   {\vlex{\atma}{\vlone}}                          {
\vlin{\bvtrurule}  {}{\vlex{\atma}{\vlsbr(\strR;\nstrR)}}         {
\vliq{\eqref{align:negation-fo}}{}
                  {\vlsbr(\vlfo{\atma}{\strR}
                         ;\vlex{\atma}{\nstrR})}                 {
\vlhy{\vlsbr(\vlfo{\atma}{\strR}
            ;\vlne{\vlex{\atma}{\strR}})}                        }}}}}
\]

}
Similar arguments apply to the cases relative to \OpNameNot, \OpNameCop,
\OpNameSeq, and \OpNamePar. Symmetrically, \dfn{general interaction down}
$\bvtintdrulein$ is derivable in $\Set{\bvtatidrulein, \bvtswirulein,
\bvtseqdrulein, \bvtrdrulein}$.
\paragraph{General \OpNameSeq-transitive up, and down rules.}
In \eqref{equation:SBVT-core-set-derivable-rules} $\bvttradrulein$ is derivable by reasoning inductively on the size of $\strS\vlhole$, and proceeding by cases on its structure, under
the proviso $(*)$ which says that $(\Set{\atma}\cup\strFN{\strT})\cap\strBN{\strS\vlhole}=\emptyset$.
If $\strS\vlhole\approx\vlhole$, then $\bvttradrulein$ is:
{\small
\[
\vlderivation                                                         {
\vliq{\eqref{align:unit-seq}}
                   {}{\vlsbr[<\strR;\atma>;<\natma;\strT>]         } {
\vlin{\bvtseqdrule}{}{\vlsbr[<\strR;\atma >;<\vlone;<\natma;\strT>>]} {
\vliq{\eqref{align:unit-pa}
     ,\eqref{align:unit-seq}}
                   {}{\vlsbr<[\strR;\vlone];[\atma ;<\natma;\strT>]>} {
\vlin{\bvtseqdrule}{}{\vlsbr<\strR;[<\atma;\vlone>;<\natma;\strT>]>}  {
\vliq{\bvtatidrule
     ,\eqref{align:unit-pa}
     ,\eqref{align:unit-seq}}
                   {}{\vlsbr<\strR;<[\atma;\natma];[\vlone;\strT]>> } {
\vlhy                {\vlsbr<\strR;\strT>                       }}}}}}}
\]

}
If $\strS\vlhole\approx\vlsbr(\strS'\vlhole;\strU)$, then:
{\small
\[
\vlderivation                                                {
\vlin{\bvtswirule}
  {}
  {\vlsbr[(\strS'<\strR;\natma>;\strU);<\atma;\strT>]}   {
\vlin{\bvttradrule}
  {\textrm{ind. hypothesis}}
  {\vlsbr([\strS'<\strR;\natma>;<\atma;\strT>];\strU)}{
\vlhy{\vlsbr([\strS'<\strR;\strT>];\strU)}              }}}
\]
}
If $\strS\vlhole\approx\vlsbr{\vlfo{\atmp}{\strS'\vlhole}}$, then:
{\small
\[
\vlderivation                                                    {
\vliq{\eqref{align:alpha-intro}
  ,\bvtrdrule}
  {}
  {\vlsbr[\vlfo{\atmp}{\strS'<\strR;\natma>};<\atma;\strT>]}{
\vlin{\bvttradrule}
  {\textrm{ind. hypothesis}}
  {\vlsbr\vlfo{\atmp}{[\strS'<\strR;\natma>;<\atma;\strT>]}}{
\vlhy{\vlsbr\vlfo{\atmp}{[\strS'<\strR;\strT>]}}              }}}
\]
}
The case with $\strS\vlhole\approx\vlsbr[\strS'\vlhole;\strU]$ is simpler than the
two here above.
\paragraph{Mix rules.}
In \eqref{equation:SBVT-core-set-derivable-rules} both $\bvtmixprulein$, and $\bvtpmixrulein$, show a hierarchy between connectives: \OpNamePar is the lowermost, \OpNameSeq lies in the middle, and \OpNameCop on top \cite{Gugl:06:A-System:kl}.
\dfn{Postfix mix rule} $\bvtmixprulein$ is derivable in $\Set{\bvtsequrulein}$.
\par\vspace{\baselineskip}
Finally, some properties that formalize simple derivations we can always build inside $ \BVT $. The first one says when two structures $\strR$, and
$\strT$ of $\BVT$ can be moved inside a context so that they get one aside the
other.
\begin{proposition}[\textit{\textbf{Context extrusion}}]
\label{proposition:Context extrusion}
$\vlstore{\vlsbr[\strR;\strT]}\strS\vlread\bvtJudGen{\Set{\bvtseqdrule, \bvtrdrule
, \bvtswirule}}{}\vlsbr[\strS\vlscn{ \strR }
;\strT]$, for every $\strS, \strR, \strT$.
\end{proposition}
\begin{proof}
By induction on
$\Size{\strS\vlhole}$, proceeding by cases on the form of
$\strS\vlhole$.
(Details in Appendix~\ref{section:Proof of proposition:Context extrusion}).
\end{proof}
The following statement highlights the scoping nature of \OpNameRen. For proving it,
it is enough to inspect the behavior of the rules in $ \BVT $.
\begin{fact}[\textit{\textbf{\OpNameRen is a scoping operator}}]
\label{fact:OpNamRen is a scoping operator}
Let $\atma, \strU$, and $\strV$ be given.
\begin{enumerate}
\item 
If $\bvtInfer{\bvtDder}
             {\strV\bvtJudGen{\BVT}{}\vlfo{\atma}{\strU}}$, then 
there exist $ \strR $, and $ \bvtDder $ \ST\ 
$\bvtInfer{\bvtDder}
          {\vlfo{\atma}{\strR}\bvtJudGen{\BVT}{}\vlfo{\atma}{\strU}}$.
\item 
For every $ \strR $, if $\bvtInfer{\bvtDder}
             {\vlfo{\atma}{\strR}\bvtJudGen{\BVT}{}\vlfo{\atma}{\strU}}$, then 
$\bvtInfer{\bvtDder'}
          {\strR\bvtJudGen{\BVT}{}\strU}$, for some $ \bvtDder' $.
\end{enumerate}
\end{fact}
The last property says that no new variable can be introduced in the course of a derivation.
\begin{proposition}[\textit{\textbf{$\BVT$ is affine}}]
\label{proposition:BVT is affine}
In every $\bvtInfer{\bvtDder} {\strT \bvtJudGen{\BVT}{} \strR}$, we have
$\Size{\strR} \geq \Size{\strT}$.
\end{proposition}
\begin{proof}
By induction on $\Size{\bvtDder}$, proceeding by cases on its last rule
$\bvtrhorule$.
\end{proof}
\section{Splitting for $\SBVT$}
\label{section:Splitting theorem of BVT}
We recall, and clean the proof of Splitting for $\SBVT$ in \cite{Roversi:2010-LLCexDI,Roversi:TLCA11}.
Splitting can be viewed as a generalization of cut-elimination for sequent calculus-like  systems. Proving Splitting of $\SBVT$ amounts to proving that $\SBVT$, and $\BVT$ are
equivalent, namely that every up-rule is admissible in $\BVT$, or,
equivalently, that we can eliminate every up-rule from any derivation of $\SBVT$.
Since $\bvtatiurulein$ is an up-rule, and it plays the role of the cut rule, proving Splitting means proving also cut-elimination for $ \SBVT $.
\par
The first part of this section traces how Splitting, and some other properties it relies on, works to eliminate $ \bvtrurulein $. The second part, Subsection~\ref{subsection:Details on Splitting}, is for technical eyes interested to the full formal details.
\par
Let us see how Splitting eliminates an occurrence $ (*) $ of $\bvtrurulein$ from a proof $ \bvtPder $ of $ \SBVT $, so focusing on the case that differentiates the proof of Splitting for $ \SBVT $ from the one for $ \SBV $. Let:
{\small
$$\vlderivation                   {
  \vlin{\bvtrurulein}{(*)}
       {\strS\vlsbr{\vlfo{\atma}{(\strR;\strT)}}}      {
  \vlpr{\bvtPder'}{}
       {\strS\vlsbr(\vlfo{\atma}{\strR};\vlfo{\atma}{\strT})}
 }}$$
}
be $\bvtPder$ with $(*)$ the instance of $\bvtrurulein$ we want to eliminate. 
We are going to rewrite $\bvtPder$ to a proof of $\BVT$ with the same conclusion as 
$\bvtPder$, but without $(*)$.
The first step to get rid of $(*)$ is Splitting (Theorem~\ref{theorem:Splitting-ALT}). The instance of Splitting we need, up to some details we can ignore at this level, is:
{\small
$$\textrm{ If }
  \vlderivation                   {
  \vlpr{\bvtQder}{}
       {\strS\vlfo{\atma}{\vlsbr(\strR;\strT)}}
  }
  \textrm{, then }
   \exists \strK, \vec{\atmb} \textrm{ such that  }
   \forall \strV, \textrm{ both  }
   \vlderivation                          {
     \vlde{\bvtDder}{}
 	  {\strS\vlscn{\strV}}{
     \vlhy{\vlfo{\vec{\atmb}}{\vlsbr[\strV;\strK]}}}}
   \textrm{, and }
   \vlderivation                   {
   \vlpr{\bvtQder'}{}
        {\vlsbr[\vlsbr(\strR;\strT);\strK]}}
$$
}
We remark that extracting $\strK$, hidden inside $\bvtDder$, might require many instances of \OpNameRen to emerge, as the outermost occurrence 
$ \vlfo{\vec{\atmb}}{\cdot} $ in the premise of $ \bvtDder $ shows.
We can apply Splitting by taking $\bvtPder$ --- beware, not $ \bvtPder'$ --- as $ \bvtQder$. 
Since $\strV$ in $\bvtDder$ can be any, we choose
$\vlstore{\vlsbr(\strR;\strT)}
 \strV\approx\vlfo{\atma}{\vlread}$, the conclusion of the instance of $\bvtrurulein$ we want to eliminate. From such an instance of $\bvtDder$ we get:
{\small
$$\vlderivation                 {
\vlde{\bvtDder'}{}
 	 {\strS\vlfo{\atma}{\vlsbr(\strR;\strT)}}{
\vlhy{\vlfo{\vec{\atmb}}{\vlsbr[(\strR;\strT);\strK]}}}}
$$
}
Now we extract from $\strK$ the, usually called, killers of $ \strR $, and $ \strT $ inside $\vlsbr(\strR;\strT)$. Namely, we apply the following instance of Shallow splitting (Proposition~\ref{proposition:Shallow Splitting}) to the above $\bvtQder'$:
{\small
$$\textrm{ If }
  \vlderivation                   {
  \vlpr{\bvtPder''}{}
       {\vlsbr[(\strR;\strT);\strK]}}
  \textrm{, then }
  \exists \strK_1, \strK_2, \vec{\atmc}
  \textrm{ such that }
    \vlderivation                          {
    \vlde{\bvtEder}{}
	     {\strK}{
    \vlhy{\vlfo{\vec{\atmc}}
               {\vlsbr[\strK_1;\strK_2]}}}}
    \textrm{ and }
    \vlderivation                          {
    \vlpr{\bvtEder_1}{}
	     {\vlsbr[\strR;\strK_1]}}
    \textrm{ and }
    \vlderivation                          {
    \vlpr{\bvtEder_2}{}
	     {\vlsbr[\strT;\strK_2]}}
$$
}
which, once more, may let instances of \OpNameRen to emerge. Composing $\bvtDder', \bvtEder, \bvtEder_1$, and $\bvtEder_2$, we get the $(*)$-free proof we are looking for:
{\small
$$\vlderivation                                       {
 \vlde{\bvtDder'}{}
      {\strS\vlfo{\atma}{\vlsbr(\strR;\strT)}}        {
 \vlde{\bvtEder}{}
      {\vlfo{\vec{\atmb}}
            {\vlsbr[(\strR;\strT);\strK]}}{
 \vldd{Proposition~\scriptsize{\ref{proposition:Context extrusion}}\ \ }
      {\BVT}
      {\vlfo{\vec{\atmb}}
            {\vlsbr[(\strR;\strT)
                   ;\vlfo{\vec{\atmc}}
                         {\vlsbr[\strK_1;\strK_2]}]}} {
 \vlin{\bvtswirule}{}
      {\vlfo{\vec{\atmc,\atmb}}
            {\vlsbr[(\strR;\strT)
                   ;\strK_1;\strK_2]}}                {
 \vlde{\bvtEder_1}{}
      {\vlfo{\vec{\atmc,\atmb}}
            {\vlsbr[([\strR;\strK_1]
                     ;\strT);\strK_2]}}               {
 \vlde{\bvtEder_2}{}
      {\vlfo{\vec{\atmc,\atmb}}
            {\vlsbr[\strT;\strK_2]}}                  {
 \vlhy{\vlfo{\vec{\atmc,\atmb}}{\vlone}
       \approx
       \vlone}}}}}}}}
$$
}
It is a proof with the same conclusion as $\bvtPder$, without $ (*) $, 
but with, at least, a couple of new instances of both
$\bvtrdrulein$, and $\bvtswirulein$, the first one being ``inside'' 
Proposition~\ref{proposition:Context extrusion}
\subsection{Details on Splitting}
\label{subsection:Details on Splitting}
\begin{proposition}[\textit{\textbf{Provability of structures in $\BVT$}}]
\label{proposition:Derivability of structures in BVT}
Let $\strR$, and $\strT$ be structures, and $\atma$ be a name, and 
$ \bvtPder, \bvtPder_1 $, and $ \bvtPder_2 $ be proofs of $ \BVT $.
\begin{enumerate}
\item\label{enum:Derivability of subformulas-seq}
$\vlstore{\vlsbr<\strR;\strT>}
 \bvtInfer{\bvtPder}
          {\ \bvtJudGen{\BVT}{} \vlread}$
\IFF\
$\bvtInfer{\bvtPder_1}
          {\ \bvtJudGen{\BVT}{} \strR}$ and
$\bvtInfer{\bvtPder_2}{\ \bvtJudGen{\BVT}{} \strT}$.

\item\label{enum:Derivability of subformulas-copar}
$\vlstore{\vlsbr(\strR;\strT)}
 \bvtInfer{\bvtPder}{\ \bvtJudGen{\BVT}{} \vlread}$
\IFF\
$\bvtInfer{\bvtPder_1}{\ \bvtJudGen{\BVT}{} \strR$ and
$\bvtInfer{\bvtPder_2}{\ \bvtJudGen{\BVT}{} \strT}}$.

\item\label{enum:Derivability of subformulas-fo}
$\bvtInfer{\bvtPder}{\ \bvtJudGen{\BVT}{} \vlfo{\atma}{\strR}}$
\IFF\
$\bvtInfer{\bvtPder'}{\ \bvtJudGen{\BVT}{} \strR\subst{\atmb}{\atma}}$,
for every variable $\atmb$.

\end{enumerate}
\end{proposition}
\begin{proof}
\emph{``If implication''}. The proofs of \ref{enum:Derivability of subformulas-seq}
and \ref{enum:Derivability of subformulas-copar}, given in \cite{Gugl:06:A-System:kl}
by induction on
$\Size{\bvtPder}$ inside $\BV$, extend to the cases when the last rule
of $\bvtPder$ is $\bvtrdrulein$. Indeed, the redex of
$\bvtrdrulein$ can only be inside $\strR$ or $\strT$.
Concerning~\ref{enum:Derivability of subformulas-fo}, the assumption implies the
existence of $\bvtInfer{\bvtPder'}{\ \bvtJudGen{}{} \strR\subst{\atma}{\atma}}$,
namely of $\bvtInfer{\bvtPder'}{\,\bvtJ \strR}$. So, we can ``wrap'' $\bvtPder'$
with $\vlfo{\atma}{\cdot}$, exploiting \eqref{align:alpha-intro},
and apply every rule of $\bvtPder'$ deep in the proof $\bvtPder$ we are building.
\par
\emph{``Only if implication''}.
In all the three cases, the proof is by induction on $\Size{\bvtPder}$, proceeding by
cases on its last rule $\bvtrhorulein$. Concerning points~\ref{enum:Derivability of
subformulas-seq}, and~\ref{enum:Derivability of subformulas-copar} a redex can only
be inside $\strR$ or $\strT$. So, the application of the inductive hypothesis is
immediate. Instead, $\atma$ may not belong to $\strFN{\strR}$ in
Point~\ref{enum:Derivability of subformulas-fo}. If this is true, then
\eqref{align:alpha-intro} implies that every instance of $\bvtPder'$ with $\atmb$ in place of $\atma$ exists. The reason is that, by definition, the
substitution~\eqref{equation:BV2-structure-substitution}
distributes over structures, preserving the scope of every instance of \OpNameRen.
Otherwise, if $\atma\in\strFN{\strR}$, then the redex of $\bvtrhorulein$ can only be
inside $\strR$. So, we can conclude thanks to the inductive hypothesis.
\end{proof}
\begin{proposition}[\textit{\textbf{Shallow Splitting in $\BVT$}}]
\label{proposition:Shallow Splitting}
Let $\strR, \strT$, and $\strP$ be structures, and $\atma$ be a name, and  $ \bvtPder$ be a proof of $\BVT$.
\begin{enumerate}
\item\label{enum:Shallow-Splitting-seq}
If $\vlstore{\vlsbr[<\strR;\strT>;\strP]}
    \bvtInfer{\bvtPder}
             {\ \bvtJudGen{\BVT}{} \vlread}$, then there are
$\vlstore{\vlsbr<\strP_1;\strP_2>\bvtJudGen{\BVT}{} \strP}
\bvtInfer{\bvtDder}{\vlread}$, and
$\vlstore{\bvtJudGen{\BVT}{} {\vlsbr[\strR;\strP_1]}}
\bvtInfer{\bvtPder_1}{\ \vlread}$, and
$\vlstore{\bvtJudGen{\BVT}{} {\vlsbr[\strT;\strP_2]}}
 \bvtInfer{\bvtPder_2}{\ \vlread}$, for some
$\strP_1$, and $\strP_2$.

\item\label{enum:Shallow-Splitting-copar}
If $\vlstore{\vlsbr[(\strR;\strT);\strP]}
     \bvtInfer{\bvtPder}{\ \bvtJudGen{\BVT}{} \vlread}$,
then there are
$\vlstore{\vlsbr[\strP_1;\strP_2] \bvtJudGen{\BVT}{} \strP}
 \bvtInfer{\bvtDder}{\vlread}$, and
$\vlstore{\bvtJudGen{\BVT}{}{\vlsbr[\strR;\strP_1]}}
 \bvtInfer{\bvtPder_1}{\ \vlread}$, and
$\vlstore{\bvtJudGen{\BVT}{}{\vlsbr[\strT;\strP_2]}}
 \bvtInfer{\bvtPder_2}{\ \vlread}$, for some $\strP_1$, and $\strP_2$.

\item\label{enum:Shallow-Splitting-atom}
Let
$\vlstore{\vlsbr[\strR;\strP]}
 \bvtInfer{\bvtPder}{\ \bvtJudGen{\BVT}{} \vlread}$
with $\strR\approx\vlsbr[\atmLabL_1;\vldots;\atmLabL_m]$, such that
$ i\neq j $ implies $ \atmLabL_i \neq \vlne{\atmLabL_j} $, 
for every $ i,j \in\Set{1,\ldots,m}$, and $ m>0 $.
Then, for every structure $\strR_0$, and $\strR_1$, if
$\strR\approx\vlsbr[\strR_0;\strR_1]$,
there exists
$\vlstore{\vlne{\strR_1}
          \bvtJudGen{\BVT}{}
          \vlsbr[\strR_0;\strP]}
 \bvtInfer{\bvtDder}
          {\ \vlread}$.

\item\label{enum:Shallow-Splitting-fo}
If $\vlstore{\vlsbr[\vlfo{\atma}{\strR};\strP]}
    \bvtInfer{\bvtPder}{\ \bvtJudGen{}{} \vlread}$,
then there are
$\vlstore{\vlfo{\atma}{\strT} \bvtJudGen{\BVT}{} \strP}
 \bvtInfer{\bvtDder}{\vlread}$, and
$\vlstore{\bvtJudGen{\BVT}{} \vlsbr[\strR;\strT]}
 \bvtInfer{\bvtPder'}{\ \vlread}$, for some $\strT$.
\end{enumerate}
\end{proposition}
\begin{proof}
Following \cite{Gugl:06:A-System:kl}, both statements
\ref{enum:Shallow-Splitting-seq}, and \ref{enum:Shallow-Splitting-copar} must be
proved simultaneously. We reason by induction on the lexicographic order of the pair
$(\Size{\strV}, \Size{\bvtPder})$, where $\strV$ is one between
${\vlsbr[<\strR;\strT>;\strP]}$ or ${\vlsbr[(\strR;\strT);\strP]}$, proceeding by
cases on the last rule $\bvtrhorule$ of $\bvtPder$.
\par
Point~\ref{enum:Shallow-Splitting-atom} relies
on points~\ref{enum:Shallow-Splitting-seq}, \ref{enum:Shallow-Splitting-copar}.
It holds by induction on
$(\Size{\strR}, \Size{\bvtPder})$, proceeding by cases on the last rule of
$\bvtPder$.
Point~\ref{enum:Shallow-Splitting-fo} relies
on points~\ref{enum:Shallow-Splitting-seq}, \ref{enum:Shallow-Splitting-copar}.
It holds by induction on
$\vlstore{\vlfo{\atma}{\strR}}
 (\Size{\vlread}, \Size{\bvtPder})$, proceeding by cases on the last rule of
$\bvtPder$.
(Details in Appendix~\ref{section:Proof of proposition:Shallow Splitting}).
\end{proof}
\begin{remark}
The proviso ``$ i\neq j $ implies $ \atmLabL_i \neq \vlne{\atmLabL_j} $, 
for every $ i,j \in\Set{1,\ldots,m}$'' of Point~\eqref{enum:Shallow-Splitting-atom} in Proposition~\ref{proposition:Shallow Splitting} serves to let the killer of every $ \atmLabL $ be inside $ \vlsbr[\strR_0;\strP] $.
\end{remark}
\vspace{\baselineskip}\par\noindent
Proposition~\ref{proposition:Context Reduction ALTERNATIVE} here below
says that $\strS\vlhole$ supplies the ``context'' $\strU$, required for proving
$\strR$, no matter which structure fills the hole of $\strS\vlhole$.
\begin{proposition}[\textit{\textbf{Context Reduction in $\BVT$}}]
\label{proposition:Context Reduction ALTERNATIVE}
Let $\strR$ be a structure, and $\strS\vlhole$ be a context \ST\
$\bvtInfer{\bvtPder}{\ \bvtJudGen{\BVT}{} \strS\vlscn{\strR}}$.
There are a structure $\strU$, and, possibly, some variables $\vec{\atmb}$
\ST, for every $\strV$, if 
$\strFN{\strV}\cap\strBN{\strR}=\emptyset$, then both
$\vlstore{
 \bvtInfer{\bvtDder}
          {\vlfo{\vec{\atmb}}
                {\vlsbr[\strV;\strU]}
           \bvtJudGen{\BVT}{}
           \strS\vlscn{\strV}}}
 \vlread$,
and 
$\vlstore{\bvtJudGen{\BVT}{} {\vlsbr[\strR;\strU]}}
 \bvtInfer{\bvtQder}{\ \vlread}$.
\end{proposition}
\begin{proof}
The proof is by induction on $\Size{\strS\vlhole}$, proceeding by cases on the
form of $\strS\vlhole$. (Details in Appendix~\ref{section:Proof of
proposition:Context Reduction ALTERNATIVE}).
\end{proof}
\begin{theorem}[{\textit{\textbf{Splitting in $\BVT$}}}]
\label{theorem:Splitting-ALT}
Let $\strR$, and $\strT$, be structures, and $\strS\vlhole$ be a context.
\begin{enumerate}
\item\label{enum:Splitting-seq}
If $\vlstore{\vlsbr<\strR;\strT>}
    \bvtInfer{\bvtPder}{\ \bvtJudGen{\BVT}{} \strS{\vlread}}$,
then there are structures $\strK_1, \strK_2$, and, possibly, some variables
$\vec{\atmb}$ \ST, for every $\strV$ with
$\vlstore{\strFN{\strV}\cap\strBN{\vlsbr<\strR;\strT>}=\emptyset}
 \vlread$,
there are
$\vlstore{\vlsbr[\strV;<\strK_1;\strK_2>]}
 \bvtInfer{\bvtDder}{\vlfo{\vec{\atmb}}{\vlread}
                     \bvtJudGen{\BVT}{} \strS\vlscn{\strV}}$, and
$\vlstore{\bvtJudGen{\BVT}{} {\vlsbr[\strR;\strK_1]}}
 \bvtInfer{\bvtPder_1}{\ \vlread}$, and
$\vlstore{\bvtJudGen{\BVT}{} {\vlsbr[\strT;\strK_2]}}
 \bvtInfer{\bvtPder_2}{\ \vlread}$.

\item\label{enum:Splitting-copar}
If $\vlstore{\vlsbr(\strR;\strT)}
    \bvtInfer{\bvtPder}{\ \bvtJudGen{\BVT}{} \strS{\vlread}}$,
then there are structures $\strK_1, \strK_2$, and, possibly, some variables
$\vec{\atmb}$ \ST, for every $\strV$ with
$\vlstore{\strFN{\strV}\cap\strBN{\vlsbr(\strR;\strT)}=\emptyset}
 \vlread$,
there are
$\vlstore{\vlsbr[\strV;\strK_1;\strK_2]}
 \bvtInfer{\bvtDder}{\vlfo{\vec{\atmb}}{\vlread}
                     \bvtJudGen{\BVT}{} \strS\vlscn{\strV}}$, and
$\vlstore{\bvtJudGen{\BVT}{} {\vlsbr[\strR;\strK_1]}}
 \bvtInfer{\bvtPder_1}{\ \vlread}$, and
$\vlstore{\bvtJudGen{\BVT}{} {\vlsbr[\strT;\strK_2]}}
 \bvtInfer{\bvtPder_2}{\ \vlread}$.

\item\label{enum:Splitting-fo-ex}
If $\vlstore{\vlfo{\atma}{\strR}}
    \bvtInfer{\bvtPder}{\ \bvtJudGen{\BVT}{} \strS{\vlread}}$,
then there are a structure $\strK$, and, possibly, some variables
$\vec{\atmb}$ \ST,
for every $\strV$ with
$\vlstore{\strFN{\strV}\cap\strBN{\vlfo{\atma}{\strR}}=\emptyset}
 \vlread$,
there exist
$\vlstore{\vlsbr[\strV;\strK]}
 \bvtInfer{\bvtDder}
          {\vlfo{\vec{\atmb}}{\vlread} 
           \bvtJudGen{\BVT}{} \strS\vlscn{\strV}}$, and
$\vlstore{\bvtJudGen{\BVT}{} {\vlsbr[\strR;\strK]}}
 \bvtInfer{\bvtPder'}{\ \vlread}$.
\end{enumerate}
\end{theorem}
\begin{proof}
We obtain the proof of the three statements by composing Context Reduction
(Proposition~\ref{proposition:Context Reduction ALTERNATIVE}), and Shallow Splitting
(Proposition~\ref{proposition:Shallow Splitting}) in this order. (Details in
Appendix~\ref{section:Proof of theorem:Splitting}).
\end{proof}
\begin{theorem}[\textit{\textbf{Admissibility of the up fragment for $\BVT$}}]
\label{theorem:Admissibility of the up fragment}
The set $\Set{\bvtatiurulein, \bvtsequrulein, \bvtrurulein}$ in $\SBVT$
is admissible for $\BVT$.
\end{theorem}
\begin{proof}
Use Splitting (Theorem~\ref{theorem:Splitting-ALT}), and Shallow Splitting
(Proposition~\ref{proposition:Shallow Splitting}) (Details in
Appendix~\ref{section:Proof of theorem:Admissibility of the up fragment}.)
\end{proof}
\section{\llcxLinlamcalc\ mapped to $ \BVT $}
\label{section:From llcxlinlamcalc llcxwithexs to structures}
To show that \OpNameRen is not an extemporaneous logical operator we interpret it as binder that, together with \OpNameSeq, models the renaming mechanism of linear $ \beta $-reduction.
\paragraph{\llcxLinlamcalc.}
We recall that \llcxlinlamcalc\ can be viewed as a pair
(\llcxlinlamterms,
linear operational semantics).
Let $\llcxVar$ be a countable set of
variable names we range over by $\llcxX,\llcxY,\llcxW,\llcxZ$. We call $\llcxVar$ the
\dfn{set of \llcxlamvars}.
The set of \dfn{\llcxlinlamterms} is
$\llcxSet{}  = \bigcup_{X\subset\llcxVar}\llcxSet{X}$ we range over by
$\llcxM,\llcxN,\llcxP,\llcxQ$. For every $X\subset\llcxVar$, the set
$\llcxSet{X}$ contains the
\dfn{\llcxlinlamterms\ whose free
variables are in $X$}, and which we define as follows:
(i)   $\llcxX\in\llcxSet{\Set{\llcxX}}$;
(ii)  $\llcxF{\llcxX}{\llcxM}\in\llcxSet{X}$
      if $\llcxM\in\llcxSet{X\cup\Set{\llcxX}}$;
(iii) $\llcxA{\llcxM}{\llcxN}\in\llcxSet{X\cup Y}$
      if  $\llcxM\in\llcxSet{X}$,
          $\llcxN\in\llcxSet{Y}$, and
          $X\cap Y=\emptyset$ ;
(iv)  $\llcxE{\llcxM}{\llcxX}{\llcxP}\in\llcxSet{X\cup Y}$
      if  $\llcxM\in\llcxSet{X\cup\Set{\llcxX}}$,
          $\llcxP\in\llcxSet{Y}$, and
          $X\cap Y=\emptyset$. 
The linear operational semantics that rewrites \llcxlinlamterms\ is the relation
$\llcxSOSJud{}{}\, \subseteq\!\llcxSet{}\!\times\!\llcxSet{}$ here below:
\begin{center}
{\small
  \fbox{
    \begin{minipage}{.974\linewidth}
      \begin{equation}
       \label{equation:LLCXS-rewriting-relation}
			\begin{aligned}
\vlinf{\llcxSOSrflrule}{}
       {\llcxSOSJud{\llcxM}{\llcxM}}
       {}
 &\qquad\quad&
\vlinf{\llcxSOSBetarule}{}
       {\llcxSOSJud{\llcxA{\llcxF{\llcxX}{\llcxM}}{\llcxN}}
                   {\llcxE{\llcxM}{\llcxX}{\llcxN}}}
       {}
 &\qquad\quad&
\vlinf{\llcxSOStrarule}{}
       {\llcxSOSJud{\llcxM}{\llcxN}}
       {\llcxSOSJud{\llcxM}{\llcxP}
        \qquad
        \llcxSOSJud{\llcxP}{\llcxN}}
 \\
\vlinf{\llcxSOSfrule}{}
       {\llcxSOSJud{\llcxF{\llcxX}{\llcxM}}{\llcxF{\llcxX}{\llcxN}}}
       {\llcxSOSJud{\llcxM}{\llcxN}}
  &\qquad\quad&
\vlinf{\llcxSOSalrule}{}
       {\llcxSOSJud{\llcxA{\llcxM}{\llcxP}}{\llcxA{\llcxN}{\llcxP}}}
       {\llcxSOSJud{\llcxM}{\llcxN}}
 &\qquad\quad&
\vlinf{\llcxSOSarrule}{}
       {\llcxSOSJud{\llcxA{\llcxP}{\llcxM}}{\llcxA{\llcxP}{\llcxN}}}
       {\llcxSOSJud{\llcxM}{\llcxN}}
\end{aligned}

      \end{equation}
    \end{minipage}
  }
}
\end{center}
where $ \llcxM\subst{\llcxN}{\llcxX} $ is the usual clash-free 
\emph{substitution}, that replaces $\llcxN$ for the forcefully single occurrence of $\llcxX$ in $\llcxM$. We remark that~\eqref{equation:LLCXS-rewriting-relation} is the reflexive, contextual, and transitive closure of linear $\beta$-reduction we find in rule $\llcxSOSBetarule$. Finally, $\Size{\llcxSOSJud{\llcxM}{\llcxN}}$ denotes the \dfn{number of instances of rules} in~\eqref{equation:LLCXS-rewriting-relation}, used to derive some given
$\llcxSOSJud{\llcxM}{\llcxN}$.
\paragraph{The map \mbox{$\mapLcToDi{\cdot}{\cdot}$}.}
We define it here below, to map terms of $\llcxSet{}$ into structures of $\BVT$.
\begin{center}
{\small
  \fbox{
    \begin{minipage}{.974\linewidth}
    \vspace{-.3cm}
		\begin{subequations}
			\vlstore{
\label{align:map-var-R}
\mapLcToDi{\llcxX}{\atmo} & =
{\vlsbr<\llcxX;\natmo>}
  \textrm{ with }  \natmo \textrm{ fresh }
\\
\label{align:map-fun-R}
\mapLcToDi{\llcxF{\llcxX}{\llcxM}}{\atmo} & =
\vlfo{\llcxX}
     {\mapLcToDi{\llcxM}{\atmo}}
\\
\label{align:map-app-R}
\mapLcToDi{\llcxA{\llcxM}{\llcxN}}{\atmo} & =
\vlex{\atmp}
     {\vlsbr[\mapLcToDi{\llcxM}{\atmp}
            ;\vlex{\atmq}{\mapLcToDi{\llcxN}{\atmq}}
            ;<\atmp;\vlne{\atmo}>]}
}
\begin{align}
\vlread
\end{align}
		\end{subequations}
    \end{minipage}
  }
}
\end{center}
For every \llcxlinlamterm\ $\llcxM$, the structure $\mapLcToDi{\llcxM}{\atmo}$ is such that
(i) $\atmo$ is a unique output channel, and
(ii) every free variable of $\llcxM$ is used as positive atom name that plays the
role of input channel.
Clause~\eqref{align:map-var-R} associates the input channel $\llcxX$ to the fresh
output channel $\atmo$. Intuitively, $\llcxX$ shall be eventually \emph{forwarded}
to $\atmo$, in accordance with terminology taken from \cite{HondaYoshida:TCS1995}.
Clause~\eqref{align:map-fun-R} uses \OpNameRen\ to abstract on the input channel
$\llcxX$. This means to let $\llcxX$ ready to merge with any output channel of a
\llcxlinlamterm\ that has to be substituted for $\llcxX$.
Such a channel comes from the argument of an application, as translated by~\eqref{align:map-app-R}. It wraps $\mapLcToDi{\llcxN}{\atmq}$, abstracting on
its output channel $\atmq$ thanks to \OpNameRen. So, thanks to \OpNameRen, linear
$\beta$-reduction, and its substitution mechanism,
become an identification of channel names inside $\BVT$, as follows:
\begin{center}
{\small
  \fbox{
    \begin{minipage}{.974\linewidth}
		\begin{equation}
		  \label{equation:LLCX-simulating-beta-proof}
				\vlderivation                                              {
	\vliq{\eqref{align:alpha-varsub}}{}
{\vlex{\llcxp}
{
 \mapLcToDi{\llcxA{\llcxF{\llcxX}{\llcxM}}{\llcxN}}{\llcxo}
 \equiv
 \vlsbr[\vlfo{\llcxX}
             {\mapLcToDi{\llcxM}{\llcxp}}
        ;\vlex{\llcxq}
              {\mapLcToDi{\llcxN}{\llcxq}}
        ;<\atmp;\vlne\atmo>]
}
}                                                     {
	\vlin{\bvtrdrule}{}
{\vlex{\llcxp}
{
 \vlsbr[\vlfo{\llcxX}
      {\mapLcToDi{\llcxM}{\llcxp}}
       ;\vlex{\llcxX}{\mapLcToDi{\llcxN}{\llcxX}}
       ;<\atmp;\vlne\atmo>]
}
}                                                     {
	\vlin{\bvtsubsrule}{}
{\vlex{\llcxp}
{
 \vlsbr[\vlfo{\llcxX}
      {[\mapLcToDi{\llcxM}{\llcxp}
              ;\mapLcToDi{\llcxN}{\llcxX}]}
       ;<\atmp;\vlne\atmo>]
}
}                                                     {
	\vliq{\eqref{align:alpha-intro}}{}
{\vlex{\llcxp}
{
 \vlsbr[\vlfo{\llcxX}
      {\mapLcToDi{\llcxM\subst{\llcxN}{\llcxX}}{\llcxp}}
       ;<\atmp;\vlne\atmo>]
}
}                                                     {
	\vlin{\bvttradrulep}{}
{\vlex{\llcxp}
{
 \vlsbr[\mapLcToDi{\llcxM\subst{\llcxN}{\llcxX}}{\llcxp}
       ;<\atmp;\vlne\atmo>]
}
}                                                     {
	\vliq{\eqref{align:alpha-intro}}{}
{\vlex{\llcxp}
{
 \mapLcToDi{\llcxM\subst{\llcxN}{\llcxX}}{\llcxo}
}
}                                                     {
	\vlhy{\mapLcToDi{\llcxM\subst{\llcxN}{\llcxX}}
           {\llcxo}
}}}}}}}}
		\end{equation}
    \end{minipage}
  }
}
\end{center}
In~\eqref{equation:LLCX-simulating-beta-proof} here above
(i) \eqref{align:alpha-varsub} holds because we have that
$\vlex{\llcxq}
      {\mapLcToDi{\llcxN}{\llcxq}}
 \approx
 \vlex{\llcxX}
      {\mapLcToDi{\llcxN}{\llcxq}\subst{\llcxX}{\llcxq}}
 \approx
 \vlex{\llcxX}
      {\mapLcToDi{\llcxN}{\llcxX}}
      $ holds
thanks to the uniqueness of input, and output channels, and thanks to
\OpNameSeq\ which never confuses left, and right-hand sides of
$\vlsbr<\strR;\strT>$,
(ii) the instance of $\bvtrdrulein$
identifies the input channel $\llcxX$ of
$\vlfo{\llcxX}{\mapLcToDi{\llcxM}{\llcxo}}$
with the output channel $\llcxX$ of $\mapLcToDi{\llcxN}{\llcxX}$, after
its renaming by means of \eqref{align:alpha-varsub},
(iii) we are going to show that both $\bvtsubsrule$, and $\bvttradruleinp$ are derivable in $\BVT$, with the second one being a specialization of the transitivity $\bvttradrulein$, 
and
(v) the two occurrences of \eqref{align:alpha-intro} apply because $\llcxX$ and
$\llcxp$ disappear.
\section{Completeness of $\BVT$ \wrt \llcxLinlamcalc}
\label{section:Completeness of SBVT and BVT}
Completeness says that we can mimic every computation step of \llcxlinlamcalc\ as proof-reconstruction inside $\BVT$.
\begin{theorem}[\textbf{\textit{Completeness of $\BVT$}}]
\label{theorem:Completeness of SBVT}
For every $\llcxM$, and $\llcxN$, and $\atmo$, if $\llcxSOSJud{\llcxM}{\llcxN}$, then
$\bvtInfer{\bvtDder}
          { \mapLcToDi{\llcxN}{\atmo}
            \bvtJudGen{\BVT}{}
            \mapLcToDi{\llcxM}{\atmo}
 	  }$.
\end{theorem}
The proof relies on some technical lemma that we detail out in the coming lines.
\begin{lemma}[\textbf{\textit{Output names are linear}}]
\label{lemma:Output names are linear}
For every $\llcxM$, and $\llcxo$, the output name $\llcxo$ of
$\mapLcToDi{\llcxM}{o}$ occurs once.
\end{lemma}
\begin{proof}
By induction on the definition of $\mapLcToDi{\cdot}{\cdot}$,
proceeding by cases on the form of $\llcxM$.
\end{proof}
\begin{lemma}[\textbf{\textit{Substitution in $\BVT$}}]
\label{lemma:Deriving substitution}
For every $\llcxM, \llcxN, \atmp, \atmo$, and $\llcxX$, such that
$\llcxX\in\strFN{\mapLcToDi{\llcxM}{\atmo}}$, in $\BVT$, we can derive:
{\small
{\small
\begin{center}
\begin{tabular}{cccc}
$\vlinf{\bvttradrulep}{}
       {\vlsbr[\mapLcToDi{\llcxM}{\atmp}
              ;<\llcxp;\vlne\atmo>]}
       {\mapLcToDi{\llcxM}{\atmo}}$
&\qquad\qquad\qquad&
$\vlinf{\bvtsubsrule}{}
       {\vlsbr[\mapLcToDi{\llcxM}{\atmo}
              ;\mapLcToDi{\llcxN}{\llcxX}]}
       {\mapLcToDi{\llcxM\subst{\llcxN}{\llcxX}}
                  {\atmo}}$
\end{tabular}
\end{center}
}
} 
\end{lemma}
\begin{proof}
Concerning $\bvttradruleinp$, we reason inductively on the size
of $\mapLcToDi{\cdot}{\cdot}$, proceeding by cases on $\llcxM$.
(Details in Appendix~\ref{section:lemma:Simulating transition prime}.)
Concerning $\bvtsubsrule$, we reason inductively on the size of
$\vlsbr[\mapLcToDi{\llcxM}{\atmo};\mapLcToDi{\llcxN}{\llcxX}]$, exploiting
$\bvttradruleinp$.
(Details in Appendix~\ref{section:Proof of lemma:Deriving substitution}.)
\end{proof}
\begin{lemma}[\textbf{\textit{Linear $\beta$ reduction in $\BVT$}}]
\label{lemma:Simulating llcxBeta}
For every $\llcxM, \llcxN, \atmo$, and $\llcxX$, in $\BVT$, we can derive:
{\small
{\small
\begin{center}
\begin{tabular}{cccc}
$\vlinf{\bvtsinrule}{}
       {\mapLcToDi{\llcxA{\llcxF{\llcxX}{\llcxM}}{\llcxN}}{\atmo}}
       {\mapLcToDi{\llcxE{\llcxM}{\llcxX}{\llcxN}}{\atmo}}
$
\end{tabular}
\end{center}
}
} 
\end{lemma}
\begin{proof}
The rule $\bvtsinrule$ is derived in~\eqref{equation:LLCX-simulating-beta-proof} exploiting the definition of $\mapLcToDi{\cdot}{\cdot}$, and Lemma~\ref{lemma:Deriving substitution}.
\end{proof}
\paragraph{Proof of Theorem~\ref{theorem:Completeness of SBVT}.}
By induction on $\Size{\llcxSOSJud{\llcxM}{\llcxN}}$, proceeding by cases on the
last rule in~\eqref{equation:LLCXS-rewriting-relation} used 
for proving $\llcxSOSJud{\llcxM}{\llcxN}$. If the last rule is $\bvtsinrule$, then
Lemma~\ref{lemma:Simulating llcxBeta} implies the thesis.
Let the last rule be ${\llcxSOStrarulein}$. The inductive hypothesis
implies the existence of $\bvtDder_0$, and $\bvtDder_1$:
{\small
\[
\vlderivation                                                   {
\vlde{\bvtDder_0}{}
     {\mapLcToDi{\llcxM}{\atmo}
     }                                                           {
\vlde{\bvtDder_1}{}
     {\mapLcToDi{\llcxP}{\atmo}
     }                                                           {
\vlhy{\mapLcToDi{\llcxN}{\atmo}
     }                                                            }}}
\]
} 
\noindent
In all the remaining cases we proceed as here above, exploiting that $\BVT$
is a \DI\ system, so we can apply deeply, namely in any context, every of its rules.
\begin{remark}
\label{corollary:Completeness of BVT}
As a corollary, under the same assumption as Theorem~\ref{theorem:Completeness of SBVT}, we have 
$\bvtJudGen{\BVT}{}\vlsbr[\mapLcToDi{\llcxM}{\atmo}
                         ;\vlne{\mapLcToDi{\llcxN}{\atmo}}]$ because we can derive 
$\bvtintdrulein$ in $\BVT$, and we can plug it on top of $\bvtDder$.
\end{remark}
\section{Conclusions and future work}
\label{section:Conclusions and future work}
On the computational interpretation side of proof-search inside $ \BVT $, this work makes no reference to soundness of $ \BVT $ \wrt\ \llcxlinlamcalc.
Soundness is the reverse of completeness. For every $\llcxM, \llcxN$, and $\atmo$, if 
$\vlderivation{
 \vlde{\bvtDder}{\BVT}
      {\mapLcToDi{\llcxM}{\atmo}}{
 \vlhy{\mapLcToDi{\llcxN}{\atmo}}}
} $, then $\llcxSOSJud{\llcxM}{\llcxN}$. A counter example to it is:
{\small
\begin{eqnarray*}
\vlstore{
       \mapLcToDi{\llcxA{\llcxA{\llcxF{\llcxX}
 			               {\llcxM}}
			        {\llcxP}}
		         {\llcxQ}
                  }{\atmo}
        }
\vlread
 & = &
\vlstore{\vlex{\atmr}
              {\vlsbr[\vlex{\atms}
                            {[\vlfo{\llcxX}
                                   {\mapLcToDi{\llcxM}{\atms}}
                              ;\vlex{\llcxp}
                                    {\mapLcToDi{\llcxP}
                                               {\llcxp}
                                    }
                             ;<\atms;\llcxnr>
                            ]}
                      ;\vlex{\llcxq}
                            {\mapLcToDi{\llcxQ}
                                       {\llcxq}
                            }
                      ;<\atmr;\llcxno>
                      ]}
        }
\vlread
\\ & \approx &
\vlstore{\vlex{\atmr}
              {\vlsbr[\vlex{\atms}
                            {[\vlfo{\llcxX}
                                   {\mapLcToDi{\llcxM}{\atms}}
                              ;\vlex{\llcxp}
                                    {\mapLcToDi{\llcxP}
                                               {\llcxp}
                                    }
                             ;<\atms;\llcxnr>
                            ]}
                      ;\vlex{\atms}
                            {\vlex{\llcxq}
                                  {\mapLcToDi{\llcxQ}
                                             {\llcxq}
                                  }}
                      ;<\atmr;\llcxno>
                      ]}
        }
\vlread
\\ & \approx &
\vlstore{\vlex{\atmr}
              {\vlsbr[\vlex{\atms}
                            {[\vlfo{\llcxX}
                                   {\mapLcToDi{\llcxM}{\atms}}
                             ;\vlex{\llcxq}
                                   {\mapLcToDi{\llcxQ}
                                              {\llcxq}}
                             ;\vlex{\llcxp}
                                   {\mapLcToDi{\llcxP}
                                              {\llcxp}
                                   }
                             ;<\atms;\llcxnr>
                            ]}
                      ;<\atmr;\llcxno>
                      ]}
        }
\vlread
\end{eqnarray*}
}
where we would erroneously substitute (the mapping of) $ \llcxQ $ for 
(the mapping of) $ \llcxX $ in (the mapping of) $ \llcxM $.
We think essentially two ways exist to react to the lack of soundness of $ \BVT $ \wrt\ \llcxlinlamcalc. The first is in \cite{Roversi:2010-LLCexDI,Roversi:TLCA11} which proves a weak, and not so interesting form of soundness.
The second way is replacing the target language \llcxlinlamcalc, so moving towards the programme that \cite{Brus:02:A-Purely:wd} begins. It suggests that the natural computational paradigm \wrt\ which $\BVT$ can be sound, is some extension of $ \CCSR $, the fragment of Milner $\CCS$ with sequential and parallel composition only. This is coming work, indeed.
\par
On the proof-theoretical side, whose concern is the minimal, and incremental extension of $\SBV$, an example of which is $\SBVT $, we plan to keep investigating self-dual operators. By By means of a self-dual operator, and in accordance with the proof-search-as-computation paradigm, we plan to model non deterministic choice.
Candidate rules that model a self-dual non-deterministic choice are\footnote{The conjecture about the existence of the two rules $\bvtpludrulein$,
and $\bvtpluurulein$, that model non-deterministic choice, results from
discussions with Alessio Guglielmi.}:
{\small
\[
\vlinf{\bvtpludrule}{}
        {{\vlsbr[[\strR\vlpl\strU];\strT]}}
        {{\vlsbr[[\strR;\strT]\vlpl[\strU;\strT]]}}
\qquad\qquad
\vlinf{\bvtpluurule}{}
        {{\vlsbr[(\strR;\strT)\vlpl(\strU;\strT)]}}
        {{\vlsbr([\strR\vlpl\strU];\strT)}}
\]}
We think they are interesting because they would internalize the non deterministic choice that we apply at the meta-level when searching for proofs, or derivations, inside $ \SBVT $ or $ \SBV $.
\bibliographystyle{plain}
\bibliography{Roversi-SBVQ-proof-theory}
\appendix
\section{Proof of \textit{\textbf{Context extrusion}}
(Proposition~\ref{proposition:Context extrusion},
page~\pageref{proposition:Context extrusion})}
\label{section:Proof of proposition:Context extrusion}
By induction on
$\Size{\strS\vlhole}$, proceeding by cases on the form of
$\strS\vlhole$. The base is with $\strS\vlhole\equiv\vlhole$. The statement
holds simply because (i)
$\strS\vlsbr[\strR;\strT]
 \equiv\vlsbr[\strS\vlscn{\strR};\strT]
 \equiv\vlsbr[\strR;\strT]$,
and (ii) $\vlsbr[\strR;\strT]$ is a structure, so, by definition, a derivation.
\par
As a \emph{first case}, let $\strS\vlhole\equiv {\vlsbr<\strS'\vlhole;\strU>}$.
Then:
{\small
\[
\vlderivation                                            {
\vliq{\bvtseqdrule
     ,\eqref{align:unit-seq}}{}
     {\vlsbr[\strS\vlscn{\strR};\strT]
      \equiv\vlsbr[<\strS'\vlscn{\strR};\strU>;\strT]} {
\vlde{\bvtDder}{}
     {\vlsbr<[\strS'\vlscn{\strR};\strT];\strU>}      {
\vlhy{\vlsbr<\strS'[\strR;\strT];\strU>
      \equiv\strS\vlsbr[\strR;\strT]                  }}}}
\]
} 
where $\bvtDder$ exists by inductive hypothesis which holds thanks to
$\Size{\strS'\vlhole}<\Size{\strS\vlhole}$.
If, instead $\strS\vlhole\equiv {\vlsbr(\strS'\vlhole;\strU)}$,
we can proceed as here above, using ${\bvtswirulein}$ in place of
${\bvtseqdrulein}$.
\par
As a \emph{second case}, let $\strS\vlhole\equiv\vlfo{\atma}{\strS'\vlhole}$.
Without loss of generality, thanks to \eqref{align:alpha-varsub}, we can assume
$\atma\not\in\strFN{\strT}$.
Then:
{\small
\[
\vlderivation                                                 {
\vliq{\bvtrdrule
     ,\eqref{align:alpha-intro}}{}
     {\vlsbr[\strS\vlscn{\strR};\strT]
      \equiv\vlsbr[\vlfo{\atma}{\strS'\vlscn{\strR}};\strT]
      \equiv\vlsbr[\vlfo{\atma}{\strS'\vlscn{\strR}};
                   \vlfo{\atma}{\strT}]
      } {
\vlde{\bvtDder}{}
     {\vlfo{\atma}{\vlsbr[\strS'\vlscn{\strR};\strT]}}  {
\vlhy{\vlfo{\atma}{\strS'\vlsbr[\strR;\strT]         }
      \equiv\strS\vlsbr[\strR;\strT]                   }}}}
\]
} 
where $\bvtDder$ exists by inductive hypothesis which holds thanks to
$\Size{\strS'\vlhole}<\Size{\strS\vlhole}$.
\section{Proof of \textit{\textbf{Shallow
Splitting}} (Proposition~\ref{proposition:Shallow Splitting},
page~\pageref{proposition:Shallow Splitting})}
\label{section:Proof of proposition:Shallow Splitting}
\begin{description}
\item[Proof of Points~\ref{enum:Shallow-Splitting-seq}
and~\ref{enum:Shallow-Splitting-copar}.]
We prove the two statements simultaneously, by induction on the
lexicographic order
$(\Size{\strU},\Size{\bvtPder})$, where $\strU$ is one among
$\vlsbr[<\strR;\strT>;\strP]$, and $\vlsbr[(\strR;\strT);\strP]$,
proceeding by cases on the last rule $\bvtrhorule$ of $\bvtPder$.
\par
As a \emph{first case} for both points~\ref{enum:Shallow-Splitting-seq}
and~\ref{enum:Shallow-Splitting-copar} we assume the redex of $\bvtrhorule$ is
inside one among $\strR, \strT$ or $\strP$. So, $\bvtPder$ is one between:
{\small
\[
\vlderivation                         {
\vlin{\bvtrhorule}{}
     {\vlsbr[<\strR;\strT>;\strP]}   {
\vlpr{\bvtPder'}{}
     {\vlsbr[<\strR';\strT'>;\strP']}}}
\qquad\qquad
\vlderivation                         {
\vlin{\bvtrhorule}{}
     {\vlsbr[(\strR;\strT);\strP]}   {
\vlpr{\bvtPder''}{}
     {\vlsbr[(\strR';\strT');\strP']}}}
\]
} 
where only one among $\strR', \strT', \strP'$ is the reduct of $\bvtrhorule$. We can
conclude by applying the inductive hypothesis on $\bvtPder'$, or $\bvtPder''$, and
$\bvtrhorule$ in the obvious way.
\par
As a \emph{second case} of Point~\ref{enum:Shallow-Splitting-seq}
let $\bvtrhorule$ be $\bvtseqdrulein$ with
$\vlsbr[<<\strR';\strR''>;\strT>;[<\strP';\strP''>;\strP''']]$ as its redex. So,
$\bvtPder$ can be:
{\small
\[
\vlderivation                                               {
\vliq{\eqref{align:assoc-se},\eqref{align:assoc-pa}}{}
     {\vlsbr[<<\strR';\strR''>;\strT>;[<\strP';\strP''>;\strP''']]}   {
\vlin{\bvtseqdrule}{}
     {\vlsbr[[<\strR';<\strR'';\strT>>;<\strP';\strP''>];\strP''']}   {
\vlpr{\bvtPder'}{}
     {\vlsbr[<[\strR';\strP'];[<\strR'';\strT>;\strP'']>;\strP''']}  }}}
\]
} 
Thanks to
$\vlstore{\vlsbr[<<\strR';\strR''>;\strT>;[<\strP';\strP''>;\strP''']]}
 \Size{\vlread}=
 \vlstore{\vlsbr[<[\strR';\strP'];[<\strR'';\strT>;\strP'']>;\strP''']}
 \Size{\vlread}$ and
$\Size{\bvtPder'}$ $<$ $\Size{\bvtPder}$ the inductive hypothesis
holds on $\bvtPder'$ which implies
$\vlstore{\vlsbr<\strP_1;\strP_2>\bvtJ \strP'''}
 \bvtInfer{\bvtEder}{\vlread}$, and
$\vlstore{\ \bvtJ \vlsbr[[\strR';\strP'];\strP_1]}
 \bvtInfer{\bvtPder''}{\vlread}$, and
$\vlstore{\ \bvtJ \vlsbr[[<\strR'';\strT>;\strP''];\strP_2]}
 \bvtInfer{\bvtQder}{\vlread}$.
\par
Thanks to
$\vlstore{\vlsbr[[<\strR'';\strT>;\strP''];\strP_2]}
 \Size{\vlread}$ $<$
 $\vlstore{\vlsbr[<[\strR';\strP'];[<\strR'';\strT>;\strP'']>;\strP''']}
 \Size{\vlread}$ 
the inductive hypothesis
holds on $\bvtQder$ which implies
$\vlstore{\vlsbr<\strU_1;\strU_2> \bvtJ \vlsbr[\strP'';\strP_2]}
 \bvtInfer{\bvtEder'}{\ \vlread}$, and
$\vlstore{\ \bvtJ \vlsbr[\strR'';\strU_1]}
 \bvtInfer{\bvtQder' }{\vlread}$, and
$\vlstore{\ \bvtJ \vlsbr[\strT;\strU_2]}
 \bvtInfer{\bvtQder''}{\vlread}$.
\par
The first derivation and the first proof of $\BVT$ in the statement we have to prove
are:
{\small
\[
\vlderivation                                       {
\vlde{\bvtEder}{}
     {\vlsbr[<\strP';\strP''>;\strP''']}           {
\vlin{\bvtseqdrule}{}
     {\vlsbr[<\strP';\strP''>;<\strP_1;\strP_2>]} {
\vlde{\bvtEder'}{}
     {\vlsbr<[\strP';\strP_1];[\strP'';\strP_2]>}{
\vliq{\eqref{align:assoc-se}}{}
     {\vlsbr<[\strP';\strP_1];<\strU_1;\strU_2>>}{
\vlhy{\vlsbr<<[\strP';\strP_1];\strU_1>;\strU_2>}}}}}}
\qquad\qquad
\vlderivation                                                  {
\vlin{\bvtseqdrule}{}
     {\vlsbr[<\strR';\strR''>;<[\strP';\strP_1];\strU_1>]}    {
\vliq{\eqref{align:assoc-pa}}{}
     {\vlsbr<[\strR';[\strP';\strP_1]];[\strR'';\strU_1]>}   {
\vlde{\bvtQder'}{}
     {\vlsbr<[[\strR';\strP'];\strP_1];[\strR'';\strU_1]>}  {
\vliq{\eqref{align:unit-seq}}{}
     {\vlsbr<[[\strR';\strP'];\strP_1];\vlone>}            {
\vlpr{\bvtPder''}{}
     {\vlsbr[[\strR';\strP'];\strP_1]}                     }}}}}
\]
} 
The second proof of $\BVT$ in the statement we have to prove is $\bvtQder''$.
\par
The situation with  $\bvtrhorule\equiv\bvtseqdrulein$ and
$\vlsbr[<\strR;<\strT';\strT''>>;[<\strP';\strP''>;\strP''']]$ its redex is
analogous to one one just developed.
\par
As a \emph{third case} of Point~\ref{enum:Shallow-Splitting-seq} let $\bvtrhorule$
be $\bvtseqdrulein$ with
$\vlsbr[<\strR;\strT>;[<\strP';\strP''>;[\strU';\strU'']]]$ as its redex. So,
$\bvtPder$ can be:
{\small
\[
\vlderivation                                                               {
\vliq{\eqref{align:assoc-pa},\eqref{align:unit-seq}}{}
     {\vlsbr[<\strR;\strT>;[<\strP';\strP''>;[\strU';\strU'']]]}           {
\vliq{\bvtseqdrule,\eqref{align:unit-seq}}{}
     {\vlsbr[[<\vlone;<\strR;\strT>>;<\strP';\strP''>];[\strU';\strU'']]} {
\vlpr{\bvtPder'}{}
     {\vlsbr[<\strP';[<\strR;\strT>;\strP'']>;[\strU';\strU'']]}          }}}
\]
} 
Thanks to
$\vlstore{\vlsbr[<\strR;\strT>;[<\strP';\strP''>;[\strU';\strU'']]]}
 \Size{\vlread}=
 \vlstore{\vlsbr[<\strP';[<\strR;\strT>;\strP'']>;[\strU';\strU'']]}
 \Size{\vlread}$ and
$\Size{\bvtPder'}$ $<$ $\Size{\bvtPder}$ the inductive hypothesis
holds on $\bvtPder'$ yielding
$\vlstore{\vlsbr<\strP_1;\strP_2>\bvtJ [\strU';\strU'']}
 \bvtInfer{\bvtEder}{\vlread}$, and
$\vlstore{\ \bvtJ \vlsbr[\strP';\strP_1]}
 \bvtInfer{\bvtPder''}{\vlread}$, and
$\vlstore{\ \bvtJ \vlsbr[[<\strR;\strT>;\strP''];\strP_2]}
 \bvtInfer{\bvtQder}{\vlread}$.
\par
Thanks to
$\vlstore{\vlsbr[[<\strR;\strT>;\strP''];\strP_2]}
 \Size{\vlread}$ $<$
$\vlstore{\vlsbr[<\strR;\strT>;[<\strP';\strP''>;[\strU';\strU'']]]}
 \Size{\vlread}$ and
$\Size{\bvtPder'}$ $<$ $\Size{\bvtPder}$ the inductive hypothesis
holds on $\bvtQder$ yielding
$\vlstore{\vlsbr<\strU_1;\strU_2>\bvtJ [\strP'';\strP_2]}
 \bvtInfer{\bvtEder'}{\vlread}$, and
$\vlstore{\ \bvtJ \vlsbr[\strR;\strU_1]}
 \bvtInfer{\bvtQder'}{\vlread}$, and
$\vlstore{\ \bvtJ \vlsbr[\strT;\strU_2]}
 \bvtInfer{\bvtQder''}{\vlread}$.
\par
Both $\bvtQder'$, and $\bvtQder''$ are the two proofs of $\BVT$ of the statement we
have to prove. The derivation of $\BVT$ is:
{\small
\[
\vlderivation                                       {
\vlde{\bvtEder}{}
     {\vlsbr[<\strP';\strP''>;[\strU';\strU'']]}{
\vlin{\bvtseqdrule}{}
     {\vlsbr[<\strP';\strP''>;<\strP_1;\strP_2>]} {
\vlde{\bvtEder'}{}
     {\vlsbr<[\strP';\strP_1];[\strP'';\strP_2]>}{
\vlde{\bvtPder''}{}
     {\vlsbr<[\strP';\strP_1];<\strU_1;\strU_2>>}{
\vliq{\eqref{align:unit-seq}}{}
     {\vlsbr<\vlone;<\strU_1;\strU_2>>}{
\vlhy{\vlsbr<\strU_1;\strU_2>}}}}}}}
\]
} 
\par
As a \emph{fourth case} of Point~\ref{enum:Shallow-Splitting-seq} let $\bvtrhorule$
be
$\bvtswirulein$ with $\vlsbr[<\strR;\strT>;[(\strP';\strP'');\strP''']]$ as its
redex. So, $\bvtPder$ can be:
{\small
\[
\vlderivation                                            {
\vliq{\eqref{align:assoc-pa},\eqref{align:symm-pa}}{}
     {\vlsbr[<\strR;\strT>;[(\strP';\strP'');\strP''']]}  {
\vlin{\bvtswirule}{}
     {\vlsbr[[(\strP';\strP'');<\strR;\strT>];\strP''']}  {
\vliq{\eqref{align:symm-pa}}{}
     {\vlsbr[([\strP';<\strR;\strT>];\strP'');\strP''']}  {
\vlpr{\bvtPder'}{}
     {\vlsbr[([<\strR;\strT>;\strP'];\strP'');\strP''']}}}}}
\]
} 
Thanks to
$\vlstore{\vlsbr[<\strR;\strT>;[(\strP';\strP'');\strP''']]}
 \Size{\vlread} =
 \vlstore{\vlsbr[([<\strR;\strT>;\strP'];\strP'');\strP''']}
 \Size{\vlread}$ and
$\Size{\bvtPder'}$ $<$ $\Size{\bvtPder}$, by the inductive hypothesis,
Point~\ref{enum:Derivability of subformulas-copar} applies to $\bvtPder'$.
This means there exist
$\vlstore{\vlsbr[\strP_1;\strP_2]\bvtJ \strP'''}
 \bvtInfer{\bvtEder}{\vlread}$, and
$\vlstore{\ \bvtJ \vlsbr[[<\strR;\strT>;\strP'];\strP_1]}
 \bvtInfer{\bvtPder''}{\vlread}$, and
$\vlstore{\ \bvtJ \vlsbr[\strP'';\strP_2]}
 \bvtInfer{\bvtQder}{\vlread}$.
\par
Thanks to
$\vlstore{\vlsbr[<\strR;\strT>;\strP']}
 \Size{\vlread}$ $<$
$\vlstore{\vlsbr[([<\strR;\strT>;\strP'];\strP'');\strP''']}
 \Size{\vlread}$ the
inductive hypothesis holds on $\bvtPder''$ which implies
$\vlstore{\vlsbr<\strU_1;\strU_2>\bvtJ \vlsbr[\strP';\strP_1]}
 \bvtInfer{\bvtEder'}{\vlread}$, and
$\vlstore{\ \bvtJ \vlsbr[\strR;\strU_1]}
 \bvtInfer{\bvtQder_1}{\vlread}$, and
$\vlstore{\ \bvtJ \vlsbr[\strT;\strU_2]}
 \bvtInfer{\bvtQder_2}{\vlread}$.
Both $\bvtQder_1$, and $\bvtQder_2$ are the two proofs of $\BVT$ in the statement we
have to prove. The derivation is:
{\small
\[
\vlderivation                                        {
\vlde{\bvtEder}{}
     {\vlsbr[(\strP';\strP'');\strP''']}            {
\vliq{\eqref{align:symm-co},\eqref{align:assoc-pa},\eqref{align:symm-pa}}{}
     {\vlsbr[(\strP';\strP'');[\strP_1;\strP_2]]}  {
\vlin{\bvtswirule}{}
     {\vlsbr[[(\strP'';\strP');\strP_2];\strP_1]}  {
\vlde{\bvtQder}{}
     {\vlsbr[([\strP'';\strP_2];\strP');\strP_1]} {
\vliq{\eqref{align:unit-co}}{}
     {\vlsbr[(\vlone;\strP');\strP_1]} {
\vlde{\bvtEder'}{}
     {\vlsbr[\strP';\strP_1]} {
\vlhy{\vlsbr<\strU_1;\strU_2>}                   }}}}}}}
\]
} 
\par
As a \emph{fifth case} of Point~\ref{enum:Shallow-Splitting-seq} let $\bvtrhorule$ be
$\bvtrdrulein$ with $\vlsbr[<\strR;\strT>;\strP]$ as its redex. This means
$\strP\approx\vlfo{\atma}{\strU}$, for some $\strU$ and $\atma$, that, without loss
of generality, thanks to \eqref{align:alpha-varsub}, we can assume such that 
$\atma\in\strFN{\strU}$, and $\vlstore{\vlsbr<\strR;\strT>}\atma\not\in\strFN{\vlread}$. So, by \eqref{align:alpha-intro}, $\vlstore{\vlsbr<\strR;\strT>}\vlread
     \approx
     \vlstore{\vlfo{\atma}{\vlsbr<\strR;\strT>}}
     \vlread$, the derivation is:
{\small
\[
\vlderivation                                                     {
\vlin{\bvtrdrule}{}
     {\vlsbr[\vlfo{\atma}{<\strR;\strT>};\vlfo{\atma}{\strU}]}{
\vlpr{\bvtPder'}{}
     {\vlfo{\atma}{\vlsbr[<\strR;\strT>;\strU]}}}}
\]
} 
Point~\ref{enum:Derivability of subformulas-fo} of
Proposition~\ref{proposition:Derivability of structures in BVT}, applied on
$\bvtPder'$, implies:
{\small
\[
\vlderivation                      {
\vlpr{\bvtPder''}{}
     {\vlsbr[<\strR;\strT>;\strU]}}
\]
} 
Thanks to
$\vlstore{\vlsbr[<\strR;\strT>;\strU]}
 \Size{\vlread}$ $<$
 $\vlstore{\vlsbr[\vlfo{\atma}{<\strR;\strT>};\vlfo{\atma}{\strU}]}\Size{\vlread}$
the inductive hypothesis holds on $\bvtPder''$ which implies
$\vlstore{\vlsbr<\strP_1;\strP_2>}
 \bvtInfer{\bvtEder}{\ \vlread \bvtJ \strU}$, and
$\vlstore{\vlsbr[\strR;\strP_1]}
 \bvtInfer{\bvtQder_1}{\ \bvtJ \vlread}$, and
$\vlstore{\vlsbr[\strT;\strP_2]}
 \bvtInfer{\bvtQder_2}{\ \bvtJ \vlread}$.
Both $\bvtQder_1$, and $\bvtQder_2$ are the two poofs of $\BVT$ in the stetement
we have to prove. The derivation is
$\vlstore{\vlfo{\atma}{\vlsbr<\strP_1;\strP_2>}}
 \vlread \bvtJ \vlfo{\atma}{\strU}$, we obtain from $\bvtEder$ thanks to
Fact~\ref{fact:OpNamRen is a scoping operator}.

\par
We have exhausted the interesting cases relative to
Point~\ref{enum:Shallow-Splitting-seq}.
\par
Recall that we prove Point~\ref{enum:Shallow-Splitting-seq}, and
Point~\ref{enum:Shallow-Splitting-copar} simultaneously, by induction on the
lexicographic order $(\Size{\strU},\Size{\bvtPder})$, where $\strU$ is one among
$\vlsbr[<\strR;\strT>;\strP]$, and $\vlsbr[(\strR;\strT);\strP]$, proceeding by cases
on the last rule $\bvtrhorule$ of $\bvtPder$. Now we explore the cases relative to
Point~\ref{enum:Shallow-Splitting-copar}.
\par
As a \emph{first case} of Point~\ref{enum:Shallow-Splitting-copar} let $\bvtrhorule$
be $\bvtseqdrulein$ with $\vlsbr[(\strR;\strT);[<\strP';\strP''>;[\strU';\strU'']]]$
as its redex. So,
$\bvtPder$ can be:
{\small
\[
\vlderivation                                                               {
\vliq{\eqref{align:assoc-pa}
     ,\eqref{align:symm-pa}
     ,\eqref{align:unit-seq}}{}
     {\vlsbr[(\strR;\strT);[<\strP';\strP''>;[\strU';\strU'']]]} {
\vliq{\bvtseqdrule}{}
     {\vlsbr[[<[(\strR;\strT);\strU'];\vlone>;<\strP';\strP''>];\strU'']} {
\vliq{\eqref{align:unit-seq}
     ,\eqref{align:symm-pa}}{}
     {\vlsbr[<[[(\strR;\strT);\strU'];\strP'];[\vlone;\strP'']>;\strU'']} {
\vlpr{\bvtPder'}{}
     {\vlsbr[<[[(\strR;\strT);\strP'];\strU'];\strP''>;\strU'']}}}}}
\]
} 
Thanks to
$\vlstore{\vlsbr[(\strR;\strT);[<\strP';\strP''>;[\strU';\strU'']]]}
 \Size{\vlread} =
 \vlstore{\vlsbr[<[[(\strR;\strT);\strP'];\strU'];\strP''>;\strU'']}
 \Size{\vlread}$ and
$\Size{\bvtPder'}$ $<$ $\Size{\bvtPder}$, by the inductive hypothesis,
Point~\ref{enum:Derivability of subformulas-seq} applies to $\bvtPder'$.
There exist
$\vlstore{\vlsbr<\strU_1;\strU_2>\bvtJ \strU''}
 \bvtInfer{\bvtEder}{\vlread}$, and
$\vlstore{\ \bvtJ \vlsbr[[[(\strR;\strT);\strP'];\strU'];\strU_1]}
 \bvtInfer{\bvtPder''}{\vlread}$, and
$\vlstore{\ \bvtJ \vlsbr[\strP'';\strU_2]}
 \bvtInfer{\bvtQder}{\vlread}$.
\par
Thanks to
$\vlstore{\vlsbr[[[(\strR;\strT);\strP'];\strU'];\strU_1]}
 \Size{\vlread}$ $<$
$\vlstore{\vlsbr[<[[(\strR;\strT);\strP'];\strU'];\strP''>;\strU'']}
 \Size{\vlread}$ the inductive hypothesis
holds on $\bvtPder''$ which implies
$\vlstore{\vlsbr[\strP_1;\strP_2]\bvtJ \vlsbr[[\strP';\strU'];\strU_1]}
 \bvtInfer{\bvtEder'}{\vlread}$, and
$\vlstore{\ \bvtJ \vlsbr[\strR;\strP_1]}
 \bvtInfer{\bvtQder_1}{\vlread}$, and
$\vlstore{\ \bvtJ \vlsbr[\strT;\strP_2]}
 \bvtInfer{\bvtQder_2}{\vlread}$.
Both $\bvtQder_1$, and $\bvtQder_2$ are the two proofs of $\BVT$ in the
statement we have to prove. The derivation of $\BVT$ in the statement we have to
prove is:
{\small
\[
\vlderivation                                        {
\vlde{\bvtEder}{}
     {\vlsbr[<\strP';\strP''>;[\strU';\strU'']]}   {
\vliq{\eqref{align:symm-pa}}{}
     {\vlsbr[<\strP';\strP''>;[\strU';<\strU_1;\strU_2>]]}  {
\vlin{\bvtseqdrule}{}
     {\vlsbr[[<\strP';\strP''>;<\strU_1;\strU_2>];\strU']}  {
\vlde{\bvtQder}{}
     {\vlsbr[<[\strP';\strU_1];[\strP'';\strU_2]>;\strU']} {
\vliq{\eqref{align:unit-seq}
     ,\eqref{align:symm-pa}
     ,\eqref{align:assoc-pa}}{}
     {\vlsbr[<[\strP';\strU_1];\vlone>;\strU']} {
\vlde{\bvtEder'}{}
     {\vlsbr[[\strP';\strU'];\strU_1]} {
\vlhy{\vlsbr[\strP_1;\strP_2]}                   }}}}}}}
\]
} 
\par
As a \emph{second case} of Point~\ref{enum:Shallow-Splitting-copar} let $\bvtrhorule$
be $\bvtswirulein$ with
$\vlsbr[((\strR';\strR'');(\strT';\strT''));[\strP';\strP'']]$ as its
redex. So, $\bvtPder$ can be:
{\small
\[
\vlderivation                                                          {
\vliq{\eqref{align:assoc-pa},\eqref{align:symm-pa}
     ,\eqref{align:assoc-co},\eqref{align:symm-co}}{}
     {\vlsbr[((\strR';\strR'');(\strT';\strT''));[\strP';\strP'']]}   {
\vlin{\bvtswirule}{}
     {\vlsbr[[((\strR';\strT');(\strR'';\strT''));\strP'];\strP'']}  {
\vlpr{\bvtPder'}{}
     {\vlsbr[([(\strR';\strT');\strP'];(\strR'';\strT''));\strP'']}  }}}
\]
} 
Both 
$\vlstore{\vlsbr[((\strR';\strR'');(\strT';\strT''));[\strP';\strP'']]}
 \Size{\vlread} =
 \vlstore{\vlsbr[([(\strR';\strT');\strP'];(\strR'';\strT''));\strP'']}
 \Size{\vlread}$, and
$\Size{\bvtPder'}$ $<$ $\Size{\bvtPder}$ imply that the inductive hypothesis applies to
$\bvtPder'$. There exist
$\vlstore{\vlsbr[\strP_1;\strP_2]\bvtJ \strP''}
 \bvtInfer{\bvtEder}{\vlread}$, and
$\vlstore{\ \bvtJ \vlsbr[[(\strR';\strT');\strP'];\strP_1]}
 \bvtInfer{\bvtPder''}{\vlread}$, and
$\vlstore{\ \bvtJ \vlsbr[(\strR'';\strT'');\strP_2]}
 \bvtInfer{\bvtQder}{\vlread}$.
\par
Both
$\vlstore{\vlsbr[[(\strR';\strT');\strP'];\strP_1]}
 \Size{\vlread}$ $<$
$\vlstore{\vlsbr[([(\strR';\strT');\strP'];(\strR'';\strT''));\strP'']}
 \Size{\vlread}$ the
inductive hypothesis holds on $\bvtPder''$ which implies
$\vlstore{\vlsbr[\strU_1';\strU_2']\bvtJ \vlsbr[\strP';\strP_1]}
 \bvtInfer{\bvtEder'}{\vlread}$, and
$\vlstore{\ \bvtJ \vlsbr[\strR';\strU_1']}
 \bvtInfer{\bvtQder'_1}{\vlread}$, and
$\vlstore{\ \bvtJ \vlsbr[\strT';\strU_2']}
 \bvtInfer{\bvtQder'_2}{\vlread}$.
\par
Thanks to
$\vlstore{\vlsbr[(\strR'';\strT'');\strP_2]}
 \Size{\vlread}$ $<$
$\vlstore{\vlsbr[([(\strR';\strT');\strP'];(\strR'';\strT''));\strP'']}
 \Size{\vlread}$ the
inductive hypothesis holds on $\bvtQder$ which implies
$\vlstore{\vlsbr[\strU_1'';\strU_2'']\bvtJ \strP_2}
 \bvtInfer{\bvtEder''}{\vlread}$, and
$\vlstore{\ \bvtJ \vlsbr[\strR'';\strU_1'']}
 \bvtInfer{\bvtQder_1''}{\vlread}$, and
$\vlstore{\ \bvtJ \vlsbr[\strT'';\strU_2'']}
 \bvtInfer{\bvtQder_2''}{\vlread}$.
\par
The derivation and the two proofs of $\BVT$ in the statement we have to prove are:
{\small
\[
\vlderivation                                                {
\vlde{\bvtEder}{}
     {\vlsbr[\strP';\strP'']}                               {
\vliq{\eqref{align:assoc-pa}}{}
     {\vlsbr[\strP';[\strP_1;\strP_2]]}                    {
\vlde{\bvtEder'}{}
     {\vlsbr[[\strP';\strP_1];\strP_2]}                   {
\vlde{\bvtEder''}{}
     {\vlsbr[[\strU_1';\strU_2'];\strP_2]}               {
\vliq{\eqref{align:assoc-pa},\eqref{align:symm-pa}}{}
     {\vlsbr[[\strU_1';\strU_2'];[\strU_1'';\strU_2'']]}               {
\vlhy{\vlsbr[[\strU_1';\strU_1''];[\strU_2';\strU_2'']]} }}}}}}
\]
} 
{\small
\[
\vlderivation                                             {
\vliq{\eqref{align:assoc-pa}, \bvtswirule}{}
     {\vlsbr[(\strR';\strR'');[\strU_1';\strU_1'']]}     {
\vlde{\bvtQder_1'}{}
     {\vlsbr[([\strR';\strU_1'];\strR'');\strU_1'']]}   {
\vliq{\eqref{align:unit-co}}{}
     {\vlsbr[(\vlone;\strR'');\strU_1'']]}             {
\vlpr{\bvtQder_1''}{}
     {\vlsbr[\strR'';\strU_1'']}                      }}}}
\qquad\qquad
\vlderivation                                             {
\vliq{\eqref{align:assoc-pa}, \bvtswirule}{}
     {\vlsbr[(\strT';\strT'');[\strU_2';\strU_2'']]}     {
\vlde{\bvtQder_2'}{}
     {\vlsbr[([\strT';\strU_2'];\strT'');\strU_2'']]}   {
\vliq{\eqref{align:unit-co}}{}
     {\vlsbr[(\vlone;\strT'');\strU_2'']]}             {
\vlpr{\bvtQder_2''}{}
     {\vlsbr[\strT'';\strU_2'']}                      }}}}
\]
} 
\par
As a \emph{third case} of Point~\ref{enum:Shallow-Splitting-copar} let $\bvtrhorule$
be $\bvtswirulein$ with
$\vlsbr[(\strR;\strT);[(\strP';\strP'');[\strU';\strU'']]]$ as its redex. So,
$\bvtPder$ can be:
{\small
\[
\vlderivation                                                               {
\vliq{\eqref{align:assoc-pa},\eqref{align:symm-pa}}{}
     {\vlsbr[(\strR;\strT);[(\strP';\strP'');[\strU';\strU'']]]}           {
\vliq{\bvtswirule}{}
     {\vlsbr[[(\strP';\strP'') ;[(\strR;\strT);\strU']];\strU'']} {
\vlpr{\bvtPder'}{}
     {\vlsbr[([\strP';[(\strR;\strT);\strU']];\strP'');\strU'']}          }}}
\]
} 
Both
$\vlstore{\vlsbr[([\strP';[(\strR;\strT);\strU']];\strP'');\strU'']}
 \Size{\vlread} =
 \vlstore{\vlsbr[(\strR;\strT);[(\strP';\strP'');[\strU';\strU'']]]}
 \Size{\vlread}$, and
$\Size{\bvtPder'}$ $<$ $\Size{\bvtPder}$ imply that the inductive hypothesis
holds on $\bvtPder'$. So, we have
$\vlstore{\vlsbr[\strP_1;\strP_2] \bvtJ \strU''}
 \bvtInfer{\bvtEder}{\vlread}$, and
$\vlstore{\ \bvtJ \vlsbr[\strP'';\strP_2]}
 \bvtInfer{\bvtPder''}{\vlread}$, and
$\vlstore{\ \bvtJ \vlsbr{[[\strP';[(\strR;\strT);\strU']];\strP_1]}}
 \bvtInfer{\bvtQder}{\vlread}$.
\par
Both
$\vlstore{\vlsbr{[[\strP';[(\strR;\strT);\strU']];\strP_1]}}
 \Size{\vlread}$ $<$
$\vlstore{\vlsbr[([\strP';[(\strR;\strT);\strU']];\strP'');\strU'']}
 \Size{\vlread}$, and
$\Size{\bvtPder'}$ $<$ $\Size{\bvtPder}$ imply that the inductive hypothesis
holds on $\bvtQder$. SO, we have
$\vlstore{\vlsbr[\strU_1;\strU_2]\bvtJ [\strP';[\strU';\strP_1]]}
 \bvtInfer{\bvtEder'}{\vlread}$, and
$\vlstore{\ \bvtJ \vlsbr[\strR;\strU_1]}
 \bvtInfer{\bvtQder'}{\vlread}$, and
$\vlstore{\ \bvtJ \vlsbr[\strT;\strU_2]}
 \bvtInfer{\bvtQder''}{\vlread}$.
\par
Both $\bvtQder'$, and $\bvtQder''$ are the two proofs of $\BVT$ of the statement we
have to prove. The derivation of $\BVT$ is:
{\small
\[
\vlderivation                                       {
\vlde{\bvtEder}{}
     {\vlsbr[(\strP';\strP'');[\strU';\strU'']]}{
\vlin{\eqref{align:assoc-pa}
     ,\eqref{align:symm-pa}
     ,\eqref{align:symm-co}}{}
     {\vlsbr[(\strP';\strP'');[\strU';[\strP_1;\strP_2]]]} {
\vlin{\bvtswirule}{}
     {\vlsbr[[(\strP'';\strP');\strP_2];[\strU';\strP_1]]} {
\vlde{\bvtPder''}{}
     {\vlsbr[([\strP'';\strP_2];\strP');[\strU';\strP_1]]}{
\vliq{\eqref{align:unit-co}}{}
     {\vlsbr[(\vlone;\strP');[\strU';\strP_1]]}{
\vlde{\bvtEder'}{}
     {\vlsbr[\strP';[\strU';\strP_1]]}{
\vlhy{\vlsbr[\strU_1;\strU_2]}}}}}}}}
\]
} 
\par
As a \emph{fourth case} of Point~\ref{enum:Shallow-Splitting-copar} let $\bvtrhorule$
be $\bvtrdrulein$ with $\vlsbr[(\strR;\strT);\strP]$ as its redex. This means
$\strP\approx\vlfo{\atma}{\strU}$, for some $\strU$ and $\atma$, that, without loss
of generality, thanks to \eqref{align:alpha-varsub}, we can assume such that
$\vlstore{\vlsbr(\strR;\strT)}\atma\not\in\strFN{\vlread}$. So, by
\eqref{align:alpha-intro},
$\vlstore{\vlsbr(\strR;\strT)}\vlread
  \approx
  \vlstore{\vlfo{\atma}{\vlsbr(\strR;\strT)}}
  \vlread$, and
$\bvtPder$ is:
{\small
\[
\vlderivation                                                     {
\vlin{\bvtrdrule}{}
     {\vlsbr[\vlfo{\atma}{(\strR;\strT)};\vlfo{\atma}{\strU}]}{
\vlpr{\bvtPder'}{}
     {\vlfo{\atma}{\vlsbr[(\strR;\strT);\strU]}}}}
\]
} 
Point~\ref{enum:Derivability of subformulas-fo} of
Proposition~\ref{proposition:Derivability of structures in BVT}, applied on
$\bvtPder'$, implies:
{\small
\[
\vlderivation                      {
\vlpr{\bvtPder''}{}
     {\vlsbr[(\strR;\strT);\strU]}}
\]
} 
Thanks to
$\vlstore{\vlsbr[(\strR;\strT);\strU]}
 \Size{\vlread}$ $<$
 $\vlstore{\vlsbr[\vlfo{\atma}{(\strR;\strT)};\vlfo{\atma}{\strU}]}\Size{\vlread}$
the inductive hypothesis holds on $\bvtPder''$ which implies
$\vlstore{\vlsbr[\strP_1;\strP_2]}
 \bvtInfer{\bvtEder}{\ \vlread \bvtJ \strU}$, and
$\vlstore{\vlsbr[\strR;\strP_1]}
 \bvtInfer{\bvtQder_1}{\ \bvtJ \vlread}$, and
$\vlstore{\vlsbr[\strT;\strP_2]}
 \bvtInfer{\bvtQder_2}{\ \bvtJ \vlread}$.
Both $\bvtQder_1$, and $\bvtQder_2$ are the two poofs of $\BVT$ in the stetement
we have to prove. The derivation is
$\vlstore{\vlfo{\atma}{\vlsbr(\strP_1;\strP_2)}}
 \vlread \bvtJ \vlfo{\atma}{\strU}$, we obtain from $\bvtEder$ thanks to
Fact~\ref{fact:OpNamRen is a scoping operator}.

\item[Proof of Point~\ref{enum:Shallow-Splitting-atom}.]
It holds by induction on
$(\Size{\strR},\Size{\bvtPder})$, proceeding by cases on the last rule $\bvtrhorule$
of $\bvtPder$.
\par
As a \emph{first case} let the redex of $\bvtrhorule$ be inside $\strP$.
So, $\bvtPder$ is:
\[
\vlderivation                         {
\vlin{\bvtrhorule}{}
     {\vlsbr[\strR;\strP]}   {
\vlpr{\bvtPder'}{}
     {\vlsbr[\strR;\strP']}}}
\]
We can conclude by applying the inductive hypothesis on $\bvtPder'$, and
$\bvtrhorule$ in the obvious way.
\par
As a \emph{second case}, let $\bvtrhorule$ be $\bvtseqdrulein$ with
$\strP\approx\vlsbr[<\strP';\strP''>;\strP''']$.
Also, let $\strR_0$, and $\strR_1$ such that $\strR\approx\vlsbr[\strR_0;\strR_1]$.
The proof $\bvtPder$ can be:
{\small
\[
\vlderivation                                                       {
\vliq{\eqref{align:assoc-pa},\eqref{align:unit-seq}}{}
     {\vlsbr[[\strR_0;\strR_1];[<\strP';\strP''>;\strP''']]}       {
\vliq{\bvtseqdrule,\eqref{align:assoc-pa}
     ,\eqref{align:unit-seq},\eqref{align:symm-pa}}{}
     {\vlsbr[\strR_0;[<\strR_1;\vlone>;<\strP';\strP''>];\strP''']}{
\vlpr{\bvtPder'}{}
     {\vlsbr[<[\strR_1;\strP'];\strP''>;[\strR_0;\strP''']]}}}}
\]
} 
Point~\ref{enum:Shallow-Splitting-seq} applies to
$\bvtPder'$.
There are structures $\strP_1, \strP_2$, \ST\
$\vlstore{
 \vlsbr<\strP_1;\strP_2>
 \bvtJudGen{}{}
 \vlsbr[\strR_0;\strP''']}
 \bvtInfer{\bvtEder_0}{\vlread}
 $, and
$\vlstore{\bvtJudGen{}{} {\vlsbr[[\strR_1;\strP'];\strP_1]
                          \approx
                          \vlsbr[\strR_1;[\strP';\strP_1]]}}
 \bvtInfer{\bvtQder_0}{\ \vlread}$, and
$\vlstore{\bvtJudGen{}{} {\vlsbr[\strP'';\strP_2]}}
 \bvtInfer{\bvtQder_1}{\ \vlread}$.
\par
We observe that
$\vlstore{\strR_1}
 \Size{\vlread} <
 \vlstore{\vlsbr[\strR_0;\strR_1]}
 \Size{\vlread}$.
So, the inductive hypothesis holds on $\bvtQder_0$. It implies that,
for every $\strR_0^1, \strR_1^1$,
if $\strR_1\approx\vlsbr[\strR_0^1;\strR_1^1]$,
then
$\vlstore{\vlne{\strR_1^1}
          \bvtJudGen{}{}
          \vlsbr[\strR_0^1;[\strP';\strP_1]]}
  \bvtInfer{\bvtEder_1}{\ \vlread}$.
In particular, it holds
$\vlstore{\vlne{\strR_1}
          \bvtJudGen{}{}
          \vlsbr[\vlone;[\strP';\strP_1]]
          \approx
          \vlsbr[\strP';\strP_1]
         }
  \bvtInfer{\bvtEder_1'}{\ \vlread}$ by
taking $\strR_1\approx\strR_1^1$, and $\vlone\approx\strR_0^1$.
\par
We can conclude as follows:
{\small
\[
\vlderivation                                     {
\vliq{\eqref{align:assoc-pa},\eqref{align:symm-pa}}{}
     {\vlsbr[\strR_0;[<\strP';\strP''>;\strP''']]}    {
\vlde{\bvtEder_0}{}
     {\vlsbr[<\strP';\strP''>;[\strR_0;\strP''']]}  {
\vlin{\bvtseqdrule}{}
     {\vlsbr[<\strP';\strP''>
            ;<\strP_1;\strP_2>]} {
\vlde{\bvtEder_1'}{}
     {\vlsbr<[\strP';\strP_1]
             ;[\strP'';\strP_2]>} {
\vlde{\bvtQder_1}{}
     {\vlsbr<\vlne{\strR_1}
            ;[\strP'';\strP_2]>} {
\vliq{\eqref{align:unit-seq}}{}
     {\vlsbr<\vlne{\strR_1};\vlone>} {
\vlhy{\vlne{\strR_1}}}}}}}}}
\]
} 
\par
As a \emph{third case} let $\bvtrhorule$ be $\bvtseqdrulein$ with
$\strP\approx\vlsbr[<\strP';\strP''>;[\strR';\strR'']]$. So, $\bvtPder$ can be:
{\small
\[
\vlderivation                                                                   {
\vliq{\eqref{align:assoc-pa},\eqref{align:unit-seq}}{}
     {\vlsbr[[\strR_0;\strR_1];[<\strP';\strP''>;[\strR';\strR'']]]}           {
\vliq{\bvtseqdrule,\eqref{align:unit-seq}}{}
     {\vlsbr[[<\vlone;[\strR_0;\strR_1]>;<\strP';\strP''>];[\strR';\strR'']]} {
\vlpr{\bvtPder'}{}
     {\vlsbr[<\strP';[[\strR_0;\strR_1];\strP'']>;[\strR';\strR'']]}          }}}
\]
} 
Point~\ref{enum:Shallow-Splitting-seq} applies to
$\bvtPder'$.
There are structures $\strP_1, \strP_2$ \ST\
there exist
$\vlstore{
 \vlsbr<\strP_1;\strP_2>
 \bvtJudGen{}{}
 \vlsbr[\strR';\strR'']}
 \bvtInfer{\bvtEder_0}{\vlread}
 $, and
$\vlstore{\bvtJudGen{}{} {\vlsbr[\strP';\strP_1]}}
 \bvtInfer{\bvtQder_0}{\ \vlread}$, and\\
$\vlstore{\bvtJudGen{}{}
         {\vlsbr[[[\strR_0;\strR_1];\strP''];\strP_2]
          \approx
          \vlsbr[\strR_1;[\strR_0;[\strP'';\strP_2]]]
          }}
 \bvtInfer{\bvtQder_1}{\ \vlread}$.
\par
We observe that
$\vlstore{\strR_1}
 \Size{\vlread} <
 \vlstore{\vlsbr[\strR_0;\strR_1]}
 \Size{\vlread}$.
So, the inductive hypothesis holds on $\bvtQder_1$.
It implies that, for every $\strR_0^1, \strR_1^1$,
if $\strR_1\approx\vlsbr[\strR_0^1;\strR_1^1]$, then
$\vlstore{\vlne{\strR_1^1}
          \bvtJudGen{}{}
          \vlsbr[\strR_0^1;[\strR_0;[\strP'';\strP_2]]]}
  \bvtInfer{\bvtEder_1}{\ \vlread}$.
In particular, it holds
$\vlstore{\vlne{\strR_1}
          \bvtJudGen{}{}
          \vlsbr[\vlone;[\strR_0;[\strP'';\strP_2]]]
          \approx
          \vlsbr[\strR_0;[\strP'';\strP_2]]
         }
  \bvtInfer{\bvtEder_1'}{\ \vlread}$,
by taking $\strR_1\approx\strR_1^1$, and $\vlone\approx\strR_0^1$.
\par
We can conclude as follows:
{\small
\[
\vlderivation                                                      {
\vliq{\eqref{align:unit-seq},\eqref{align:assoc-pa}}{}
     {\vlsbr[\strR_0;[<\strP';\strP''>;[\strR';\strR'']]]}         {
\vlin{\bvtseqdrule}{}
     {\vlsbr[[<\vlone;\strR_0>;<\strP';\strP''>];[\strR';\strR'']]}{
\vliq{\eqref{align:unit-pa}}{}
     {\vlsbr[<[\vlone;\strP'];[\strR_0;\strP'']>;[\strR';\strR'']]}{
\vlde{\bvtEder_0}{}
     {\vlsbr[<\strP';[\strR_0;\strP'']>;[\strR';\strR'']]}         {
\vlin{\bvtseqdrule}{}
     {\vlsbr[<\strP';[\strR_0;\strP'']>
            ;<\strP_1;\strP_2>]}                                   {
\vliq{\eqref{align:assoc-pa}}{}
     {\vlsbr<[\strP';\strP_1]
                  ;[[\strR_0;\strP''];\strP_2]>}                   {
\vlde{\bvtEder_1'}{}
     {\vlsbr<[\strP';\strP_1]
                  ;[\strR_0;[\strP'';\strP_2]]>}                   {
\vlde{\bvtQder_0}{}
     {\vlsbr<[\strP';\strP_1]
             ;\vlne{\strR_1}>}                                     {
\vliq{\eqref{align:unit-seq}}{}
     {\vlsbr<\vlone;\vlne{\strR_1}>}                               {
\vlhy{\vlne{\strR_1}}}                                     }}}}}}}}}
\]
} 
\par
As a \emph{fourth case}  let $\bvtrhorule$ be
$\bvtswirulein$ with
$\strP\approx\vlsbr[(\strP';\strP'');\strP''']$.
Also, let $\strR_0$, and $\strR_1$ such that $\strR\approx\vlsbr[\strR_0;\strR_1]$.
The proof $\bvtPder$ can be:
{\small
\[
\vlderivation                                                       {
\vliq{\eqref{align:assoc-pa},\eqref{align:symm-pa}}{}
     {\vlsbr[[\strR_0;\strR_1];[(\strP';\strP'');\strP''']]}       {
\vlin{\bvtswirule}{}
     {\vlsbr[[\strR_1;(\strP';\strP'')];[\strR_0;\strP''']]}{
\vlpr{\bvtPder'}{}
     {\vlsbr[([\strR_1;\strP'];\strP'');[\strR_0;\strP''']]}}}}
\]
} 
Point~\ref{enum:Shallow-Splitting-copar} applies to
$\bvtPder'$.
There are structures $\strP_1, \strP_2$, \ST\
there exist
$\vlstore{
 \vlsbr[\strP_1;\strP_2]
 \bvtJudGen{}{}
 \vlsbr[\strR_0;\strP''']}
 \bvtInfer{\bvtEder_0}{\vlread}
 $, and
$\vlstore{\bvtJudGen{}{} {\vlsbr[[\strR_1;\strP'];\strP_1]}}
 \bvtInfer{\bvtQder_0}{\ \vlread}$, and
$\vlstore{\bvtJudGen{}{} {\vlsbr[\strP'';\strP_2]}}
 \bvtInfer{\bvtQder_1}{\ \vlread}$.
\par
We observe that
$\vlstore{\strR_1}
 \Size{\vlread} <
 \vlstore{\vlsbr[\strR_0;\strR_1]}
 \Size{\vlread}$.
So, the inductive hypothesis holds on $\bvtQder_0$.
It implies that, for every $\strR_0^1, \strR_1^1$, if
$\strR_1\approx\vlsbr[\strR_0^1;\strR_1^1]$,
then
$\vlstore{\vlne{\strR_1^1}
          \bvtJudGen{}{}
          \vlsbr[\strR_0^1;[\strP';\strP_1]]}
          \bvtInfer{\bvtEder_1}{\ \vlread}$.
In particular it holds
$\vlstore{\vlne{\strR_1}
          \bvtJudGen{}{}
          \vlsbr[\vlone;[\strP';\strP_1]]
          \approx
          \vlsbr[\strP';\strP_1]
         }
  \bvtInfer{\bvtEder_1'}{\ \vlread}$ by
taking $\strR_1\approx\strR_1^1$, and $\vlone\approx\strR_0^1$.
We can conclude as follows:
{\small
\[
\vlderivation                                           {
\vliq{\eqref{align:assoc-pa},\eqref{align:symm-pa}}{}
     {\vlsbr[\strR_0;[(\strP';\strP'');\strP''']]}    {
\vlde{\bvtEder_0}{}
     {\vlsbr[(\strP';\strP'');[\strR_0;\strP''']]}  {
\vliq{\eqref{align:assoc-pa}
     ,\bvtswirule
     ,\eqref{align:symm-co}}{}
     {\vlsbr[(\strP';\strP'')
                  ;[\strP_1;\strP_2]]}     {
\vlin{\bvtswirule}{}
     {\vlsbr[(\strP'';[\strP';\strP_1])
                    ;\strP_2]}             {
\vlde{\bvtEder_1'}{}
     {\vlsbr([\strP'';\strP_2]
                  ;[\strP';\strP_1])}     {
\vlde{\bvtQder_1}{}
     {\vlsbr([\strP'';\strP_2]
                  ;\vlne{\strR_1})}     {
\vliq{\eqref{align:unit-co}}{}
     {\vlsbr(\vlone;\vlne{\strR_1})} {
\vlhy{\vlne{\strR_1}}}}}}}}}}
\]
} 
\par
As a \emph{fifth case} let $\bvtrhorule$ be
$\bvtrdrulein$ with $\strP\approx\vlfo{\atma}{\strP'}$.
The proof $\bvtPder$ can be:
{\small
\[
\vlderivation                                                     {
\vliq{\eqref{align:alpha-intro}}{}
     {\vlsbr[[\strR_0;\strR_1];\vlfo{\atma}{\strP'}]}{
\vlin{\bvtrdrule}{}
     {\vlsbr[\vlfo{\atma}{[\strR_0;\strR_1]};\vlfo{\atma}{\strP'}]}{
\vlpr{\bvtPder'}{}
     {\vlfo{\atma}{\vlsbr[[\strR_0;\strR_1];\strP']}}}}}
\]
} 
because, thanks to \eqref{align:alpha-varsub}, we can always assume $\strP'$ is
such
that $\vlstore{\vlsbr[\strR_0;\strR_1]}\atma\not\in\strFN{\vlread}$.
Point~\ref{enum:Derivability of subformulas-fo} of
Proposition~\ref{proposition:Derivability of structures in BVT}, applied on
$\bvtPder'$, implies:
{\small
\[
\vlderivation                      {
\vlpr{\bvtPder''}{}
     {\vlsbr[[\strR_0;\strR_1];\strP']}}
\]
} 
We observe that
$\vlstore{\bvtPder''}
 \Size{\vlread}
 <
 \vlstore{\bvtPder}
 \Size{\vlread}$.
So the inductive hypothesis holds on $\bvtPder''$.
It implies that,
for every $\strR_0^1, \strR_1^1$, if
$\vlsbr[\strR_0;\strR_1]\approx\vlsbr[\strR_0^1;\strR_1^1]$,
there are
$\vlstore{\vlne{\strR_1^1}
          \bvtJ
          \vlsbr[\strR_0^1;\strP']}
 \bvtInfer{\bvtEder}{\ \vlread }$.
In particular it holds
$\vlstore{\vlne{\strR_1}
          \bvtJudGen{}{}
          \vlsbr[\strR_0;\strP']
         }
  \bvtInfer{\bvtEder_1'}{\ \vlread}$ by
taking $\strR_1\approx\strR_1^1$, and $\strR_0\approx\strR_0^1$.
We can conclude as follows:
{\small
\[
\vlderivation                                 {
\vliq{\eqref{align:alpha-intro}}{}
     {\vlsbr[\strR_0;\vlfo{\atma}{\strP'}]}   {
\vlin{\bvtrdrule}{}
     {\vlsbr[\vlfo{\atma}{\strR_0}
            ;\vlfo{\atma}{\strP'}]}           {
\vlde{\bvtEder_1'}{}
     {\vlfo{\atma}{\vlsbr[\strR_0;\strP']}}   {
\vliq{\eqref{align:alpha-intro}}{}
     {\vlfo{\atma}{\vlne{\strR_1}}}           {
\vlhy{\vlne{\strR_1}}}                     }}}}
\]
} 
The topmost instance of \eqref{align:alpha-intro} is legal thanks to
$\vlstore{\vlsbr[\strR_0;\strR_1]}\atma\not\in\strFN{\vlread}$.

\item[Proof of Point~\ref{enum:Shallow-Splitting-fo}.]
The proof is by induction on $\Size{\bvtPder}$, proceeding by cases on the last rule
$\bvtrhorule$ of $\bvtPder$.
\par
As a \emph{first case} let the last rule of $\bvtPder$ be $\bvtseqdrulein$
with $\vlsbr[\vlfo{\atma}{\strR};[<\strP';\strP''>;[\strU';\strU'']]]$ as its redex.
So, $\bvtPder$ can be:
{\small
\[
\vlderivation                                                           {
\vliq{\eqref{align:assoc-pa},\eqref{align:unit-seq}
     ,\bvtseqdrule,\eqref{align:unit-seq}          }{}
     {\vlsbr[\vlfo{\atma}{\strR};[<\strP';\strP''>;[\strU';\strU'']]]} {
\vlpr{\bvtPder'}{}
     {\vlsbr[<\strP';[\vlfo{\atma}{\strR};\strP'']>;[\strU';\strU'']]} }}
\]
} 
Point~\ref{enum:Shallow-Splitting-seq} applies to $\bvtPder'$.
There exist
$\vlstore{\vlsbr<\strP_1;\strP_2>\bvtJ [\strU';\strU'']}
 \bvtInfer{\bvtEder}{\vlread}$, and
$\vlstore{\ \bvtJ \vlsbr[\strP';\strP_1]}
 \bvtInfer{\bvtPder''}{\vlread}$, and
$\vlstore{\ \bvtJ \vlsbr[[\vlfo{\atma}{\strR};\strP''];\strP_2]}
 \bvtInfer{\bvtQder}{\vlread}$.
The inductive hypothesis holds on $\bvtQder$. Thanks to
$\vlstore{\vlsbr[[\vlfo{\atma}{\strR};\strP''];\strP_2]}
 \Size{\vlread}$ $<$
$\vlstore{\vlsbr[\vlfo{\atma}{\strR};[<\strP';\strP''>;[\strU';\strU'']]]}
 \Size{\vlread}$  we get
$\vlstore{\vlfo{\atma}{\strU}\bvtJ \vlsbr[\strP'';\strP_2]}
 \bvtInfer{\bvtEder'}{\vlread}$, and
$\vlstore{\ \bvtJ \vlsbr[\strR;\strU]}
 \bvtInfer{\bvtQder'}{\vlread}$.
The proof of $\BVT$ in the statement we have to prove is $\bvtQder'$.
The derivation of $\BVT$ in the statement we have to prove is:
{\small
\[
\vlderivation                                       {
\vlde{\bvtEder}{}
     {\vlsbr[<\strP';\strP''>;[\strU';\strU'']]}{
\vlin{\bvtseqdrule}{}
     {\vlsbr[<\strP';\strP''>;<\strP_1;\strP_2>]} {
\vlde{\bvtEder'}{}
     {\vlsbr<[\strP';\strP_1];[\strP'';\strP_2]>}{
\vlde{\bvtPder''}{}
     {\vlsbr<[\strP';\strP_1];\vlfo{\atma}{\strU}>}{
\vliq{\eqref{align:unit-seq}}{}
     {\vlsbr<\vlone;\vlfo{\atma}{\strU}>}{
\vlhy{\vlfo{\atma}{\strU}}}}}}}}
\]
} 
\par
As a \emph{second case} let the last rule of $\bvtPder$ be
$\bvtseqdrulein$ with $\vlsbr[\vlfo{\atma}{\strR};[<\strP';\strP''>;\strP''']]$ as
its redex. So, $\bvtPder$ can be:
{\small
\[
\vlderivation                                                   {
\vliq{\eqref{align:assoc-pa},\eqref{align:unit-seq}
     ,\bvtseqdrule,\eqref{align:unit-pa}}{}
     {\vlsbr[\vlfo{\atma}{\strR};[<\strP';\strP''>;\strP''']]} {
\vlpr{\bvtPder'}{}
     {\vlsbr[<\strP';[\vlfo{\atma}{\strR};\strP'']>;\strP''']} }}
\]
} 
Point~\ref{enum:Shallow-Splitting-seq} applies to $\bvtPder'$.
There exist
$\vlstore{\vlsbr<\strP_1;\strP_2>\bvtJ \strP'''}
 \bvtInfer{\bvtEder}{\vlread}$, and
$\vlstore{\ \bvtJ \vlsbr[\strP';\strP_1]}
 \bvtInfer{\bvtPder''}{\vlread}$, and
$\vlstore{\ \bvtJ \vlsbr[\vlfo{\atma}{\strR};[\strP'';\strP_2]]}
 \bvtInfer{\bvtQder}{\vlread}$.
Thanks to
$\vlstore{\vlsbr[\vlfo{\atma}{\strR};[\strP'';\strP_2]]}
 \Size{\vlread}$ $<$
 $\vlstore{\vlsbr[\vlfo{\atma}{\strR};[<\strP';\strP''>;\strP''']]}
 \Size{\vlread}$
the inductive hypothesis holds on $\bvtQder$ which implies
$\vlstore{\vlfo{\atma}{\strU} \bvtJ \vlsbr[\strP'';\strP_2]}
 \bvtInfer{\bvtEder'}{\ \vlread}$, and
$\vlstore{\ \bvtJ \vlsbr[\strR;\strU]}
 \bvtInfer{\bvtQder'}{\vlread}$. The proof of $\BVT$ in the statement we have to
prove is $\bvtQder'$. The derivation is:
{\small
\[
\vlderivation                                       {
\vlde{\bvtEder}{}
     {\vlsbr[<\strP';\strP''>;\strP''']}           {
\vlin{\bvtseqdrule}{}
     {\vlsbr[<\strP';\strP''>;<\strP_1;\strP_2>]} {
\vlde{\bvtPder''}{}
     {\vlsbr<[\strP';\strP_1];[\strP'';\strP_2]>}{
\vliq{\eqref{align:unit-seq}}{}
     {\vlsbr<\vlone;[\strP'';\strP_2]>}         {
\vlde{\bvtEder'}{}
     {\vlsbr[\strP'';\strP_2]}                 {
\vlhy{\vlfo{\atma}{\strU}             }        }}}}}}
\]
} 
\par
As a \emph{third case} let the last rule of $\bvtPder$ be
$\bvtswirulein$ with
$\vlsbr[\vlfo{\atma}{\strR};[(\strP';\strP'');\strP''']]$ as its redex.
So, $\bvtPder$ can be:
{\small
\[
\vlderivation                                                   {
\vlin{\eqref{align:assoc-pa},\eqref{align:symm-pa}
     ,\bvtswirule,\eqref{align:symm-pa}}{}
     {\vlsbr[\vlfo{\atma}{\strR};[(\strP';\strP'');\strP''']]} {
\vlpr{\bvtPder'}{}
     {\vlsbr[([\strP';\vlfo{\atma}{\strR}];\strP'');\strP''']}}}
\]
} 
Point~\ref{enum:Derivability of subformulas-copar} applies to $\bvtPder'$.
There exist
$\vlstore{\vlsbr[\strP_1;\strP_2]\bvtJ \strP'''}
 \bvtInfer{\bvtEder}{\vlread}$, and
$\vlstore{\ \bvtJ \vlsbr[[\vlfo{\atma}{\strR};\strP'];\strP_1]}
 \bvtInfer{\bvtPder''}{\vlread}$, and
$\vlstore{\ \bvtJ \vlsbr[\strP'';\strP_2]}
 \bvtInfer{\bvtQder}{\vlread}$.
Thanks to
$\vlstore{
 \vlsbr[\vlfo{\atma}{\strR};[\strP';\strP_1]]}
 \Size{\vlread}$ $<$
$\vlstore{\vlsbr[\vlfo{\atma}{\strR}
                ;[(\strP';\strP'');\strP''']
                ]
         }
 \Size{\vlread}$ the
inductive hypothesis holds on $\bvtPder''$ which implies
$\vlstore{\vlfo{\atma}{\strU}\bvtJ \vlsbr[\strP';\strP_1]}
 \bvtInfer{\bvtEder'}{\vlread}$, and
$\vlstore{\ \bvtJ \vlsbr[\strR;\strU]}
 \bvtInfer{\bvtQder'}{\vlread}$.
The proof of $\BVT$ in the statement we have to prove is $\bvtQder'$.
The derivation of $\BVT$ in the statement we have to prove is:
{\small
\[
\vlderivation                                        {
\vlde{\bvtEder}{}
     {\vlsbr[(\strP';\strP'');\strP''']}            {
\vlin{\eqref{align:symm-co},\eqref{align:assoc-pa}
     ,\eqref{align:symm-pa},\bvtswirule} { }
     {\vlsbr[(\strP';\strP'');[\strP_1;\strP_2]]}  {
\vlde{\bvtQder}{}
     {\vlsbr[([\strP'';\strP_2];\strP');\strP_1]} {
\vliq{\eqref{align:unit-co}}{}
     {\vlsbr[(\vlone;\strP');\strP_1]} {
\vlde{\bvtEder'}{}
     {\vlsbr[\strP';\strP_1]} {
\vlhy{\vlfo{\atma}{\strU}}                   }}}}}}
\]
} 
\par
As a \emph{fourth case} let the last rule of $\bvtPder$ be $\bvtrdrulein$ with
$\vlsbr[\vlfo{\atma}{\strR};\strP]$ as its redex. This means
$\strP\approx\vlfo{\atma}{\strU}$. So, $\bvtPder$ is:
{\small
\[
\vlderivation                                             {
\vliq{}{}
     {\vlsbr[\vlfo{\atma}{\strR};\strP]}                 {
\vlin{\bvtrdrule}{}
     {\vlsbr[\vlfo{\atma}{\strR};\vlfo{\atma}{\strU}]}  {
\vlpr{\bvtPder'}{}
     {\vlfo{\atma}{\vlsbr[\strR;\strU]}                }}}}
\]
} 
Point~\ref{enum:Derivability of subformulas-fo} of
Proposition~\ref{proposition:Derivability of structures in BVT}, applied on
$\bvtPder'$, implies the existence of 
$\vlstore{\vlsbr[\strR;\strU]}
 \bvtInfer{\bvtPder''}{\ \bvtJudGen{}{} \vlread}$, which is the proof of $\BVT$ in
the statement we have to prove. The derivation is
$\vlfo{\atma}{\strU} \bvtJudGen{}{} \vlfo{\atma}{\strU}$.
\end{description}
\section{Proof of \textit{Context Reduction} (Proposition~\ref{proposition:Context Reduction ALTERNATIVE}, page~\pageref{proposition:Context Reduction ALTERNATIVE})}
\label{section:Proof of proposition:Context Reduction ALTERNATIVE}
The proof is by induction on $\Size{\strS\vlhole}$, proceeding by cases on the
form of $\strS\vlhole$.
\par
As a \emph{first case}, let
$\strS\vlhole
 \approx \vlsbr<\strS'\vlhole;\strP>$.
So, the assumption is
$\vlstore{\vlsbr<\strS'\vlscn\strR;\strP>}
 \bvtInfer{\bvtPder}{\ \bvtJ{\vlread}}$.
Point~\ref{enum:Derivability of subformulas-seq} of
Proposition~\ref{proposition:Derivability of structures in BVT} implies
$\bvtInfer{\bvtPder'}{\ \bvtJ \strS'\,\vlscn{\strR}}$, and
$\bvtInfer{\bvtPder''}{\ \bvtJ \strP}$.
Thanks to
$\Size{\strS'\,\vlscn{\strR}} <
 \vlstore{\vlsbr<\strS'\,\vlscn{\strR};\strP>}
 \Size{\vlread}$
the inductive hypothesis holds on $\bvtPder'$.
There are $\strU$, and $\vec{\atmb}$ \ST, for every
$\strV$ with $\strFN{\strV}\cap\strBN{\strR}=\emptyset$, both
$\vlstore{
 \bvtInfer{\bvtDder}
          {\vlfo{\vec{\atmb}}{\vlsbr[\strV;\strU]}
           \bvtJ
           \strS'\,\vlscn{\strV}}}
 \vlread$,
and $\vlstore{\vlsbr[\strR;\strU]}
 \bvtInfer{\bvtPder'''}{\ \bvtJ \vlread}$.
The proof $\bvtPder'''$ is the one we are looking for.
To get the derivation we are looking for, we fix $\strV$ such that
$\strFN{\strV}\cap\strBN{\strR}=\emptyset$. This allows to use $\bvtDder$
as follows:
{\small \[
\vlderivation                                    {
\vlde{\bvtPder''}{}
     {\vlsbr<\strS'\vlscn{\strV};\strP>}        {
\vliq{\eqref{align:unit-seq}}{}
     {\vlsbr<\strS'\vlscn{\strV};\vlone>}      {
\vlde{\bvtDder}{}
     {\strS'\vlscn{\strV}}                    {
\vlhy{\vlfo{\vec{\atmb}}
           {\vlsbr[\strV;\strU]}}             }}}}
\] }
\par 
As a \emph{second case}, let 
$\strS\vlhole\approx\vlsbr(\strS'\vlhole;\strP)$.
So, the assumption is
$\vlstore{\vlsbr(\strS'\,\vlscn\strR;\strP)}
 \bvtInfer{\bvtPder}{\ \bvtJ{\vlread}}$.
Point~\ref{enum:Shallow-Splitting-copar} of
Proposition~\ref{proposition:Derivability of structures in BVT}
implies $\bvtInfer{\bvtPder'}{\ \bvtJ \strS'\,\vlscn{\strR}}$, and
$\bvtInfer{\bvtPder''}{\ \bvtJ \strP}$.
Thanks to
$\Size{\strS'\,\vlscn{\strR}} <
 \vlstore{\vlsbr(\strS'\,\vlscn{\strR};\strP)}
 \Size{\vlread}$
the inductive hypothesis holds on $\bvtPder'$.
There are $\strU$, and $\vec{\atmb}$ \ST, for every $\strV$ with $\strFN{\strV}\cap\strBN{\strR}=\emptyset$, both
$\vlstore{
 \bvtInfer{\bvtDder}
          {\vlfo{\vec{\atmb}}
                {\vlsbr[\strV;\strU]}
           \bvtJ
           \strS'\,\vlscn{\strV}}}
 \vlread$,
and
$\vlstore{\vlsbr[\strR;\strU]}
 \bvtInfer{\bvtPder'''}{\ \bvtJ \vlread}$. 
The proof $\bvtPder'''$ is the one we are looking for. 
To get the derivation we are looking for, we fix $\strV$ such that
$\strFN{\strV}\cap\strBN{\strR}=\emptyset$. This allows to use $\bvtDder$
as follows:
{\small \[
\vlderivation                                {
\vlde{\bvtPder''}{}
     {\vlsbr(\strS'\vlscn{\strV};\strP)}    {
\vliq{\eqref{align:unit-co}}{}
     {\vlsbr(\strS'\vlscn{\strV};\vlone)}  {
\vlde{\bvtDder}{}
     {\strS'\vlscn{\strV}}                {
\vlhy{\vlfo{\vec{\atmb}}
           {\vlsbr[\strV;\strU]}}         }}}}
\] }
\par
As a \emph{third case}, 
let $\strS\vlhole\approx\vlfo{\atmb}{\strS'\vlhole}$
with $\atmb\in\strFN{\strS'\vlhole}$. Otherwise it would be meaningless assuming to have
$\strS\vlhole$ with such a form.
So, the assumption is 
$\bvtInfer{\bvtPder}{\ \bvtJ \vlfo{\atmb}{\strS'\,\vlscn{\strR}}}$. 
Point~\ref{enum:Derivability of subformulas-fo}
of Proposition~\ref{proposition:Derivability of structures in BVT} implies
$\bvtInfer{\bvtPder'}{\ \bvtJ \strS'\,\vlscn{\strR}}$.
So, $\Size{\strS'\,\vlscn{\strR}} < 
     \Size{\vlfo{\atmb}{\strS'\,\vlscn{\strR}}}$
implies the inductive hypothesis holds on $\bvtPder'$.
There are $\strU$, and $\vec{\atmb}$ \ST, for every
$\strV$ with $\strFN{\strV}\cap\strBN{\strR}=\emptyset$, both
$\vlstore{
 \bvtInfer{\bvtDder}
          {\vlfo{\vec{\atmb}}
                {\vlsbr[\strV;\strU]}
           \bvtJudGen{}{}
           \strS'\,\vlscn{\strV}}}
 \vlread$,
and 
$\vlstore{\vlsbr[\strR;\strU]}
 \bvtInfer{\bvtPder'''}{\ \bvtJ \vlread}$.
The proof $\bvtPder'''$ is the one we are looking for.
To get the derivation we are looking for, we fix $\strV$ such that
$\strFN{\strV}\cap\strBN{\strR}=\emptyset$. This allows to use $\bvtDder$
as follows:

{\small \[
\vlderivation                     {
\vlde{\bvtDder}{}
     {\vlfo{\atmb}
           {\strS'\vlscn{\strV}}}{
\vlhy{\vlfo{\vec{\atmb}}
           {\vlsbr[\strV;\strU]}}}}
\] }
\par
As a \emph{fourth case}, let
$\strS\vlhole\approx\vlsbr[<\strS'\vlhole;\strP'>;\strP]$.
The assumption is
$\vlstore{\vlsbr[<\strS'\,\vlscn\strR;\strP'>;\strP]}
 \bvtInfer{\bvtPder}{\ \bvtJ{\vlread}}$.
Shallow splitting implies the existence of
$\strP_1, \strP_2$ \ST\
$\vlstore{\vlsbr<\strP_1;\strP_2>}
 \bvtInfer{\bvtDder}{\vlread \bvtJ \strP}$, and
$\vlstore{\vlsbr[\strS'\,\vlscn\strR;\strP_1]}
 \bvtInfer{\bvtPder_1}{\ \bvtJ {\vlread}}$, and
$\vlstore{\vlsbr[\strP';\strP_2]}
 \bvtInfer{\bvtPder_2}{\ \bvtJ {\vlread}}$.
The relation
$\vlstore{\vlsbr[\strS'\,\vlscn\strR;\strP_1]}
 \Size{{\vlread}} <
 \vlstore{\vlsbr[<\strS'\,\vlscn\strR;\strP'>;\strP]}
 \Size{{\vlread}}$, which holds also
thanks to
$\Size{\strP_1}<\Size{\strP}$,
implies the inductive hypothesis holds on $\bvtPder_1$.
There are $\strU$, and  $\vec{\atmb}$ \ST, for every
$\strV$
with $\strFN{\strV}\cap\strBN{\strR}=\emptyset$, both
$\vlstore{\vlfo{\vec{\atmb}}
               {\vlsbr[\strV;\strU]}
          \bvtJ
          \vlsbr[\strS'\,\vlscn{\strV};\strP_1]
          }
 \bvtInfer{\bvtDder'}{\vlread}$, and
$\vlstore{\vlsbr[\strR;\strU]}
 \bvtInfer{\bvtPder'''}{\ \bvtJ \vlread}$. 
The proof $ \bvtPder''' $ is the one we are looking for. 
To get the derivation we are looking for, we fix $\strV$ such that
$\strFN{\strV}\cap\strBN{\strR}=\emptyset$. This allows to use $ \bvtDder' $
as follows:
{\small \[
\vlderivation                                                {
\vlde{\bvtDder}{}
     {\vlsbr[<\strS'\vlscn{\strV};\strP'>;\strP]}            {
\vlin{\bvtseqdrule}{}
     {\vlsbr[<\strS'\vlscn{\strV};\strP'>;<\strP_1;\strP_2>]}{
\vlde{\bvtPder_2}{}
     {\vlsbr<[\strS'\vlscn{\strV};\strP_1];[\strP';\strP_2]>}{
\vliq{\eqref{align:unit-seq}}{}
     {\vlsbr<[\strS'\vlscn{\strV};\strP_1];\vlone>}          {
\vlde{\bvtDder'}{}
     {\vlsbr[\strS'\vlscn{\strV};\strP_1]}                   {
\vlhy{\vlfo{\vec{\atmb}}
           {\vlsbr[\strV;\strU]}}                       }}}}}}
\]
}
\par
As a \emph{fifth case}, let
$\strS\vlhole\approx\vlsbr[(\strS'\,\vlscn\strR;\strP');\strP]$.
The assumption is
$\vlstore{\vlsbr[(\strS'\,\vlscn\strR;\strP');\strP]}
\bvtInfer{\bvtPder}{\ \bvtJ{\vlread}}$.
Shallow splitting implies the existence of
$\strP_1, \strP_2$ \ST\
$\vlstore{\vlsbr[\strP_1;\strP_2]}
 \bvtInfer{\bvtDder}{{\vlread} \bvtJ \strP}$, and
$\vlstore{\vlsbr[\strS'\,\vlscn\strR;\strP_1]}
 \bvtInfer{\bvtPder_1}{\ \bvtJ {\vlread}}$, and
$\vlstore{\vlsbr[\strP';\strP_2]}
 \bvtInfer{\bvtPder_2}{\ \bvtJ {\vlread}}$.
The relation
$\vlstore{\vlsbr[\strS'\,\vlscn\strR;\strP_1]}
 \Size{{\vlread}} <
 \vlstore{\vlsbr[(\strS'\,\vlscn\strR;\strP');\strP]}
 \Size{{\vlread}}$, which holds also thanks to
$\Size{\strP_1}<\Size{\strP}$,
implies the inductive hypothesis holds on $\bvtPder_1$.
There are $\strU$, and $\vec{\atmb}$ \ST, for every
$\strV$ with $\strFN{\strV}\cap\strBN{\strR}=\emptyset$, we have
$\vlstore{\vlfo{\vec{\atmb}}
               {\vlsbr[\strV;\strU]}
           \bvtJ {\vlsbr[\strS'\vlscn{\strV};\strP_1]}}
 \bvtInfer{\bvtDder'}{\vlread}$,
and 
$\vlstore{\vlsbr[\strR;\strU]}
 \bvtInfer{\bvtPder'''}{\ \bvtJ \vlread}$. The proof $\bvtPder'''$ is the one we are looking for.
To get the derivation we are looking for, we fix $\strV$ such that
$\strFN{\strV}\cap\strBN{\strR}=\emptyset$. This allows to use $ \bvtDder' $
as follows:
{\small \[
\vlderivation                                            {
\vlde{\bvtDder}{}
     {\vlsbr[(\strS'\vlscn{\strV};\strP');\strP]}             {
\vliq{\eqref{align:assoc-pa}}{}
     {\vlsbr[(\strS'\vlscn{\strV};\strP');[\strP_1;\strP_2]]} {
\vlin{\bvtswirule}{}
     {\vlsbr[[(\strS'\vlscn{\strV};\strP');\strP_1];\strP_2]} {
\vliq{\eqref{align:symm-co}
     ,\bvtswirule}{}
     {\vlsbr[([\strS'\vlscn{\strV};\strP_1];\strP');\strP_2]} {
\vlde{\bvtPder_2}{}
     {\vlsbr([\strP';\strP_2];[\strS'\vlscn{\strV};\strP_1])} {
\vliq{\eqref{align:unit-co}}{}
     {\vlsbr(\vlone;[\strS'\vlscn{\strV};\strP_1])}           {
\vlde{\bvtDder'}{}
     {\vlsbr[\strS'\vlscn{\strV};\strP_1]}                    {
\vlhy{\vlfo{\vec{\atmb}}
           {\vlsbr[\strV;\strU]}}                      }}}}}}}}
\] }
\par
As a \emph{sixth case}, let
$\vlstore{\vlsbr[\vlfo{\atma}{\strS'\vlhole};\strP]}
 \strS\vlhole\approx\vlread$ with $\atma\in\strBN{\strS'\,\vlscn{\strR}}$.
Otherwise, it would be meaningless to assume $ \strS\vlhole $ as such.
The assumption is
$\vlstore{\vlsbr[\vlfo{\atma}{\strS'\vlscn\strR};\strP]}
 \bvtInfer{\bvtPder}{\ \bvtJ
 {\vlread}}$.
Shallow splitting implies the existence of
$\strP'$ \ST\
$\bvtInfer{\bvtDder}{\vlex{\atma}{\strP'} \bvtJ \strP}$, and
$\vlstore{\vlsbr[\strS'\,\vlscn\strR;\strP']}
 \bvtInfer{\bvtPder'}{\ \bvtJ {\vlread}}$.
The relation
$\vlstore{\vlsbr[\strS'\,\vlscn\strR;\strP']}
 \Size{{\vlread}} <
 \vlstore{\vlsbr[\vlfo{\atma}{\strS'\,\vlscn\strR};\strP]}
 \Size{{\vlread}}$, which holds also because $\atma\in\strFN{\strS'\,\vlscn{\strR}}$,
implies that the inductive hypothesis on $\bvtPder'$ is true.
There are $\strU$, and $\vec{\atmb}$ \ST, for every $\strV$ with $\strFN{\strV}\cap\strBN{\strR}=\emptyset$, both
$\vlstore{\vlfo{\vec{\atmb}}
               {\vlsbr[\strV;\strU]}
          \bvtJ {\vlsbr[\strS'\,\vlscn{\strV};\strP']}}
 \bvtInfer{\bvtDder'}{\vlread}$,
and 
$\vlstore{\vlsbr[\strR;\strU]}
 \bvtInfer{\bvtPder''}{\ \bvtJ \vlread}$. 
The proof $\bvtPder''$ is the one we are looking for.
To get the derivation we are looking for, we fix $\strV$ such that
$\strFN{\strV}\cap\strBN{\strR}=\emptyset$. This allows to use $ \bvtDder' $
as follows:
{\small \[
\vlderivation                                                            {
\vlde{\bvtDder}{}
     {\vlsbr[\vlfo{\atma}{\strS'\vlscn{\strV}};\strP]}                  {
\vlin{\bvtrdrule}{}
     {\vlsbr[\vlfo{\atma}{\strS'\vlscn{\strV}};\vlfo{\atma}{\strP'}]}  {
\vlde{\bvtDder'}{}
     {\vlfo{\atma}{\vlsbr[\strS'\vlscn{\strV};\strP']}}               {
\vlhy{\vlfo{\atmb_1, \ldots, \atmb_n,\atma}
           {\vlsbr[\strV;\strU]}}                                     }}}}
\] }

\par
As a \emph{seventh case}, let
$\vlstore{\vlsbr[\strS'\vlhole;\vlfo{\atma}{\strP}]}
 \strS\vlhole\approx\vlread$ with $\atma\in\strBN{\vlfo{\atma}{\strP}}$.
Also, without loss of generality, can always choose $\atma$ \ST 
$ \atma\not\in\strFN {\strS'\,\vlscn{\strR}} $.
The assumption is
$\vlstore{\vlsbr[\strS'\,\vlscn{\strR};\vlfo{\atma}{\strP}]}
 \bvtInfer{\bvtPder}{\ \bvtJ
 {\vlread}}$.
Shallow splitting implies the existence of
$\strP'$ \ST\
$\bvtInfer{\bvtDder}{\vlex{\atma}{\strP'} \bvtJ \strP}$, and
$\vlstore{\vlsbr[\strS'\,\vlscn\strR;\strP']}
 \bvtInfer{\bvtPder'}{\ \bvtJ {\vlread}}$.
The relation
$\vlstore{\vlsbr[\strS'\,\vlscn\strR;\strP']}
 \Size{{\vlread}} <
 \vlstore{\vlsbr[\strS'\,\vlscn\strR;\vlfo{\atma}{\strP]}}
 \Size{{\vlread}}$, which holds also because $\atma\in\strBN{\vlfo{\atma}{\strP}}$,
implies that the inductive hypothesis on $\bvtPder'$ is true.
There are $\strU$, and $\vec{\atmb}$ \ST, for every $\strV$ with $\strFN{\strV}\cap\strBN{\strR}=\emptyset$, both
$\vlstore{\vlfo{\vec{\atmb}}
               {\vlsbr[\strV;\strU]}
          \bvtJ {\vlsbr[\strS'\,\vlscn{\strV};\strP']}}
 \bvtInfer{\bvtDder'}{\vlread}$,
and 
$\vlstore{\vlsbr[\strR;\strU]}
 \bvtInfer{\bvtPder''}{\ \bvtJ \vlread}$. 
The proof $\bvtPder''$ is the one we are looking for.
To get the derivation we are looking for, we fix $\strV$ such that
$\strFN{\strV}\cap\strBN{\strR}=\emptyset$. This allows to use $ \bvtDder' $
as follows:
{\small \[
\vlderivation                                                          {
\vlde{\bvtDder}{}
     {\vlsbr[\strS'\,\vlscn{\strV};\vlfo{\atma}{\strP}]}   {
\vliq{\eqref{align:alpha-intro}}{}
     {\vlsbr[\strS'\,\vlscn{\strV};\vlfo{\atma}{\strP'}]}  {
\vlin{\bvtrdrule}{}
     {\vlsbr[\vlfo{\atma}{\strS'\,\vlscn{\strV}};\vlfo{\atma}{\strP'}]}  {
\vlde{\bvtDder'}{}
     {\vlfo{\atma}{\vlsbr[\strS'\,\vlscn{\strV};\strP']}}                {
\vlhy{\vlfo{\atmb_1, \ldots, \atmb_n,\atma}
           {\vlsbr[\strV;\strU]}}                                      }}}}}
\] }
We remark that $\eqref{align:alpha-intro}$ applies thanks to 
$ \atma\not\in\strFN{\strS'\,\vlscn{\strR}}$.
\section{Proof of \textit{Splitting} (Theorem~\ref{theorem:Splitting-ALT},
page~\pageref{theorem:Splitting-ALT})}
\label{section:Proof of theorem:Splitting}
We obtain the proof of the three statements by composing Context Reduction
(Proposition~\ref{proposition:Context Reduction ALTERNATIVE}), and Shallow Splitting
(Proposition~\ref{proposition:Shallow Splitting}) in this order.
We develop the details of Points~\ref{enum:Splitting-seq},
and~\ref{enum:Splitting-fo-ex}. The proof of Point~\ref{enum:Splitting-copar} is analogous
to the one of~\ref{enum:Splitting-seq}.
\begin{description}
\item[Point~\ref{enum:Splitting-seq}.]
Context Reduction (Proposition~\ref{proposition:Context Reduction ALTERNATIVE})
applies to $\bvtPder$. So, there are $\strU$, and $\vec{\atmb}$ \ST, for every $\strV$, with
$\vlstore{\strFN{\strV}\cap\strBN{\vlsbr<\strR;\strT>}=\emptyset}
 \vlread$, there exist
$\vlstore{\vlsbr[\strV;\strU]}
 \bvtInfer{\bvtDder}
          {\vlfo{\vec{\atmb}}{\vlread} \bvtJ \strS\vlscn{\strV}}$,
and $\vlstore{\vlsbr[<\strR;\strT>;\strU]}
     \bvtInfer{\bvtQder}{\ \bvtJ {\vlread}}$. Shallow
Splitting (Proposition~\ref{proposition:Shallow Splitting}) applies to
$\bvtQder$. So,
$\vlstore{\vlsbr<\strK_1;\strK_2>}
 \bvtInfer{\bvtEder}{{\vlread} \bvtJ \strU}$, and
$\vlstore{\vlsbr[\strR;\strK_1]}
 \bvtInfer{\bvtQder_1}{\ \bvtJ {\vlread}}$, and
$\vlstore{\vlsbr[\strT;\strK_2]}
 \bvtInfer{\bvtQder_2}{\ \bvtJ {\vlread}}$, for some
$\strK_1, \strK_2$. Both $\bvtQder_1$, and $\bvtQder_2$ are the two proofs we are
looking for. The derivation is:
{\small \[
\vlderivation                                 {
\vlde{\bvtDder}{}
     {\strS\vlscn{\strV}                   } {
\vlde{\bvtEder}{}
     {\vlfo{\vec{\atmb}}
           {\vlsbr[\strV;\strU]}           }{
\vlhy{\vlfo{\vec{\atmb}}
           {\vlsbr[\strV;<\strK_1;\strK_2>]}}}}}
\]}
\item[Point~\ref{enum:Splitting-fo-ex}.]
Context Reduction (Proposition~\ref{proposition:Context Reduction ALTERNATIVE})
applies to $\bvtPder$. So, there are $\strU$, and $\vec{\atmb}$ \ST, for every $\strV$ with
$\vlstore{\strFN{\strV}\cap\strBN{\vlfo{\atma}{\strR}}=\emptyset}
 \vlread$, there exist
$\vlstore{\vlsbr[\strV;\strU]}
 \bvtInfer{\bvtDder}{\vlfo{\vec{\atmb}}
                          {\vlread} \bvtJ \strS\vlscn{\strV}}$,
and $\vlstore{\vlsbr[\vlfo{\atma}{\strR};\strU]}
     \bvtInfer{\bvtQder}{\ \bvtJ {\vlread}}$.
Shallow Splitting (Proposition~\ref{proposition:Shallow Splitting}) applies to
$\bvtQder$.
So, $\bvtInfer{\bvtEder}{\vlfo{\atma}{\strK} \bvtJ \strU}$, and
$\vlstore{\vlsbr[\strR;\strK]}
 \bvtInfer{\bvtQder'}{\ \bvtJ {\vlread}}$, for some
$\strK$. So, $\bvtQder'$ is the proof we are looking for.
The derivation is:
{\small \[
\vlderivation                              {
\vlde{\bvtDder}{}
     {\strS\vlscn{\strV}              }   {
\vlde{\bvtEder}{}
     {\vlfo{\vec{\atmb}}
           {\vlsbr[\strV;\strU]}      }  {
\vliq{\eqref{align:alpha-intro}}{}
     {\vlfo{\vec{\atmb}}
           {\vlsbr[\strV
                  ;\vlfo{\atma}{\strK}
                  ]}                  } {
\vlin{\bvtrdrule}{}
     {\vlfo{\vec{\atmb}}
           {\vlsbr[\vlfo{\atma}{\strV}
                  ;\vlfo{\atma}{\strK}
                  ]                   }}{
\vlhy{\vlfo{\atma,\vec{\atmb}}
           {\vlsbr[\strV;\strK]}}}}}}}
\]}
The step \eqref{align:alpha-intro} applies thanks to the assumption that
$\vlstore{\strFN{\strV}\cap\strBN{\vlfo{\atma}{\strR}}=\emptyset}
 \vlread$, which implies $\atma\not\in\strFN{\strV}$.
\end{description}

\section{Proof of \textit{Admissibility of the up fragment}
(Theorem~\ref{theorem:Admissibility of the up fragment},
page~\pageref{theorem:Admissibility of the up fragment})}
\label{section:Proof of theorem:Admissibility of the up fragment}
As a \emph{first case} we show that ${\bvtatiurulein}$ is admissible
for $\BVT$. So, we start by assuming:
{\small \[
\vlderivation                           {
\vlin{\bvtatiurule}{}
     {\strS\vlscn{\vlone}}             {
\vlpr{\bvtPder'}{}
     {\strS\vlsbr(\atma;\natma)}}}
\] }
Point~\ref{enum:Splitting-copar} of Splitting (Theorem~\ref{theorem:Splitting-ALT}) applies to $\bvtPder'$, whose conclusion is $\strS\vlsbr(\atma;\natma)$.
There are $\strK_1, \strK_2$, and $\vec{\atmb}$ \ST,
for every $\strV$ with
$\vlstore{\strFN{\strV}\cap\strBN{\vlsbr(\atma;\natma)}=\emptyset}
 \vlread$,
there exist
$\vlstore{\vlsbr[\strV;[\strK_1;\strK_2]]}
 \bvtInfer{\bvtDder}
          {\vlfo{\vec{\atmb}}{{\vlread}} \bvtJ \strS\vlscn{\strV}}$, and
$\vlstore{\vlsbr[\atma;\strK_1]}
 \bvtInfer{\bvtPder_1}{\ \bvtJ {\vlread}} $, and
$\vlstore{\vlsbr[\natma;\strK_2]}
 \bvtInfer{\bvtPder_2}{\ \bvtJ {\vlread}}$.
Shallow splitting (Proposition~\ref{proposition:Shallow Splitting}) on
$\bvtPder_1$, and $\bvtPder_2$ implies
$\bvtInfer{\bvtEder_1}{\natma \bvtJ{\strK_1}}$, and
$\bvtInfer{\bvtEder_2}{\atma \bvtJ {\strK_2}}$.
To build the following proof with the same conclusion as $\bvtPder$, but
without its bottommost instance of $\bvtatiurulein$ it is enough to observe
that among all the possible instances of $\strV$ there is $\vlone$, because
$\vlstore{\strFN{\vlone}\cap\strBN{\vlsbr(\atma;\natma)}=\emptyset}
 \vlread$. So, we can prove:
{\small \[
\vlderivation                                                 {
\vlde{\bvtDder'}{}
     {\strS\vlscn{\vlone}}                                    {
\vliq{\eqref{align:unit-co}}{}
     {\vlfo{\vec{\atmb}}
           {\vlsbr[\vlone;[\strK_1;\strK_2]]}}         {
\vlde{\bvtEder_1}{}
     {\vlfo{\vec{\atmb}}
           {\vlsbr[\strK_1;\strK_2]}}                  {
\vlde{\bvtEder_2}{}
     {\vlfo{\vec{\atmb}}
           {\vlsbr[\natma;\strK_2]}}                  {
\vlin{\bvtatidrule}{}
     {\vlfo{\vec{\atmb}}
           {\vlsbr[\natma;\atma]}}                     {
\vliq{\eqref{align:alpha-intro}}{}{\vlfo{\atmb}{\vlone}}{
\vlhy{\vlone}                                                 }}}}}}}
\] }
where $\bvtDder'$ is $\bvtDder$ with $\strV$ instantiated as $\vlone$.
\par
As a \emph{second case} we show that ${\bvtsequrulein}$ is admissible
for $\BVT$. So, we start by assuming:
{\small \[
\vlderivation                                    {
\vlin{\bvtsequrule}{}
     {\strS\vlsbr<(\strR;\strT);(\strU;\strV)>} {
\vlpr{\bvtPder'}{}
     {\strS\vlsbr(<\strR;\strU>;<\strT;\strV>) }}}
\] }
Point~\ref{enum:Splitting-seq} of Splitting (Theorem~\ref{theorem:Splitting-ALT}) applies to $\bvtPder$ --- beware, not $ \bvtPder' $ ---, whose conclusion is $\strS\vlsbr<(\strR;\strT);(\strU;\strV)>$.
There are
$\strK_1, \strK_2$, and $\vec{\atmb}$ \ST,
for every $\strV'$ with 
$\vlstore{\strFN{\strV'}\cap
          \strBN{\strS\vlsbr<(\strR;\strT);(\strU;\strV)>}=\emptyset}
 \vlread$,
there exist
$\vlstore{\vlsbr[\strV';<\strK_1;\strK_2>]}
 \bvtInfer{\bvtDder}
 {\vlfo{\vec{\atmb}}
       {\vlread}
  \bvtJ \strS\vlscn{\strV'}}$, and
$\vlstore{\vlsbr[(\strR;\strT);\strK_1]}
 \bvtInfer{\bvtPder_1}{\ \bvtJ {\vlread}}$, and
$\vlstore{\vlsbr[(\strU;\strV);\strK_2]}
 \bvtInfer{\bvtPder_2}{\ \bvtJ {\vlread}}$.
Shallow splitting (Proposition~\ref{proposition:Shallow Splitting}) on both
$\bvtPder_1$, and $\bvtPder_2$ implies
$\vlstore{\vlsbr[\strK_\strR;\strK_\strT]}
 \bvtInfer{\bvtEder   }{\vlread \bvtJ {\strK_1}} $, and
$\vlstore{\vlsbr{[\strR;\strK_\strR]}}
 \bvtInfer{\bvtQder_1}{\ \bvtJ {\vlread}}$, and
$\vlstore{\vlsbr{[\strT;\strK_\strT]}}
 \bvtInfer{\bvtQder_2}{\ \bvtJ {\vlread}}$, and
$\vlstore{\vlsbr[\strK_\strU;\strK_\strV] \bvtJ {\strK_2}}
 \bvtInfer{\bvtEder' }{\vlread}$, and
$\vlstore{\vlsbr{[\strU;\strK_\strU]}}
 \bvtInfer{\bvtQder'_1}{\ \bvtJ {\vlread}}$, and
$\vlstore{\vlsbr{[\strV;\strK_\strV]}}
 \bvtInfer{\bvtQder'_2}{\ \bvtJ {\vlread}}$.
To build the following proof with the same conclusion as $\bvtPder$, but without its
bottommost instance of ${\bvtsequrulein}$, it is enough to observe that one of the
possible instances of $\strV'$ is
$\vlsbr<(\strR;\strT);(\strU;\strV)>$ because,
thanks to \eqref{align:alpha-varsub}, we can always assume
$\vlstore{\strFN{\vlsbr<(\strR;\strT);(\strU;\strV)>}
 \cap
 \strBN{\vlsbr<(\strR;\strT);(\strU;\strV)>}}
 \vlread=\emptyset$:
{\small \[
\vlderivation                                            {
\vlde{\bvtDder'}{}
     {\strS\vlsbr<(\strR;\strT);(\strU;\strV)>}          {
\vlde{\bvtEder}{}
     {\vlfo{\vec{\atmb}}
           {\vlsbr[<(\strR;\strT);(\strU;\strV)>
                  ;[\strK_1
                   ;\strK_2]]  }}                 {
\vlde{\bvtEder'}{}
     {\vlfo{\vec{\atmb}}{\vlsbr[<(\strR;\strT);(\strU;\strV)>
                         ;[\strK_\strR;\strK_\strT
                          ;\strK_2]]  }}                 {
\vlin{\bvtpmixrule}{}
     {\vlfo{\vec{\atmb}}
           {\vlsbr[<(\strR;\strT);(\strU;\strV)>
                  ;[\strK_\strR;\strK_\strT
                   ;\strK_\strU;\strK_\strV]]}} {
\vlin{\bvtseqdrule}{}
     {\vlfo{\vec{\atmb}}
                {\vlsbr[<(\strR;\strT);(\strU;\strV)>
                       ;<[\strK_\strR;\strK_\strT]
                        ;[\strK_\strU;\strK_\strV]>]}} {
\vlin{\bvtswirule^{2}}{}
     {\vlfo{\vec{\atmb}}
           {\vlsbr<[(\strR;\strT);\strK_\strR;\strK_\strT]
                  ;[(\strU;\strV)
                   ;\strK_\strU;\strK_\strV]>}}   {
\vlde{\bvtQder_1}{}
     {\vlfo{\vec{\atmb}}
                {\vlsbr<[([\strR;\strK_\strR];\strT);\strK_\strT]
                       ;[([\strU;\strK_\strU]
                        ;\strV);\strK_\strV]>}}   {
\vlde{\bvtQder'_1}{}
     {\vlfo{\vec{\atmb}}
           {\vlsbr<[\strT;\strK_\strT]
                  ;[([\strU;\strK_\strU]
                   ;\strV);\strK_\strV]>}}   {
\vlde{\bvtQder''_1}{}
     {\vlfo{\vec{\atmb}}
           {\vlsbr<[\strT;\strK_\strT]
                  ;[\strV;\strK_\strV]>}}   {
\vlde{\bvtQder''_2}{}
     {\vlfo{\vec{\atmb}}
                {\vlsbr[\strV;\strK_\strV]}}         {
\vliq{\eqref{align:alpha-intro}}{}
     {\vlfo{\vec{\atmb}}
                {\vlone}}         {
\vlhy{\vlone}                                }}}}}}}}}}}}
\] }
where $\bvtDder'$ is $\bvtDder$ with $\strV'$ instantiated as
$\vlsbr<(\strR;\strT);(\strU;\strV)>$.
\par
As a \emph{third case} we show that ${\bvtrurulein}$ is admissible for
$\BVT$. So, we start by assuming:
{\small \[
\vlderivation                                               {
\vlin{\bvtrurule}{}
     {\strS\vlfo{\atma}{\vlsbr(\strR;\strT)}}              {
\vlpr{\bvtPder'}{}
     {\strS\vlsbr(\vlfo{\atma}{\strR};\vlfo{\atma}{\strT})}}}
\] }
Point~\ref{enum:Splitting-fo-ex} of Splitting (Theorem~\ref{theorem:Splitting-ALT}) applies to $\bvtPder$ --- beware, not
$ \bvtPder' $ ---, whose conclusion is 
$\vlstore{\vlsbr(\strR;\strT)}
 \strS\vlfo{\atma}{\vlread}$.
There is $\strK$, and $\vec{\atmb}$ \ST,
for every $\strV$ with
$\vlstore{\strFN{\strV}
          \cap
          \strBN{\strS\vlfo{\atma}{\vlsbr(\strR;\strT)}} =
          \emptyset}
 \vlread$,
there exist
$\vlstore{\vlfo{\atma,\vec{\atmb}}
               {\vlsbr[\strV;\strK]}}
 \bvtInfer{\bvtDder}
          {\vlread \bvtJ \strS\vlscn{\strV}}$, and
$\vlstore{\vlsbr[(\strR;\strT);\strK]}
 \bvtInfer{\bvtPder_1}{\ \bvtJ {\vlread}}$.
Shallow splitting (Proposition~\ref{proposition:Shallow Splitting}) on
$\bvtPder_1$ implies
$\vlstore{\vlsbr[\strK_\strR;\strK_\strT]}
 \bvtInfer{\bvtEder}{\vlread \bvtJ \strK}$, and
$\vlstore{\vlsbr{[\strR;\strK_\strR]}}
 \bvtInfer{\bvtQder_1 }{\ \bvtJ {\vlread}}$, and
$\vlstore{\vlsbr{[\strT;\strK_\strT]}}
 \bvtInfer{\bvtQder_2}{\ \bvtJ {\vlread}}$.
To build the following proof with the same conclusion as $\bvtPder$, but
without its bottommost instance of ${\bvtrurulein}$ it is enough to observe that
one of the possible instances of $ \strV $ is 
$\vlstore{\vlsbr(\strR;\strT)}
 \strS\vlfo{\atma}{\vlread} $
such that 
$\vlstore{\strFN{\vlfo{\atma}{\vlsbr(\strR;\strT)}}
          \cap
          \strBN{\vlfo{\atma}{\vlsbr(\strR;\strT)}} =
          \emptyset}
 \vlread$:
{\small \[
\vlderivation                                     {
\vlde{\bvtDder}{}
     {\strS\vlfo{\atma}
                {\vlsbr(\strR;\strT)}}            {
\vlde{\bvtEder}{}
     {\vlfo{\atma.\vec{\atmb}}
           {\vlsbr[(\strR;\strT)
                  ;\strK]
            }}                                     {
\vlin{\bvtswirule}{}
     {\vlfo{\atma,\vec{\atmb}}
           {\vlsbr[(\strR;\strT)
                   ;\strK_\strR
                   ;\strK_\strT]]
            }}                                     {
\vlin{\eqref{align:symm-co}
     ,\bvtswirule}{}
     {\vlfo{\atma,\vec{\atmb}}
           {\vlsbr[([\strR;\strK_\strR];\strT)
                         ;\strK_\strT]}}           {
\vlde{\bvtQder_2}{}
     {\vlfo{\atma,\vec{\atmb}}
           {\vlsbr([\strT;\strK_\strT]
                        ;[\strR;\strK_\strR])}}    {
\vlde{\bvtQder_1}{}
     {\vlfo{\atma,\vec{\atmb}}
           {\vlsbr(\vlone
                        ;[\strR;\strK_\strR])}}    {
\vliq{\eqref{align:unit-co}
     ,\eqref{align:alpha-intro}}{}
     {\vlfo{\atma,\vec{\atmb}}
           {\vlsbr(\vlone;\vlone)}} {
\vlhy{\vlone}                        }}}}}}}}
\] }

\section{Proof that \textit{$\bvttradruleinp$ is derivable in $\BVT$}
(Lemma~\ref{lemma:Deriving substitution},
page~\pageref{lemma:Deriving substitution})}
\label{section:lemma:Simulating transition prime}

We proceed by induction on the size
$\vlstore{\mapLcToDi{\llcxM}{\atmo}}
 \Size{\vlread}$,
of
$\mapLcToDi{\llcxM}{\atmo}$, that occurs in the conclusion of
$\bvttradruleinp$, proceeding by cases on the form of $\llcxM$.
\par
The first base case is $\llcxM\equiv\llcxX$.
{\small
\[
\vlderivation                          {
\vlin{\bvttradrule}{}
     {\vlsbr[\mapLcToDi{\llcxX}{\atmr}
            ;<\atmr;\vlne\atmo>]
     \equiv
     \vlsbr[<\llcxX;\vlne\atmr>
            ;<\atmr;\vlne\atmo>]       }{
\vlhy{\mapLcToDi{\llcxX}{\llcxo}
      \equiv
      \vlsbr<\llcxX;\llcxno>           }}}
\]
}
\par
The second base case is $\llcxM\equiv\llcxA{\llcxM'}{\llcxM''}$.
{\small
\[
\vlderivation                                            {
\vliq{\eqref{align:symm-pa}
     ,\bvtrdrule
     }{}
     {\vlsbr[\mapLcToDi{\llcxA{\llcxM'}{\llcxM''}}{\atmr}
                  ;<\atmr;\vlne\atmo>
                  ]
            \equiv
      \vlsbr[\vlex{\atmp}
                  {\vlsbr[\mapLcToDi{\llcxM'}{\atmp}
                         ;\vlex{\atmq}{\mapLcToDi{\llcxM''}{\atmq}}
                         ;<\atmp;\vlne\atmr>]
                  }
             ;<\atmr;\vlne\atmo>
            ]}{
\vlin{\bvttradrule
     }{}
     {\vlex{\atmp}
           {\vlsbr[\mapLcToDi{\llcxM'}{\atmp}
                   ;\vlex{\atmq}{\mapLcToDi{\llcxM''}{\atmq}}
                   ;<\atmp;\vlne\atmr>
                   ;<\atmr;\vlne\atmo>
                  ]
           }
     }                                                                {
\vlhy{\mapLcToDi{\llcxA{\llcxM'}{\llcxM''}}
                      {\llcxo}
      \equiv
      \vlex{\atmp}
           {\vlsbr[\mapLcToDi{\llcxM'}{\atmp}
                   ;\vlex{\atmq}{\mapLcToDi{\llcxM''}{\atmq}}
                   ;<\atmp;\vlne\atmo>
                  ]
           }
      }}}}
\]
}
\par
The unique inductive case is with $\llcxM\equiv\llcxF{\llcxY}{\llcxM'}$ that,
without
loss of generality, can have $\llcxY\neq\llcxX$.
{\small
\[
\vlderivation                                            {
\vliq{\eqref{align:alpha-intro}
     ,\bvtrdrule}{}
     {\vlsbr[\mapLcToDi{\llcxF{\llcxY}{\llcxM'}}{\atmr}
                  ;<\atmr;\vlne\atmo>
                  ]
            \equiv
      \vlsbr[\vlfo{\llcxY}
                  {\mapLcToDi{\llcxM'}{\atmr}}
            ;<\atmr;\vlne\atmo>]}{
\vlin{\bvttradrulep}{}
     {\vlfo{\llcxY}
           {\vlsbr[\mapLcToDi{\llcxM'}{\atmr}
                  ;<\atmr;\vlne\atmo>]
           }                                             }{
\vlhy{\mapLcToDi{\llcxF{\llcxY}{\llcxM'}}
                {\llcxo}
      \equiv
      \vlfo{\llcxY}
           {\mapLcToDi{\llcxM'}{\llcxo}}
      }}}}
\]
}
where $\bvttradruleinp$ applies by induction because
$\vlstore
 {\mapLcToDi{\llcxM'}{\atmr}
 }
\Size{\vlread}$
$<
\vlstore
 {\mapLcToDi{\llcxF{\llcxY}{\llcxM'}}{\atmr}}
\Size{\vlread}$.
\section{Proof that \textit{$\bvtsubsrule$ is derivable in $\BVT$}
(Lemma~\ref{lemma:Deriving substitution},
page~\pageref{lemma:Deriving substitution})}
\label{section:Proof of lemma:Deriving substitution}
We proceed by induction on the size
$\vlstore{
 \vlsbr[\mapLcToDi{\llcxM}{\atmo}
             ;\mapLcToDi{\llcxN}{\llcxX}]}
 \Size{\vlread}$
of
$\vlsbr[\mapLcToDi{\llcxM}{\atmo}
      ;\mapLcToDi{\llcxN}{\llcxX}]$, that occurs in the conclusion of
$\bvtsubsrule$, proceeding by cases on the form of $\llcxM$.
\par
Let $\llcxM\equiv\llcxX$. We have three situations:
\begin{description}
\item[$\llcxN\equiv\llcxY$.]
{\small
\[
\vlderivation                                                  {
\vlin{\bvttradrule}{}
     {\vlsbr[\mapLcToDi{\llcxX}{\llcxo}
                  ;\mapLcToDi{\llcxY}{\llcxX}]
      \equiv
      \vlsbr[<\llcxX;\llcxno>
            ;<\llcxY;\vlne\llcxX>]}                     {
\vlhy{\mapLcToDi{\llcxX\subst{\llcxY}{\llcxX}}{\llcxo}
      \equiv
      \mapLcToDi{\llcxY}{\llcxo}
      \equiv
      \vlsbr<\llcxY;\llcxno>
     }                            }}
\]
}

\item[$\llcxN\equiv\llcxA{\llcxN'}{\llcxN''}$.]
{\small
\[
\vlderivation                                            {
\vliq{\eqref{align:assoc-pa}
     ,\bvtrdrule
     ,\eqref{align:alpha-intro}}{}
     {\vlsbr[\mapLcToDi{\llcxX}{\llcxo}
                  ;\mapLcToDi{\llcxA{\llcxN'}{\llcxN''}}
                              {\llcxX}]
            \equiv
      \vlsbr[<\llcxX;\llcxno>
            ;\vlex{\atmp}
                  {[\mapLcToDi{\llcxN'}
                              {\atmp}
                    ;\vlex{\atmq}
                          {\mapLcToDi{\llcxN''}
                                     {\atmq}
                          }
                    ;<\atmp;\vlne\llcxX>
                   ]}
            ]                                           }{
\vlin{\bvttradrule}{}
     {\vlex{\atmp}
           {\vlsbr[\mapLcToDi{\llcxN'}
                              {\atmp}
                    ;\vlex{\atmq}
                          {\mapLcToDi{\llcxN''}
                                     {\atmq}
                          }
                    ;<\llcxX;\llcxno>
                    ;<\atmp;\vlne\llcxX>
                   ]}                                   }{
\vlhy{\mapLcToDi{\llcxX\subst{\llcxA{\llcxN'}{\llcxN''}}
                             {\llcxX}}{\atmo}
      \equiv
      \mapLcToDi{\llcxA{\llcxN'}{\llcxN''}}{\atmo}
      \equiv
      \vlex{\atmp}
           {\vlsbr[\mapLcToDi{\llcxN'}
                              {\atmp}
                    ;\vlex{\atmq}
                          {\mapLcToDi{\llcxN''}
                                     {\atmq}
                          }
                    ;<\atmp;\llcxno>
                   ]}
      }}}}
\]
}

\item[$\llcxN\equiv\llcxF{\llcxY}{\llcxN'}$] that, without loss of generality, can
be $\llcxY\neq\llcxX$.
{\small
\[
\vlderivation                                            {
\vliq{\eqref{align:alpha-intro}
     ,\bvtrdrule}{}
     {\vlsbr[\mapLcToDi{\llcxX}{\llcxo}
                  ;\mapLcToDi{\llcxF{\llcxY}{\llcxN'}}
                              {\llcxX}]
            \equiv
      \vlsbr[<\llcxX;\vlne\llcxo>
            ;\vlfo{\llcxY}
                  {\mapLcToDi{\llcxN'}{\llcxX}}]}{
\vlin{\bvttradrulep}{}
     {\vlfo{\llcxY}
           {\vlsbr[<\llcxX;\vlne\llcxo>
                  ;\mapLcToDi{\llcxN'}{\llcxX}]
                 }}{
\vlhy{\mapLcToDi{\llcxX\subst{\llcxF{\llcxY}{\llcxN'}}{\llcxX}}{\llcxo}
      \equiv
      \mapLcToDi{\llcxF{\llcxY}{\llcxN'}}{\llcxo}
      \equiv
      \vlfo{\llcxY}
           {\mapLcToDi{\llcxN'}{\llcxo}}
      }}}}
\]
}
\end{description}

\noindent
Let $\llcxM\equiv\llcxF{\llcxY}{\llcxM'}$ that, without loss of generality,
can always be such that $\llcxY\neq\llcxX$.
{\small
\[
\vlderivation                                            {
\vliq{\eqref{align:alpha-intro}
     ,\bvtrdrule}{}
     {\vlsbr[\mapLcToDi{\llcxF{\llcxY}{\llcxM'}}{\llcxo}
                  ;\mapLcToDi{\llcxN}{\llcxX}
                  ]
            \equiv
      \vlsbr[\vlfo{\llcxY}
                  {\mapLcToDi{\llcxM'}{\atmo}}
            ;\mapLcToDi{\llcxN}{\llcxX}]}{
\vlin{\bvtsubsrule}{}
     {\vlfo{\llcxY}
           {\vlsbr[\mapLcToDi{\llcxM'}{\atmo}
                  ;\mapLcToDi{\llcxN'}{\llcxX}]
           }                                             }{
\vlhy{\mapLcToDi{\llcxF{\llcxY}{\llcxM'\subst{\llcxN'}{\llcxX}}}
                {\llcxo}
      \equiv
      \vlfo{\llcxY}
           {\mapLcToDi{\llcxM'\subst{\llcxN'}{\llcxX}}{\llcxo}}
      }}}}
\]
}
where $\bvtsubsrule$ applies by induction because
$\vlstore
 {\vlsbr[\mapLcToDi{\llcxM'}{\llcxo}
        ;\mapLcToDi{\llcxN'}{\llcxX}]
 }
\Size{\vlread}$
$<
\vlstore
 {\vlsbr[\mapLcToDi{\llcxF{\llcxY}{\llcxM'}}{\llcxo}
        ;\mapLcToDi{\llcxN'}
                    {\llcxX}]}
\Size{\vlread}$.

\noindent
Let $\llcxM\equiv\llcxA{\llcxM'}{\llcxM''}$ with $\llcxX\in\llcxFV{\llcxM'}$.
{\small
\[
\vlderivation                                            {
\vliq{\eqref{align:symm-pa}
     ,\bvtrdrule
     ,\eqref{align:assoc-pa}
     }{}
     {\vlsbr[\mapLcToDi{\llcxA{\llcxM'}{\llcxM''}}{\llcxo}
            ;\mapLcToDi{\llcxN}{\llcxX}
            ]
            \equiv
      \vlsbr[\vlex{\atmp}
                  {\vlsbr[\mapLcToDi{\llcxM'}{\atmp}
                         ;\vlex{\atmq}{\mapLcToDi{\llcxM''}{\atmq}}
                         ;<\atmp;\vlne\atmo>
                         ]
                  }
            ;\mapLcToDi{\llcxN}{\llcxX}
            ]}{
\vlin{\bvtsubsrule}{}
     {\vlex{\atmp}
           {\vlsbr[[\mapLcToDi{\llcxM'}{\atmp}
                    ;\mapLcToDi{\llcxN}{\llcxX}]
                  ;\vlex{\atmq}{\mapLcToDi{\llcxM''}{\atmq}}
                  ;<\atmp;\vlne\atmo>
                  ]
           }                                                      }{
\vlhy{\mapLcToDi{\llcxA{\llcxM'\subst{\llcxN}{\llcxX}}{\llcxM''}}
                      {\llcxo}
      \equiv
      \vlex{\atmp}
           {\vlsbr[\mapLcToDi{\llcxM'\subst{\llcxN}{\llcxX}}
                              {\atmp}
                  ;\vlex{\atmq}{\mapLcToDi{\llcxM''}{\atmq}}
                  ;<\atmp;\vlne\atmo>
                  ]
           }
      }}}}
\]
}
where $\bvtsubsrule$ can be applied by induction because
$\vlstore
 {\vlsbr[\mapLcToDi{\llcxM'}{\llcxp}
        ;\mapLcToDi{\llcxN}{\llcxX}]
 }
\Size{\vlread}$
$<
\vlstore
 {\vlsbr[\mapLcToDi{\llcxA{\llcxM'}{\llcxM''}}{\llcxo}
         ;\mapLcToDi{\llcxN}
                     {\llcxX}]}
\Size{\vlread}$.

\noindent
Let $\llcxM\equiv\llcxA{\llcxM'}{\llcxM''}$ with
$\llcxX\in\llcxFV{\llcxM''}$.
{\small
\[
\vlderivation                                            {
\vliq{\eqref{align:symm-pa}
     ,\bvtrdrule
     ,\eqref{align:assoc-pa}
     }{}
     {\vlsbr[\mapLcToDi{\llcxA{\llcxM'}{\llcxM''}}{\llcxo}
                  ;\mapLcToDi{\llcxN}{\llcxX}
                  ]
            \equiv
      \vlsbr[\vlex{\atmp}
                  {\vlsbr[\mapLcToDi{\llcxM'}{\atmp}
                         ;\vlex{\atmq}{\mapLcToDi{\llcxM''}{\atmq}}
                         ;<\atmp;\vlne\atmo>]
                  }
             ;\mapLcToDi{\llcxN}{\llcxX}
            ]}{
\vliq{\eqref{align:symm-pa}
     ,\bvtrdrule
     }{}
     {\vlex{\atmp}
           {\vlsbr[\mapLcToDi{\llcxM'}{\atmp}
                   ;[\vlex{\atmq}{\mapLcToDi{\llcxM''}{\atmq}}
                     ;\mapLcToDi{\llcxN}{\llcxX}]
                   ;<\atmp;\vlne\atmo>
                  ]
           }
     }                                                                {
\vlin{\bvtsubsrule}{}
     {\vlex{\atmp}
           {\vlsbr[\mapLcToDi{\llcxM'}{\atmp}
                   ;[\vlex{\atmq}{\mapLcToDi{\llcxM''}{\atmq}}
                     ;\mapLcToDi{\llcxN}{\llcxX}]
                   ;<\atmp;\vlne\atmo>
                  ]
           }
     }                                                                {
\vlhy{\mapLcToDi{\llcxA{\llcxM'}{\llcxM''\subst{\llcxN}{\llcxX}}}
                      {\llcxo}
      \equiv
      \vlex{\atmp}
           {\vlsbr[\mapLcToDi{\llcxM'}{\atmp}
                   ;\vlex{\atmq}{\mapLcToDi{\llcxM''\subst{\llcxN}
                                                           {\llcxX}}
                                            {\atmq}}
                   ;<\atmp;\vlne\atmo>
                  ]
           }
      }}}}}
\]
}
where $\bvtsubsrule$ applies by induction as
$\vlstore
 {\vlsbr[\mapLcToDi{\llcxM''}{\llcxq}
        ;\mapLcToDi{\llcxN}{\llcxX}]
 }
\Size{\vlread}$
$<
\vlstore
 {\vlsbr[\mapLcToDi{\llcxA{\llcxM'}{\llcxM''}}{\llcxo}
         ;\mapLcToDi{\llcxN}
                     {\llcxX}]}
\Size{\vlread}$.
\end{document}